\documentclass[letterpaper,12pt]{article}

\usepackage{setspace}
\usepackage[utf8]{inputenc}
\usepackage[margin=1in, includefoot, footskip=0.5in]{geometry}
\usepackage{natbib}

\usepackage{fvextra}

\usepackage{amsmath, amssymb, amsthm}
\usepackage{graphicx}
\usepackage[dvipsnames]{xcolor}
\usepackage{tikz}
\usepackage{bm}
\usepackage{proof}
\usepackage{amsfonts}
\usepackage{tcolorbox}
\usepackage{ascmac}
\usepackage{fancybox}
\usepackage[all]{xy}
\usepackage{varwidth}
\usepackage{etoolbox}
\usepackage{siunitx}
\usepackage{physics}
\usepackage{float}
\usepackage{bbm}
\usepackage{threeparttable}
\usepackage{rotating} 

\usepackage{graphicx} 

\usepackage{booktabs}

\usepackage{float}
\AtBeginDocument{\RenewCommandCopy\qty\SI}  

\usepackage{apalike}
\usepackage{url}
\usepackage{hyperref}
\hypersetup{
  colorlinks=true,
  linkcolor=blue,
  filecolor=magenta,
  urlcolor=cyan,
  pdftitle={Overleaf Example},
  pdfpagemode=FullScreen,
}
\hypersetup{
    colorlinks,
    citecolor=black,
    filecolor=black,
    linkcolor=black,
    urlcolor=black
}

\usepackage{enumerate}
\usepackage[colorinlistoftodos]{todonotes}  

\usepackage[edges]{forest}
\usepackage{tree}
\usepackage{ulem}

\DeclareMathOperator{\Cov}{Cov}

\newcommand{\Reg}{\mathcal{R}}

\newtheorem{theorem}{Theorem}[section]
\newtheorem{proposition}{Proposition}[section]

\newtheorem{assumption}{Assumption}[section]

\theoremstyle{definition}
\AtEndEnvironment{remark}{\qed}
\newtheorem{remark}{Remark}[section]

\DeclareMathOperator*{\argmin}{arg\,min}

\renewcommand{\emph}[1]{#1}

\newcommand{\indep}{\perp \!\!\! \perp}
\newcommand{\att}{{\theta_{\text{ATT}}}}

\newcommand{\ATT}{{\Theta_{\text{ATT}}}}
\newcommand{\attdid}{{\theta_{\text{ATT}}^{\text{DID}}}}
\newcommand{\attm}{{\theta_{\text{ATT}}^{\text{M}}}}
\newcommand{\attdidm}{{\theta_{\text{ATT}}^{\text{DIDM}}}}

\newcommand{\E}{{\text{E}}}

\begin{document}
\allowdisplaybreaks

\title{\LARGE 
Choosing A Headline Estimand from Matching, DID, and Hybrid Designs: A Minimax-Regret Approach%
\thanks{\setlength{\baselineskip}{4mm} 
We are grateful to Isaiah Andrews and Jesse Shapiro for repeated feedback and guidance for this project. We also thank Jeff Smith and Petra Todd for repeated guidance on replicating their papers. We benefitted from useful comments by Chris Campos, Raj Chetty, Jamie Fogel, Hidehiko Ichimura, Sho Miyaji, Jon Roth, Matt Staiger, Davide Viviano, and Oscar Volpe. Angela Arora, Yannick Feussi, and Stefan Nicov provided excellent research assistance. All remaining errors are ours. Sasaki gratefully acknowledges the generous research support provided by Brian and Charlotte Grove.\smallskip}}
\author{
Yechan Park\thanks{\setlength{\baselineskip}{4mm}Department of Economics, Harvard University Email: \texttt{yechanpark@fas.harvard.edu}\smallskip}
\and
Yuya Sasaki\thanks{\setlength{\baselineskip}{4mm}Brian and Charlotte Grove Chair and Professor of Economics. Department of Economics, Vanderbilt University, VU Station B \#351819, 2301 Vanderbilt Place, Nashville, TN 37235-1819 Email: \texttt{yuya.sasaki@vanderbilt.edu}}
}
\date{}
\maketitle
\begin{abstract}\setlength{\baselineskip}{6mm}


Researchers using panel data to estimate causal effects routinely choose among three approaches to using past outcomes: difference-in-differences (DID), conditioning on lagged outcomes (matching, M), and a hybrid that does both (DIDM). The corresponding identifying assumptions are non-nested, leaving little guidance on which to report. We give conditions under which the corresponding estimands are ordered, with DIDM bracketed between matching and DID. This makes DIDM the minimax-regret choice among the three under a broad class of loss functions. We recommend reporting DIDM as the headline estimate, with matching and DID as bounds. We illustrate in applications.


\end{abstract}

\newpage

\section{Introduction}
\label{sec:introduction}

Many policy evaluations in economics rely on non-experimental panel data. Large policy changes, such as mass layoffs, minimum wage reforms, the rollout of job-training programs, and major education policies, are typically not assigned at random. Instead, researchers observe repeated outcomes for the same units and exploit quasi-experimental variation in policy timing or exposure.
The growing availability of large-scale and administrative data, together with a broader shift toward identification-focused empirical research, has made non-experimental panel and repeated-cross-section settings increasingly common in applied economics \citep{currie2020technology,goldsmithpinkham2024tracking}.

When panel data are available, a natural and widely used way to address
confounding is to exploit the lagged outcome itself.\label{sec:lagged_outcome} In
job-training applications, recent earnings histories strongly predict both participation
and future untreated earnings \citep{HeckmanSmith1999PreProgDip}; in education settings, prior test
scores play the same role \citep{chetty2014measuring1}. 
The lagged outcome can be used in
at least three distinct ways. A first approach is difference-in-differences (DID), which
compares changes rather than levels \citep{lalonde1986evaluating,angrist2009mostly,bertrand2004did}.
A second approach conditions on lagged outcomes, for example via matching or flexible
regression adjustment, so that identification rests on selection on lagged outcomes. We
refer to this class as matching-type (M) estimands
\citep{dehejia1999causal,dehejia2002propensity,heckman1997matching}.
A third, hybrid strategy combines the two ideas: it conditions on lagged outcomes and then
differences over time, yielding a difference-in-differences-matching (DIDM) estimand
\citep{heckman1998characterizing,abadie2005semiparametric,smith2005does,chetty2026creating}.\footnote{We use the term \emph{difference-in-differences matching} (DIDM) following \citet{heckman1997matching, smith2005does}; recent applied work uses cognate labels, e.g., \citet{chetty2026creating} term their design a ``matched difference-in-differences.''} The three designs
thus differ only in how they use the lagged outcome: DID differences it out, M conditions
on its level, and DIDM does both. The identifying assumptions associated with these three
designs are mutually non-nested restrictions on the joint distribution of potential
outcomes and treatment.\footnote{In particular, the assumptions under which DID is valid do
not imply those under which M is valid (and vice versa), and the assumptions for DIDM do
not reduce to either DID or M. We illustrate this fact in Appendix~\ref{sec:non_nested}.}


These three designs account for a substantial share of published work using panel or repeated cross-section data to estimate causal effects of time-varying treatments. To gauge how common they are, we conducted a census of \textit{American Economic Review} articles from 2020--2024 that use panel or repeated cross-section data.\footnote{We provide a detailed description of our methodology in Appendix~\ref{app:aer_census}.} We find 77 studies employing panel-data identification strategies, and more than 80\% use at least one of DID, M, or DIDM. About 70\% implement some form of DID, roughly a third use matching or flexible conditioning on lagged outcomes, and just over 10\% employ a hybrid DIDM design. More than a quarter of the in-scope papers employ more than one of these approaches within the same study.\footnote{The remaining papers primarily rely on other non-experimental strategies, such as regression adjustment under selection-on-covariates, dynamic panel instrumental variables, or synthetic control methods.} 

There is, however, little formal guidance on how to choose among these alternative designs when experimental benchmarks are unavailable. In principle, applied researchers should select the estimand whose identifying assumptions are most credible in the application at hand. In practice, economic theory rarely delivers a single preferred specification. Dynamic models emphasize persistence and selection, making it natural to incorporate information on lagged outcomes, but they are typically too coarse to determine precisely how such information should enter the empirical specification.\footnote{A large literature on training programs, job search, and human capital emphasizes that both transitory shocks and persistent differences in ability shape participation decisions and earnings dynamics; see classic discussions of the ``Ashenfelter dip'' in training evaluations \citep{ashenfelter1978estimating} and more recent work on dynamic selection and labor market histories \citep[e.g.,][]{HeckmanSmith1999PreProgDip,McCallSmithWunsch2016GovVocEd}.}
Reflecting this ambiguity, closely related empirical settings often make different choices among M, DID, and DIDM. Among job-displacement, minimum-wage, and value-added papers using similar administrative data and institutional environments, some studies rely on fixed-effects DID or event-study designs, while others explicitly condition on rich pre-treatment outcome histories or embed propensity-score matching and reweighting within DID-style estimators \citep[e.g.,][]{jacobson1993earnings,cengiz2019effect,rothstein2010teacher}.\footnote{For displaced workers, \citet{jacobson1993earnings} estimate long-run earnings losses using worker fixed-effects event-study DID in unemployment-insurance records, while \citet{couch2010earnings,hyslop2017impacts,illing2024gender,lachowska2020sources} augment DID-type designs with propensity-score matching, reweighting, or controls for averages of pre-displacement earnings and hours. For minimum wages, health, and value-added, \citet{cengiz2019effect} implement standard fixed-effects event-study DID, whereas \citet{kaminska2015effects,hafner2022minimum,lenhart2017uk,lenhart2017oecd,rothstein2010teacher,chetty2014measuring1,chetty2014measuring2,angrist2017leveraging,angrist2022methods} use lagged-dependent-variable or matching-type estimators, propensity-score DID designs (whose scores often include lagged outcomes), or hybrids that combine lagged outcomes with gains-style differencing.} Even within a single paper, researchers often report multiple specifications, such as fixed-effects DID, lagged-outcome or matching estimators, and hybrid DIDM designs, and then compare them informally \citep[e.g.,][]{couch2010earnings,illing2024gender}. We read this pattern as evidence that multiple strategies are considered credible for the same setting, and that a clear criterion for choosing among them is lacking.\footnote{For instance, \citet{couch2010earnings} directly compare fixed-effects, random-growth, ATT, and ``differenced ATT'' estimators for displaced workers; \citet{hafner2022minimum} present both fixed-effects DID and propensity-score DIDM estimators that explicitly match on lagged self-rated health before the German minimum wage reform; and \citet{kaminska2015effects,hyslop2017impacts,illing2024gender} juxtapose regression-adjusted DID with matching- or reweighting-based DID in related administrative settings.}

In this paper, we provide formal criteria for choosing among M, DID, and DIDM in such environments. We consider a researcher who must commit to a single ``headline'' estimate but is uncertain about which of the three identifying assumptions (unconditional parallel trends for DID, selection on lagged outcomes for M, or conditional parallel trends for DIDM) is closest to the truth. 
Our main result shows that, under two economically interpretable conditions, the hybrid DIDM estimand is minimax-regret optimal among M, DID, and DIDM: it incurs the smallest worst-case loss, across the three possible identifying assumptions, relative to the estimand a researcher would have chosen had she known which assumption was correct. The two conditions are negative selection into treatment and stable (non-explosive) untreated outcome dynamics, which are common in labor and public economics applications. For a researcher seeking to limit the largest possible misspecification error without insisting that any one assumption holds exactly, committing to DIDM minimizes the worst-case deviation from the true average treatment effect on the treated.

These choices can matter substantively. In our empirical analysis in Section~\ref{sec:empirical_double_bracketing}, using canonical job-training settings such as NSW and JTPA \citep{lalonde1986evaluating,dehejia2002propensity,smith2005does,heckman1998characterizing}, as well as the education application in \citet{athey2025combining}, we show that switching among M, DID, and DIDM can materially change the estimated effects and, in some cases, even reverse the sign of the point estimates. 
A similar issue arises in recent work on banking deregulation. In a critique of \citet{boissel2022dividend}, \citet{bach2023dividend} identify the choice between DID and DIDM as one of two central points of contention: they report that, holding the data and most specification details fixed, replacing the hybrid DIDM design with a standard DID design renders the estimated positive treatment effect statistically insignificant.

Our minimax result is obtained through an intermediate analytical step that we believe is itself informative for applied work. 
We show that, under (i) negative selection into treatment (so that treated units would, on average, have had lower untreated outcomes than controls) and (ii) stable, non-explosive untreated outcome dynamics, the population estimands satisfy the same ordering:
M $\le$ DIDM $\le$ DID.
This result generalizes the insight in \citet{angrist2009mostly} that the DID and lagged-dependent-variable estimands lie on opposite sides of the true effect, so that the truth is bracketed between them, by showing that the hybrid DIDM estimand lies systematically between the two endpoints. Across multiple program-evaluation settings spanning job training and educational interventions, and using four benchmark datasets \citep{lalonde1986evaluating,heckman1998characterizing,smith2005does,chetty2014measuring1,athey2025combining}, we observe an empirical pattern consistent with our theory: matching-based estimates tend to be lower, DID estimates higher, and hybrid DIDM estimates lie in between.

To complement the estimand-level minimax-regret result, we consider a calibrated Monte Carlo design based on the NSW data from \citet{lalonde1986evaluating}. The decision-theoretic question requires comparing the three candidate procedures across data-generating processes under which different identifying assumptions hold. We therefore construct three such environments, one each favoring M, DIDM, and DID, designed to remain observationally similar in the sense of matching key moments and cross-moment relationships in the data. This makes it plausibly difficult for a researcher to know which design is preferred from the observed data alone. The resulting \(3\times 3\) regret matrix then provides a compact summary of how each estimator performs across the three environments, and makes the minimax logic tangible: while no single procedure is pointwise best in every world, DIDM minimizes worst-case regret, that is, the largest amount by which it underperforms the best procedure for the realized world.

{\color{black}The framework has two main implications for applied work. First, it yields a principled default choice among common panel-data designs. When researchers must report a single headline estimate (for policy communication, executive summaries, or meta-analysis), the hybrid DIDM design provides a natural default because it minimizes worst-case regret across the three leading approaches. 
Second, it offers a structured way to interpret differences when multiple estimates are reported.
}

\subsection{Relation to the Literature}\label{sec:literature}

This paper relates to three strands of literature.

First, it relates to the classical literature on nonexperimental evaluation following \citet{lalonde1986evaluating}. Foundational contributions such as \citet{heckman1998characterizing,heckman1998_2matching}, \citet{dehejia1999causal,dehejia2002propensity}, and \citet{smith2005does} study which nonexperimental methods best replicate experimental benchmarks in job-training settings. That literature compares specifications that, in our language, map naturally into M, DID, and DIDM-type estimands. Its main organizing question is typically which estimator has the smallest bias, often measured in absolute value, relative to the experimental ATT. Our paper asks a different question: how a researcher should choose among these competing observational estimands when the underlying identifying assumptions are mutually non-nested and no benchmark is available. The emphasis therefore shifts from ex post estimator comparison to ex ante design choice under model uncertainty.

This paper is also related to work such as \citet{chabe2017should,daw2018matching}, which studies matching- and DID-based estimators in specific simulated or parametric environments. In particular, \citet{chabe2017should} analyzes DID combined with conditioning on pre-treatment outcomes in a model with permanent and transitory confounders and in simulations calibrated to job-training settings, while \citet{daw2018matching} uses Monte Carlo simulations to study regression-to-the-mean bias from matching on pre-period variables in DID designs. Relative to that literature, our contribution is twofold. First, our main comparative result is analytical and nonparametric.
Second, whereas that literature studies performance within particular simulated environments, our contribution is decision-theoretic: we ask which headline estimand is safest when the researcher is uncertain which identifying restriction is closest to the truth. Under this model uncertainty, DIDM is minimax-regret optimal among the three leading panel-data estimands under a broad class of loss functions.

Second, our theory builds on the literature on bracketing between lagged-outcome and fixed-effects or DID estimands. In linear panel models, \citet[][Section 5]{angrist2009mostly} show that, in linear panel models, the lagged-dependent-variable and fixed-effects estimands bound the true effect from opposite sides, so that the truth lies between them. \citet{ding2019bracketing} extend that insight to a nonparametric framework. Our paper contributes to this literature by introducing a third object, the hybrid DIDM estimand, and showing that under negative selection and stable untreated dynamics,
\(
\theta^M_{ATT} \;\le\; \theta^{DIDM}_{ATT} \;\le\; \theta^{DID}_{ATT}.
\)
DIDM reduces to M in the special case $s=0$, where the matching variable coincides with the differencing baseline ($Y_{-s}=Y_0$) and the $Y_0$ terms cancel. Our framework thus nests the familiar LDV-versus-FE bracketing as the case in which the middle object coincides with one endpoint, while extending the logic to the matched-DID designs (DIDM) that are common in practice.

Finally, the paper relates to the modern literature on event studies, DID, and panel matching. Recent surveys such as \citet{roth2023review} and \citet{deChaisemartin2023two} emphasize both the centrality of parallel-trends assumptions and the unresolved role of lagged outcomes and matching-type adjustments in DID practice. A related empirical literature uses lagged outcomes and other pre-treatment histories either to construct M-type estimands, as in \citet{acemoglu2019democracy}, or to construct hybrid DIDM-type designs, as in \citet{deChaisemartin2020two}, \citet{dube2023local}, and \citet{imai2023matching}. Our contribution to this literature is to place M, DIDM, and DID in a common framework, characterize when they are systematically ordered, and show how that ordering should guide design choice and interpretation in panel-data applications.


\section{Setup and Estimands}\label{sec:setup}\label{sec:simple}

Let $W$ denote treatment-group status. Units with $W=1$ receive treatment between periods $t=0$ and $t=1$, whereas units with $W=0$ remain untreated throughout.\footnote{For expositional clarity, this section focuses on a two-group, single-treatment-timing setup. Appendix~\ref{sec:general} extends the framework to more general matching variables and to event-study or staggered-adoption settings by redefining $(\widetilde Y_0,\widetilde Y_1,W,X)$ appropriately; see especially Section~\ref{sec:examples} and Sections~\ref{sec:event_did}--\ref{sec:event_didm}. In such settings, cohort- and horizon-specific effects can be written as analogous DID- or DIDM-style contrasts and, when desired, aggregated by averaging the relevant comparisons across groups and horizons.} Hence, all units are untreated for $t \leq 0$, and only treated units are exposed to treatment for $t \geq 1$. 
Let $Y_t(w)$ denote the potential outcome at time $t$ under treatment status $w \in \{0,1\}$. The parameter of interest is the average treatment effect on the treated (ATT) at $t=1$, defined as
\[
\att \equiv E\!\left[ Y_1(1) - Y_1(0) \mid W = 1 \right].
\]

We write $Y_t$ for the observed outcome at time $t$. Throughout, we focus on the lagged untreated outcome $Y_{-s} = Y_{-s}(0)$, where $s \geq 0$, as the key matching variable. This choice reflects the emphasis in the evaluation literature on lagged outcomes as particularly informative predictors of both selection into treatment and the dynamics of untreated outcomes.\footnote{
Lagged outcomes play a central role in the empirical literatures motivating this paper. In job-training applications, recent earnings histories are highly predictive of both participation in treatment and future untreated earnings. In education settings, prior test scores play an analogous role. 
Our framework accommodates an arbitrary lag order $s \geq 0$. When $s=0$, DIDM reduces to M, thereby nesting the classical LDV-versus-FE bracketing logic of \citet{angrist2009mostly} as a special case. When $s>0$, it encompasses the matched difference-in-differences strategies of \citet{heckman1998characterizing} that condition on a lagged outcome measured $s>0$ periods before treatment, sometimes termed symmetric difference-in-differences.}
Section~\ref{sec:general} extends the analysis to a general vector-valued matching variable $X$.

\subsection{Three Candidate Estimands}

Following \citet{heckman1998characterizing}, consider the following three observational estimands for $\att$:
\begin{align*}
\attm 
&\equiv E[Y_1\mid W=1]
    - E\!\left[ E[Y_1\mid W=0, Y_{-s}] \mid W=1 \right], \\
\attdid
&\equiv E[Y_1-Y_0\mid W=1]
    - E[Y_1-Y_0\mid W=0], 
\qquad\text{and}\\
\attdidm
&\equiv E\!\left[
      E[Y_1-Y_0\mid Y_{-s},W=1]
      -E[Y_1-Y_0\mid Y_{-s},W=0]
      \,\middle|\, W=1
    \right].
\end{align*}
We refer to $\attm$, $\attdid$, and $\attdidm$ as the \emph{matching} (M), \emph{difference-in-differences} (DID), and \emph{difference-in-differences matching} (DIDM) estimands, respectively.

The three estimands differ only in how untreated outcomes are used to construct the missing counterfactual for treated units. The M estimand adjusts solely for selection on lagged outcomes, the DID estimand accounts only for average untreated outcome growth, and the DIDM estimand combines both approaches by conditioning DID-style comparisons on lagged outcomes.

\subsection{Identification Conditions}

Each estimand identifies $\att$ under a distinct restriction on untreated potential outcomes.

First, the M estimand identifies $\att$ if treatment assignment is conditionally independent of period-$1$ potential outcomes given $Y_{-s}$:
\[
\text{Condition M:}\qquad
(Y_1(1),Y_1(0)) \indep W \mid Y_{-s}.
\]
Second, the DID estimand identifies $\att$ under unconditional parallel trends:
\[
\text{Condition DID:}\qquad
E[Y_1(0)-Y_0(0)\mid W=1]
=
E[Y_1(0)-Y_0(0)\mid W=0].
\]
Third, the DIDM estimand identifies $\att$ under conditional parallel trends:
\[
\text{Condition DIDM:}\qquad
E[Y_1(0)-Y_0(0)\mid Y_{-s},W=1]
=
E[Y_1(0)-Y_0(0)\mid Y_{-s},W=0].
\]

\paragraph{Mutual Non-Nestedness of the M, DID, and DIDM Conditions:}
The above three restrictions corresponding to M, DID, and DIDM are distinct and mutually non-nested. Formal arguments are provided in Appendix~\ref{sec:detailed_calc_nonnest}. Here, we provide the intuition.

Condition M imposes a restriction on untreated \emph{levels} conditional on lagged outcomes, whereas DID and DIDM impose restrictions on untreated \emph{growth rates}. Accordingly, M may hold even when DIDM fails if treated and control units with the same $Y_{-s}$ share the same untreated outcome level at $t=1$ but exhibit different untreated growth between $t=0$ and $t=1$. Conversely, DIDM may hold while M fails if untreated growth is the same conditional on $Y_{-s}$, but untreated levels differ systematically across treatment status.

Likewise, DIDM and DID differ because DIDM imposes a \emph{conditional} parallel-trends restriction, whereas DID imposes an \emph{unconditional} one. DIDM may hold even when DID fails due to compositional differences across $Y_{-s}$ strata, while DID may hold even when DIDM fails if mutually offsetting conditional trend differences cancel out in the aggregate.

Thus, these assumptions are not ordered along a single robustness dimension, as they rule out different features of the data-generating process. In practice, a researcher \textit{ex ante} does not know which restriction is most credible. The three estimands need not coincide, and each may be biased when its identifying condition fails. 
This gives rise to the design problem studied in this paper: when a researcher must select a single headline observational estimand from $\{\attm, \attdidm, \attdid\}$, is there a principled default choice?

To address this question, the remainder of the paper proceeds in the following steps.
Section~\ref{sec:typology} documents that M, DID, and DIDM are recurring headline designs in applied work. Section~\ref{sec:double_bracketing} then develops a nonparametric proposition giving conditions under which $\attm \leq \attdidm \leq \attdid$, and documents (Section~\ref{sec:empirical_double_bracketing}) that benchmark applications exhibit this ordering in practice. Section~\ref{sec:minimax} demonstrates that the ordering implies a minimax-regret rationale for DIDM as the default headline estimand.

\section{Matching, DID, and DIDM in Applied Work}
\label{sec:typology}

Across several leading applied literatures, researchers addressing closely related causal questions employ empirical designs that map naturally into
$
\theta_{ATT}^M,\, \theta_{ATT}^{DIDM},\, \text{ and } \theta_{ATT}^{DID}.
$
Table~\ref{tab:typology_main} summarizes canonical examples from job training, displaced workers, minimum wages, and teacher/school value-added.
Our goal here is descriptive and selective: we use the table to show that the choice among M, DID, and DIDM recurs in applied work.

\begin{table}[t]
\centering
\caption{Illustrative Designs and Their Classification into M, DID, and DIDM}
\label{tab:typology_main}
\scalebox{0.70}{
\begin{tabular}{llp{4cm}p{4.5cm}c}
\toprule
Domain & Paper & Target effect $\theta_i$ & Design (paper's description) & Our type \\
\midrule
\multicolumn{5}{l}{\textit{(A) Job training}} \\
\midrule
Job training &
\begin{tabular}[t]{@{}l@{}}
LaLonde (1986);\\ Heckman et al.\ (1998a,b);\\
Dehejia--Wahba (1999,2002);\\
Smith--Todd (2005)
\end{tabular} &
Experimental benchmark vs.\ observational earnings effects &
Matching, gains, and ``symmetric DID'' comparisons on NSW/JTPA-style data &
M / DID / DIDM \\[0.3em]

\midrule
\multicolumn{5}{l}{\textit{(B) Displaced workers}} \\
\midrule
Displacement & Jacobson--LaLonde--Sullivan (1993) &
Long-run earnings losses after displacement &
Worker FE event-study DID &
DID \\[0.3em]

Displacement & Couch--Placzek (2010) &
Long-run earnings losses &
Propensity-score ATT; matched DID (DATT) &
M \& DIDM \\[0.3em]

Displacement & Lachowska--Mas--Woodbury (2020) &
Long-run earnings/hours/wage losses &
Event-study DID with flexible controls for pre-displacement outcomes &
DIDM \\[0.3em]

Displacement & Hyslop--Townsend (2019) &
Post-displacement earnings/income losses &
Regression-adjusted DID and matching reported side by side &
DID \& M \\[0.3em]

Displacement & Schmieder et al.\ (2023) &
Long-run earnings/wage losses over the business cycle &
Propensity-score matched event-study DID (``matched difference-in-differences'') on German administrative data &
M / DID / DIDM \\
\midrule
\multicolumn{5}{l}{\textit{(C) Minimum wages and low-wage labor markets}} \\
\midrule
Minimum wage & Cengiz et al.\ (2019) &
Employment effects of minimum-wage hikes &
Stacked event-study DID &
DID \\[0.3em]

Minimum wage & Kami\'nska--Lewandowski (2015) &
Employment effects of minimum-wage increases &
Propensity-score matching plus DID &
DIDM \\[0.3em]


Minimum wage/health & Hafner--Lochner (2022) &
Health effects of minimum wages &
Matching on pre-reform characteristics plus DID &
M \& DIDM \\
\midrule
\multicolumn{5}{l}{\textit{(D) Teacher and school value-added}} \\
\midrule
VAM & Rothstein (2010) &
Teacher value-added on test scores &
Lagged-score VAM, gains VAM, and hybrid variants &
M / DID / DIDM \\[0.3em]

VAM & Kane--Staiger (2008) &
Teacher effectiveness on achievement &
Levels-with-lagged-scores versus gains specifications &
M \& DID \\[0.3em]

VAM & Chetty et al.\ (2014a,b) &
Teacher VA and adult outcomes &
Rich lagged-score VA models with quasi-experimental validation &
M \\[0.3em]

\bottomrule\\
\end{tabular}
}
\end{table}

We classify a design as M when identification relies primarily on conditioning, matching, or reweighting based on lagged outcomes or other rich pre-treatment levels. We classify a design as DID when identification relies on a parallel-trends-type restriction implemented through first differences, fixed effects, or event-study specifications, without explicit conditioning on lagged outcomes. Finally, we classify a design as DIDM when it combines both elements, for example, by applying DID within a matched or reweighted sample, or by estimating a trends-based specification after explicitly conditioning on lagged outcomes.
Here and throughout, we use M in a broad sense to include not only literal matching or reweighting estimators, but also lagged-outcome regression adjustments. While these specifications are not matching estimators in the narrow algorithmic sense, they share the same identifying logic: conditioning on pre-treatment outcomes so that treated units are compared to control units with similar outcome histories; see, for example, \citet{heckman1997matching} and the balancing/weighting synthesis in \citet{doudchenko2016balancing}.


The choice among M, DID, and DIDM has been central to program evaluation since its earliest non-experimental benchmarking exercises. Beginning with \citet{lalonde1986evaluating} and continuing through \citet{heckman1998characterizing,heckman1998_2matching,dehejia1999causal,dehejia2002propensity,smith2005does}, the job-training literature already compares all three designs against experimental benchmarks, making it a natural starting point for our framework.

The same design choice reappears in later applied work, which suggests that the question studied in this paper is relevant beyond the classical job-training context. In displaced-worker studies, canonical event-study specifications such as \citet{jacobson1993earnings} employ DID, while later work incorporates matching and matched-DID hybrids; see, for example, \citet{couch2010earnings}, \citet{hysloptownsend2019longer}, \citet{lachowska2020sources}, and \citet{schmieder2023costs}. In minimum-wage applications, stacked event-study designs such as \citet{cengiz2019effect} provide a canonical DID benchmark, whereas studies such as \citet{kaminska2015effects}, \citet{hafner2022minimum}, and \citet{arranz2025assessing} combine matching or reweighting with DID-type comparisons. 
Finally, in teacher and school value-added research, lagged-score models map naturally into M, gains models into DID, and hybrid gains-plus-lagged-score specifications into DIDM.\footnote{In the value-added literature, empirical-Bayes or other shrinkage adjustments are conceptually distinct from the design taxonomy used here. Our classification concerns the identifying structure of the underlying causal signal--whether it is constructed from lagged-score conditioning (M), gains-style differencing (DID), or both (DIDM)--prior to any post-estimation shrinkage. Accordingly, the comparison in this paper is conducted in estimand space and speaks to identification bias across designs, rather than to a full finite-sample risk or MSE ranking that incorporates variance.} See, among others, \citet{kane2008estimating}, \citet{rivkin2005teachers}, \citet{rothstein2010teacher}, \citet{chetty2014measuring1,chetty2014measuring2}, and \citet{angrist2023methods}.

The practical implication is that applied researchers repeatedly face the same design choice, often while targeting closely related causal parameters. This motivates a principled comparison among the three estimands. The next section establishes the ordering $\attm \leq \attdidm \leq \attdid$ first theoretically and then empirically: a nonparametric proposition (Section~\ref{sec:double_bracketing}) gives conditions under which it holds, and four benchmark applications (Section~\ref{sec:empirical_double_bracketing}) show that it arises in practice. Section~\ref{sec:minimax} then shows that this ordering is what makes DIDM the minimax-regret choice among the three.

\section{The Double Bracketing}\label{sec:double_bracketing}

The goal of this section is to provide a formal condition under which the ordering
$
\attm \leq \attdidm \leq \attdid
$
holds. This ordering serves as the key input for the minimax-regret decision problem studied in Section~\ref{sec:minimax}. 
We begin with a nonparametric proposition, illustrate its intuition using a simple linear dynamic model, and then document the same pattern in benchmark applications.

\subsection{A Nonparametric Double-Bracketing Proposition}\label{sec:double_bracketing_np}

Let
\begin{align*}
\Delta(\attm)   &= \attm - \att, \\
\Delta(\attdidm) &= \attdidm - \att, \\
\Delta(\attdid)  &= \attdid - \att
\end{align*}
denote the identification errors of the three estimands relative to the target $\att$. Then, the ordering
\[
\attm \leq \attdidm \leq \attdid
\]
is equivalent to the ordering
\[
\Delta(\attm) \leq \Delta(\attdidm) \leq \Delta(\attdid).
\]

We introduce the following assumption as a condition under which this ordering holds.

\begin{assumption}[Negative Selection and Stable Untreated Dynamics]\label{a:special}
${}$
\begin{enumerate}[(i)]
\item\label{a:special:selection}
For all $y$,
\[
E[Y_0 \mid W=0, Y_{-s}=y]
\;\geq\;
E[Y_0 \mid W=1, Y_{-s}=y].
\]

\item\label{a:special:sd}
The distribution of $Y_{-s}$ among untreated units first-order stochastically dominates the distribution among treated units:
\[
F_{Y_{-s}\mid W=0}
\;\text{FOSD}\;
F_{Y_{-s}\mid W=1}.
\]

\item\label{a:special:dec}
The untreated growth function
\[
\Phi(y)
:=
E[Y_1 - Y_0 \mid W=0, Y_{-s}=y]
\]
is weakly decreasing in $y$.
\end{enumerate}
\end{assumption}

Assumption~\ref{a:special} provides a nonparametric formulation of negative selection together with stable untreated dynamics. Part~\eqref{a:special:selection} states that, conditional on lagged untreated outcomes, treated units have weakly lower period-$0$ untreated outcomes. Part~\eqref{a:special:sd} requires that untreated units are positively shifted in terms of lagged untreated outcomes. Part~\eqref{a:special:dec} stipulates that untreated growth is weakly smaller for units with higher lagged untreated outcomes.
Importantly, the three conditions in Assumption~\ref{a:special} have testable implications. In Appendix~\ref{sec:empirical_assumption}, we show that these implications are not rejected in our job-training and education applications.

\begin{proposition}[Double Bracketing]\label{lem:double_bracketing}
Suppose that Assumption~\ref{a:special} holds. Then,
\[
\Delta(\attm)
\;\leq\;
\Delta(\attdidm)
\;\leq\;
\Delta(\attdid),
\]
and therefore
\[
\attm
\;\leq\;
\attdidm
\;\leq\;
\attdid.
\]
\end{proposition}

Proposition~\ref{lem:double_bracketing} provides the main analytical input for the remainder of the paper. It delivers the ordering used in the minimax-regret argument developed below.

The proposition also yields an immediate interpretation under each candidate identifying restriction. If Condition M holds, then
\[
0 = \Delta(\attm) \;\leq\; \Delta(\attdidm) \;\leq\; \Delta(\attdid),
\]
so $\attm = \att$, while $\attdidm$ and $\attdid$ weakly exceed it. If Condition DID holds, then
\[
\Delta(\attm) \;\leq\; \Delta(\attdidm) \;\leq\; \Delta(\attdid) = 0,
\]
so $\attdid = \att$, while $\attm$ and $\attdidm$ weakly fall below it. If Condition DIDM holds, then
\[
\Delta(\attm) \;\leq\; \Delta(\attdidm) = 0 \;\leq\; \Delta(\attdid),
\]
so $\attdidm = \att$ and lies between the other two candidate estimands. In all three cases, DIDM occupies the interior position implied by the double-bracketing structure.

\begin{remark}[Reversed Selection Ordering]\label{rem:reverse_bracketing}
Proposition~\ref{lem:double_bracketing} is stated under a negative-selection ordering, motivated by the empirical literature on job-training programs and educational interventions. If instead the selection ordering is reversed, in the sense that
\[
E[Y_0 \mid W=0, Y_{-s}=y]
\;\leq\;
E[Y_0 \mid W=1, Y_{-s}=y]
\qquad \text{for all } y,
\]
and the distribution of $Y_{-s}$ among treated units first-order stochastically dominates that among untreated units, while Assumption~\ref{a:special}, part~\eqref{a:special:dec}, remains unchanged, then the same argument yields the reversed bracketing
\[
\Delta(\attm)
\;\geq\;
\Delta(\attdidm)
\;\geq\;
\Delta(\attdid),
\]
and therefore
\[
\attm
\;\geq\;
\attdidm
\;\geq\;
\attdid.
\]
Thus, reversing the selection ordering reverses the direction of the bracketing, but $\attdidm$ remains the interior estimand. Since the minimax-regret result in Section~\ref{sec:minimax} depends on DIDM occupying the middle position, rather than on the direction of the ordering, its logic remains unchanged under this reversed-selection case.
\end{remark}

\subsection{A Parametric Illustration}\label{sec:parametric}

To illustrate how Proposition~\ref{lem:double_bracketing} can arise in a familiar setting, consider the following linear dynamic model:
\begin{equation}\label{eq:parametric}
Y_{i,t}
=
\alpha
+ \beta W_i \mathbbm{1}\{t \geq 1\}
+ \gamma W_i
+ \delta_t
+ \rho Y_{i,t-1}
+ \epsilon_{i,t},
\end{equation}
with
\begin{equation}\label{eq:parametric_orthogonality}
E[\epsilon_{i,1}\mid Y_{i,-1},W_i]
=
E[\epsilon_{i,0}\mid Y_{i,-1},W_i]
=
0.
\end{equation}
Here, $\rho$ captures persistence in outcomes, $\gamma$ captures selection into treatment through baseline outcomes, and $\beta$ represents the causal effect of interest.

Assume:
\begin{align}
\text{Negative Selection I:}\quad
&\gamma
=
E[Y_{i,0}\mid W_i=1,Y_{i,-1}]
-
E[Y_{i,0}\mid W_i=0,Y_{i,-1}]
\leq 0,
\label{eq:parametric:negative_selection_1}
\\
\text{Negative Selection II:}\quad
&E[Y_{i,-1}\mid W_i=1]
\leq
E[Y_{i,-1}\mid W_i=0],
\label{eq:parametric:negative_selection_2}
\\
\text{Stable Dynamics:}\quad
&0 \leq \rho \leq 1.
\label{eq:parametric:non_explosive}
\end{align}
Under \eqref{eq:parametric}--\eqref{eq:parametric:non_explosive}, one obtains (with formal derivations found in Appendix~\ref{sec:detailed_calculations}) the gaps
\begin{align}
\attdidm-\attm
&=
-\gamma
\;\geq\;
0,
\label{eq:parametric:gap1_reordered}
\\
\attdid-\attdidm
&=
\rho(1-\rho)\Big(E[Y_{i,-1}\mid W_i=0]-E[Y_{i,-1}\mid W_i=1]\Big)
\;\geq\;
0.
\label{eq:parametric:gap2_reordered}
\end{align}
Hence, $\attm \le \attdidm \le \attdid$ holds.

The illustration also clarifies when DIDM coincides with an endpoint. If $\gamma = 0$, then $\attm = \attdidm \leq \attdid$; if $\rho \in \{0,1\}$, then $\attm \leq \attdidm = \attdid$. Outside such knife-edge cases, DIDM is generically distinct from both M and DID.

\subsection{Empirical Evidence on Double Bracketing}
\label{sec:empirical_double_bracketing}

Section~\ref{sec:typology} showed that empirical designs corresponding to
\(
\theta_{ATT}^M,\, \theta_{ATT}^{DIDM},\, \text{ and } \theta_{ATT}^{DID}
\)
appear across a wide range of applied fields. We particularly focus on four benchmark datasets: the NSW program with the CPS comparison sample, the NSW program with the PSID comparison sample, the JTPA program, and an education application based on \citet{athey2025combining} and the related value-added literature \citep{chetty2014measuring1,chetty2014measuring2}. \label{sec:nsw}\label{sec:jtpa}\label{sec:educ}

The first three datasets come from the job-training literature, the classical laboratory for comparing non-experimental estimators to experimental benchmarks; see \citet{lalonde1986evaluating}, \citet{heckman1998characterizing,heckman1998_2matching}, \citet{dehejia1999causal,dehejia2002propensity}, and \citet{smith2005does}. In these applications, the outcome variable is real earnings, and treatment corresponds to participation in a job-training program. The availability of experimental benchmarks makes it possible to assess the sign of the bias directly.

The education application serves a complementary role. In this setting, the outcome is student achievement, measured by standardized test scores, and treatment corresponds to assignment to smaller classes. Unlike the job-training benchmarks, its value does not primarily lie in comparison to an experimental benchmark. Rather, it provides both an external-domain validation of the same ordering and a high-precision environment in which the separation among $\attm$, $\attdidm$, and $\attdid$ can be clearly observed, including across subgroups defined by observed characteristics.

Across these four datasets, we document a common empirical pattern: the estimates corresponding to M, DIDM, and DID satisfy
\(
\attm \;\leq\; \attdidm \;\leq\; \attdid
\)
up to sampling uncertainty. In the job-training applications, this manifests as an ordering of \emph{signed biases} relative to the experimental benchmark: matching-type estimands tend to be comparatively conservative, DID-type estimands comparatively optimistic, and DIDM lies in between. 
In the education application, the same ordering appears directly in the estimated effects across multiple subpopulations. Detailed descriptions of the data, institutional settings, variable construction, and estimation procedures are provided in Appendix~\ref{sec:appendix:details}.

Figure~\ref{fig:jobtraining:bracket} consolidates the benchmark job-training evidence. The top panel presents the signed-bias version of the figure in \citet{chabe2017should} using the NSW and JTPA experiments. The lower two panels report a broader collection of estimates from \citet{smith2005does} for the NSW--CPS and NSW--PSID comparisons. Although the point estimates vary across specifications, the ordering of M, DIDM, and DID remains stable. For our purposes, the key point is that the same bracketing relationship repeatedly appears across benchmark job-training datasets and estimation choices.

\begin{figure}[t]
\centering
\includegraphics[width=0.62\textwidth]{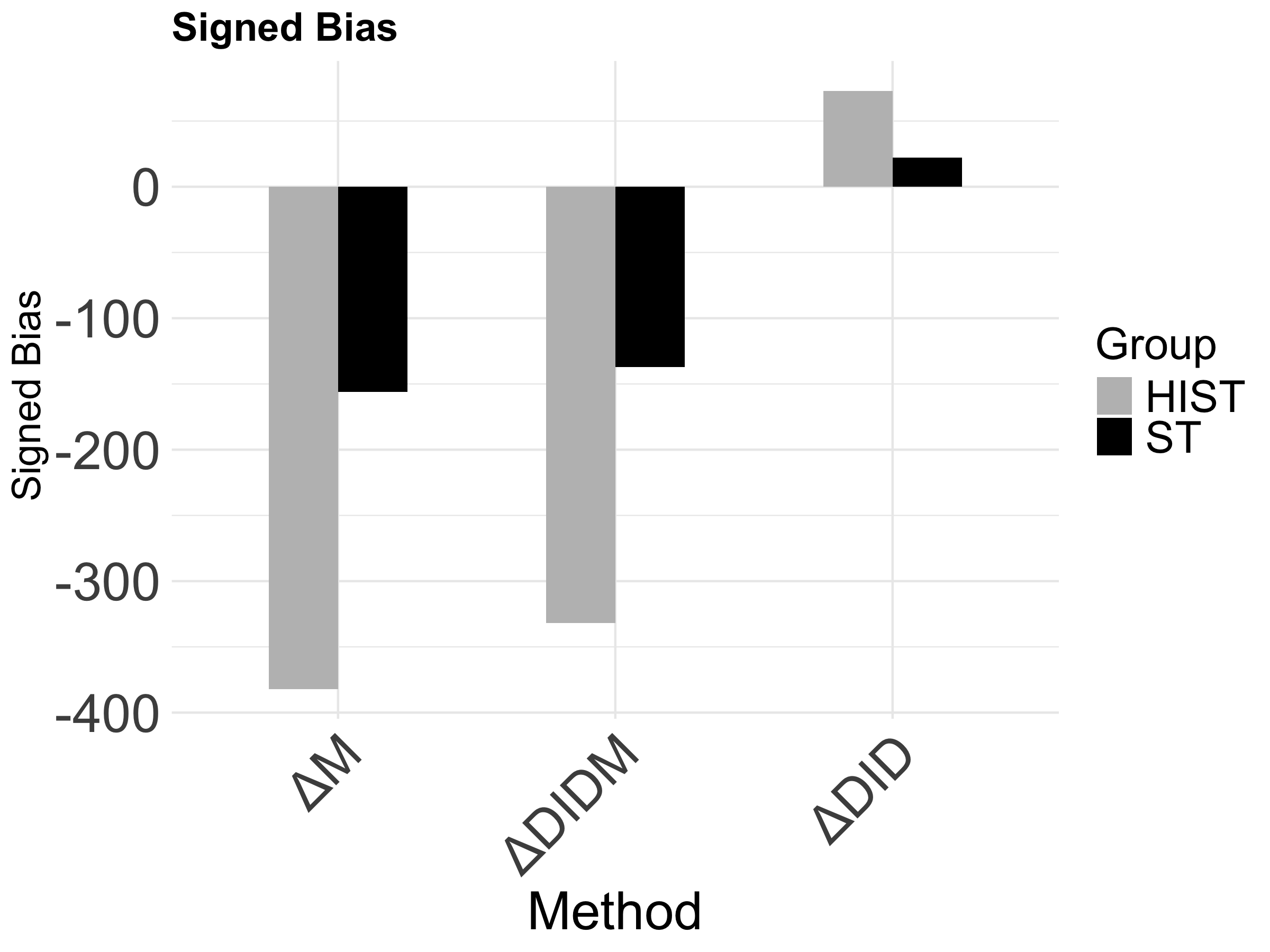}
\\[0.8em]
\begin{minipage}{0.48\textwidth}
\centering
\includegraphics[width=\textwidth]{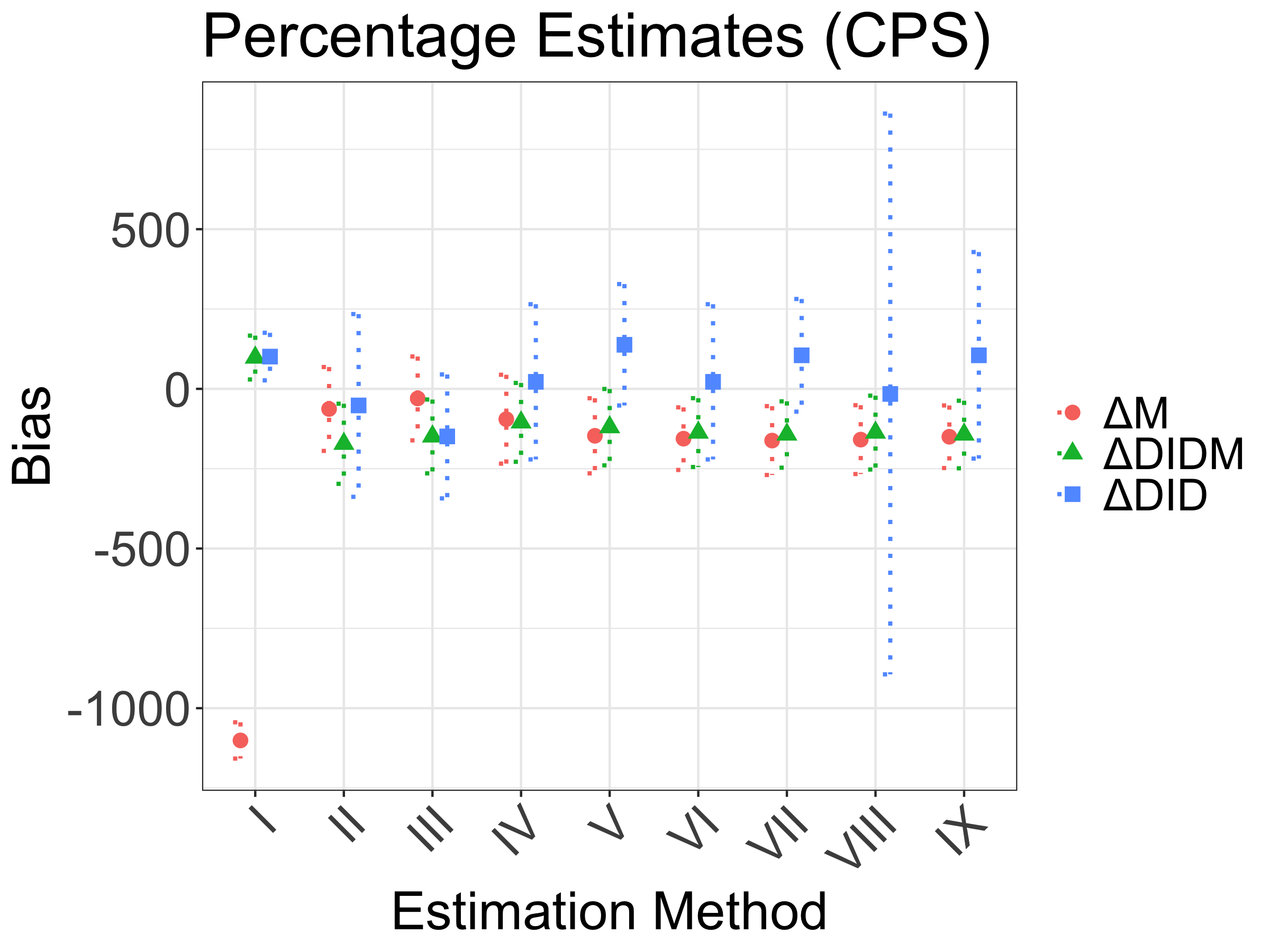}
\end{minipage}
\hfill
\begin{minipage}{0.48\textwidth}
\centering
\includegraphics[width=\textwidth]{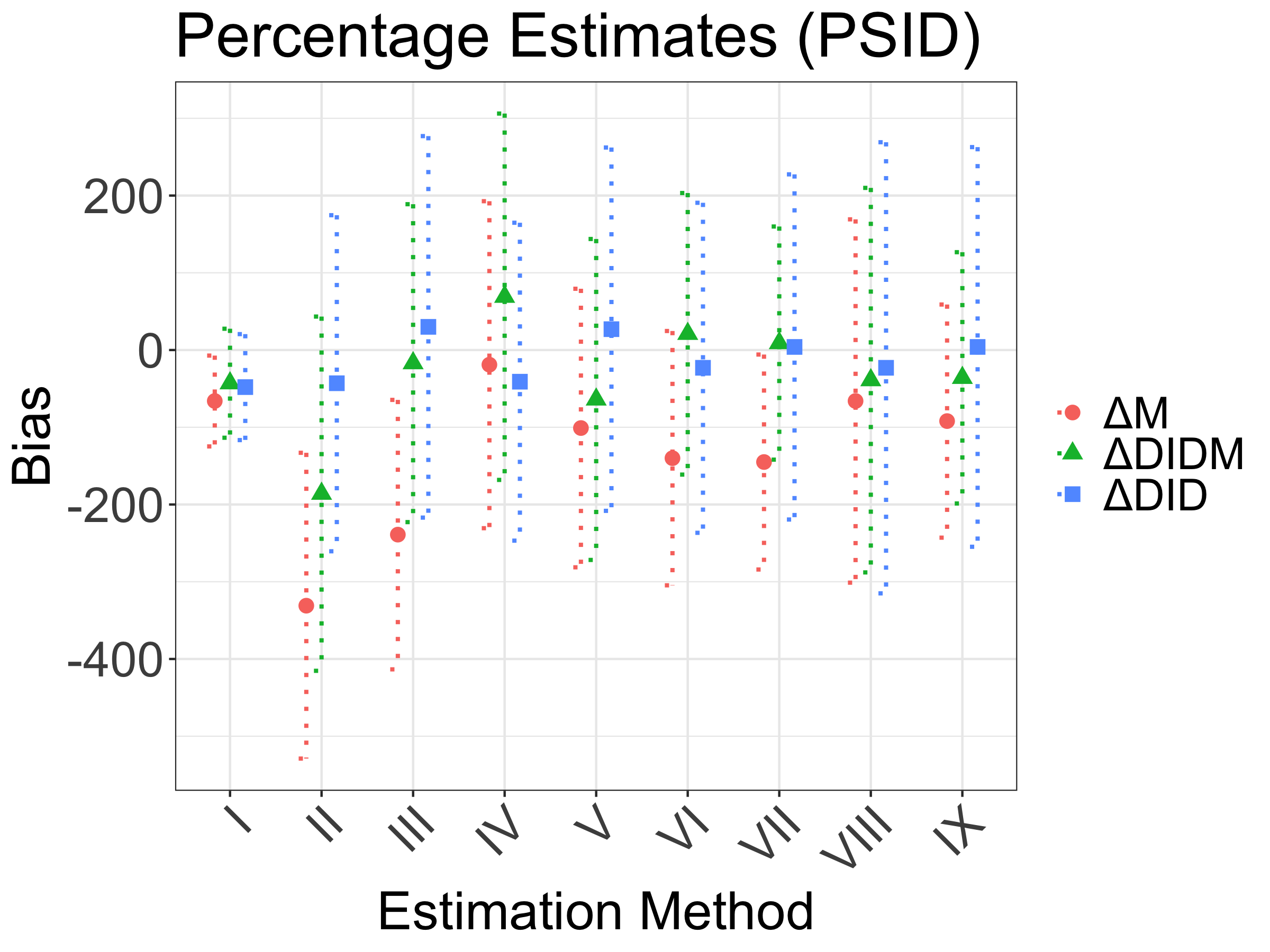}
\end{minipage}
\caption{Benchmark job-training evidence on double bracketing. The top panel summarizes the signed-bias version of the figure in \citet{chabe2017should}, based on \citet{heckman1998characterizing} (labeled HIST) and \citet{smith2005does} (labeled ST). The lower left and lower right panels report signed biases of the M, DIDM, and DID estimates across the estimation methods reported in \citet{smith2005does}, using the NSW--CPS and NSW--PSID comparison samples. 
The estimation methods are labeled as follows: I - Mean difference, II - 1 Nearest neighbor without support, III - 10 Nearest-neighbors without support, IV - 1 Nearest-neighbor with support, V - 10 Nearest-neighbors with support, VI - Local linear matching (bw = 1.0), VII - Local linear matching (bw = 4.0), VIII - Local linear regression adjusted (bw = 1.0), IX - Local linear regression adjusted (bw = 4.0).
}
\label{fig:jobtraining:bracket}
\end{figure}

Figure~\ref{fig:educ:bracket} shows the corresponding pattern in the education application of \citet{athey2025combining}. Across multiple subpopulations, the estimated effects continue to satisfy the same ordering, despite substantial variation in levels across groups.

\begin{figure}[t]
\centering
\includegraphics[width=0.7\textwidth]{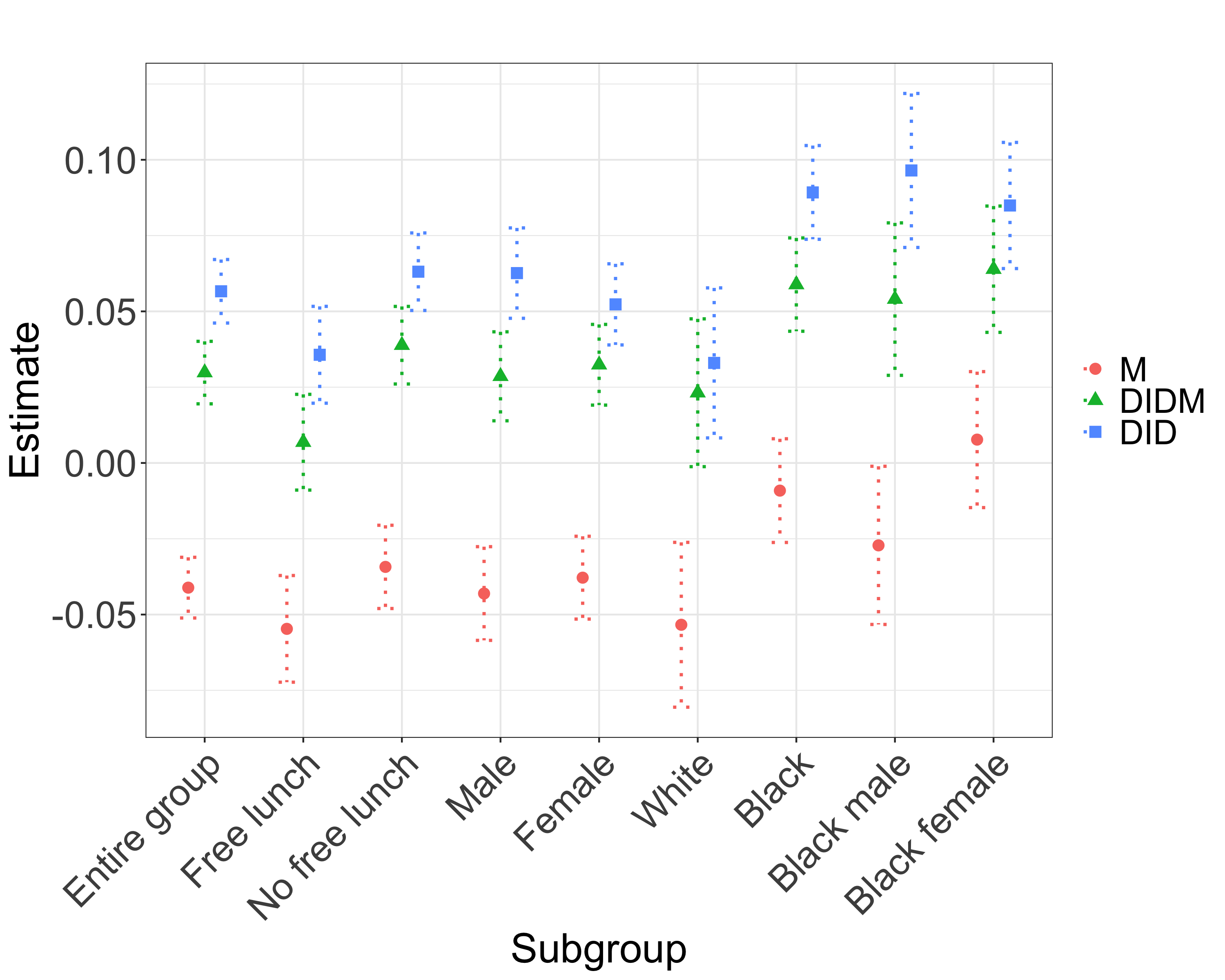}
\caption{Estimates and 95\% confidence intervals for the M, DIDM, and DID estimands in the education application, by subpopulation. The ordering $\attm \leq \attdidm \leq \attdid$ remains stable across groups defined by observed characteristics.}
\label{fig:educ:bracket}
\end{figure}

Taken together, these four benchmark datasets point to a common empirical regularity. They show that the ordering in Proposition~\ref{lem:double_bracketing} is a relationship that appears repeatedly in practice.

\section{The Minimax-Regret Choice of a Headline Estimand}\label{sec:minimax}\label{sec:minimax}

As noted earlier, applied researchers often report estimates of $\attm$, $\attdidm$, and $\attdid$ side by side. If a researcher must select a single observational target as the headline estimand, a natural question is which choice is best when the most credible identifying restriction is unclear. We answer this by interpreting “best” as safest in the minimax sense.

\subsection{The Minimax-Regret Result}

Let $\ATT = \{\attm, \attdidm, \attdid\}$, and let $L$ denote a loss function defined on $\ATT \times \ATT$. We interpret each element of $\ATT$ both as a possible \emph{action} (the estimand a researcher reports as the headline effect) and as a possible \emph{state} (the value of the true ATT), reflecting uncertainty about which identifying assumption is closest to the truth.
For actions $a \in \ATT$ and states $\theta \in \ATT$, define the regret as
\[
R(a,\theta) = L(a,\theta) - \min_{a' \in \ATT} L(a',\theta).
\]

We impose the following condition on $L$.

\begin{assumption}\label{a:loss}
$L(a,\theta) = c(\theta) + \ell(|a-\theta|)$,
where $\ell(0)=0$ and $\ell$ is nondecreasing.
\end{assumption}

This assumption accommodates a wide range of commonly used choices for $\ell$, including absolute loss, squared loss, power loss, Huber loss, $\varepsilon$-insensitive loss, exponential loss, logistic loss, Tukey's biweight loss, Cauchy loss, Welsch loss, and fair loss, among others. It requires only that the loss depends on the absolute deviation $|a - \theta|$ through a function $\ell$ that is weakly increasing in this deviation.

\begin{theorem}[Minimax-regret optimality of DIDM]\label{theorem:main}
If Assumptions \ref{a:special} and \ref{a:loss} are satisfied, then
\[
\max_{\theta \in \ATT} R(\attdidm,\theta) 
\;=\; 
\min_{a' \in \ATT} \max_{\theta \in \ATT}R(a',\theta).
\]
\end{theorem}

This theorem shows that $\attdidm$ is the minimax-regret choice among the three alternatives in $\ATT$ for a researcher whose loss $L$ satisfies Assumption~\ref{a:loss}. Hence, a researcher who seeks to minimize the worst-case regret over the three candidate identifying assumptions should adopt DIDM as the headline estimand. 
This recommendation is similar to, though distinct from, the use of midpoint estimators for interval-identified parameters \citep[e.g.,][]{song2014point} in the partial-identification literature.

\paragraph{Interpretation:}
M is safest only if one is primarily concerned about overstatement, while DID is safest only if one is primarily concerned about understatement. DIDM represents the middle option. Once regret is evaluated symmetrically in terms of the distance between the reported estimate and the true target, the middle option provides the best hedge against uncertainty regarding which identifying restriction is closest to the truth.
\
\subsection{An Illustration Calibrated to the NSW Data}
\label{sec:simulation_common_rf}

Theorem~\ref{theorem:main} is a statement about estimands. To make the regret
ranking concrete, this section reports a simple Monte Carlo illustration calibrated to the
NSW comparison samples. We construct three data-generating processes (``worlds''), one in
which each of M, DID, and DIDM is the correctly specified design, and ask which estimand a
researcher should report when she cannot tell the worlds apart.

\paragraph{The three worlds.}
Let $Y_{-s}$ denote the lagged outcome, measured in thousands of dollars, and define a
binary lagged-outcome index
\[
L=2\cdot\mathbf{1}\{Y_{-s}>\operatorname{med}(Y_{-s})\}-1\in\{-1,1\},
\]
so that $L=1$ for units with lagged outcome above the sample median and $L=-1$ otherwise.
Let $U\in\{-1,1\}$ be an unobserved binary confounder, independent of $Y_{-s}$. The common
assignment rule is
\[
\Pr(W=1\mid Y_{-s},U)=\Lambda(aL+cLU),
\qquad
\Lambda(z)=\frac{1}{1+e^{-z}}.
\]
Untreated potential outcomes are
\[
Y_0(0)=g_0(Y_{-s})+(m-q)U+\varepsilon_0,
\qquad
Y_1(0)=g_0(Y_{-s})+pL+mU+\varepsilon_1,
\]
where $g_0(y)=\alpha_0+\alpha_1 y$ is the linear regression of the period-0 untreated
outcome on $Y_{-s}$, with $(\alpha_0,\alpha_1)$ estimated from the comparison sample under
study (CPS or PSID). The implied untreated trend is
\[
\Delta(0)=Y_1(0)-Y_0(0)=pL+qU+(\varepsilon_1-\varepsilon_0).
\]
Observed post-period outcomes are $Y_1=Y_1(0)+W\tau(g)$, where $g$ indexes the twelve
subgroups defined by race, marital status, and high-school-degree status, and $\tau(g)$ is
the treatment effect for subgroup $g$, set equal to the corresponding subgroup estimate in
the NSW experimental sample. The DGP is calibrated once to the realized comparison sample;
the Monte Carlo then varies only the simulation draws, so the exercise illustrates the
population regret ranking of Theorem~\ref{theorem:main} rather than quantifying estimation
uncertainty.

The three candidate worlds use the same five coefficients $\theta=(a,c,p,q,m)$. One
coefficient is then set to zero in each world:
\[
\theta_M=(a,c,p,q,0),
\qquad
\theta_{DID}=(0,c,p,q,m),
\qquad
\theta_{DIDM}=(a,c,p,0,m).
\]
These restrictions encode the identifying logic of the construction. Setting $m=0$ removes
the hidden post-period level channel, so matching is valid after conditioning on the
lagged-outcome index $L$. Setting $a=0$ removes marginal selection on $L$, so the
unconditional trend difference cancels by symmetry. Setting $q=0$ removes the hidden trend
channel, so DIDM is valid after conditioning on $L$.

\paragraph{Calibration.}
The free coefficients $(a,c,p,q,m)$ are calibrated to match a set of reduced-form moments
of the comparison sample, chosen so that the simulated data resemble the real data on
features an applied researcher could inspect. Because the structural coefficients are not
separately observed, we anchor their magnitudes to sample quantities on the same scale.
For each world $k$, we first reconstruct an estimated untreated post-period outcome by
removing the subgroup treatment effect from the observed post-period outcome,
\[
\widehat Y_{1,k}(0)=Y_1-W\,\widehat\tau_k(g),
\qquad
\widehat\Delta_k(0)=\widehat Y_{1,k}(0)-Y_0,
\]
where $\widehat\Delta_k(0)$ is the implied untreated trend. The calibration then targets:
the standard deviations of $\widehat Y_{1,k}(0)$ and $\widehat\Delta_k(0)$, which anchor the
scales of $m$, $p$, and $q$; the treatment shares within each value of $L$, which anchor the
selection parameters $a$ and $c$; and the treated--control gaps in $\widehat Y_{1,k}(0)$ and
$\widehat\Delta_k(0)$ within each $L$ cell, which anchor the hidden level and trend channels.
Appendix~\ref{app:common_rf_details} states the exact moment vector and objective.

Table~\ref{tab:moment_calibrated_coefficients} reports the calibrated coefficients. The
unrestricted row gives the common coefficient vector before the world-specific zero
restriction is imposed. The remaining rows show the actual coefficients used in each
simulated world.

\begin{table}[t]
\centering
\scriptsize
\caption{Calibrated Coefficients in the Moment-Calibrated DGP}
\label{tab:moment_calibrated_coefficients}
\begin{threeparttable}
\begin{tabular}{llrrrrr}
\toprule
Sample & World & $a$ & $c$ & $p$ & $q$ & $m$ \\
\midrule
NSW + CPS  & unrestricted & 1.65 & 1.00 & 1.25 & 2.00 & 5.95 \\
NSW + CPS  & $M$          & 1.65 & 1.00 & 1.25 & 2.00 & 0.00 \\
NSW + CPS  & DID          & 0.00 & 1.00 & 1.25 & 2.00 & 5.95 \\
NSW + CPS  & DIDM         & 1.65 & 1.00 & 1.25 & 0.00 & 5.95 \\
\midrule
NSW + PSID & unrestricted & 1.15 & 1.00 & 1.75 & 2.00 & 6.85 \\
NSW + PSID & $M$          & 1.15 & 1.00 & 1.75 & 2.00 & 0.00 \\
NSW + PSID & DID          & 0.00 & 1.00 & 1.75 & 2.00 & 6.85 \\
NSW + PSID & DIDM         & 1.15 & 1.00 & 1.75 & 0.00 & 6.85 \\
\bottomrule
\end{tabular}
\begin{tablenotes}[flushleft]
\footnotesize
\item The three simulated worlds share a common calibrated coefficient vector and differ
only by the zero restriction defining the world.
\end{tablenotes}
\end{threeparttable}
\end{table}

\paragraph{Results.}
For each world $k$ we draw $B$ Monte Carlo samples. In draw $b$ we compute the sample analog
of each candidate estimand and its mean squared error (MSE) relative to the draw-specific
true ATT among the treated. Writing $\widehat\theta_{e,k,b}$ for the estimate of
$e\in\{\attm,\attdidm,\attdid\}$ in draw $b$ of world $k$ and $\att_{k,b}$ for the
corresponding true ATT,
\[
\widehat r_{e,k}=\frac{1}{B}\sum_{b=1}^B
\bigl(\widehat\theta_{e,k,b}-\att_{k,b}\bigr)^2
\]
is the MSE of estimand $e$ in world $k$. Regret subtracts the smallest MSE in the same
world, and worst-case regret takes the maximum over worlds:
\[
\widehat{\Reg}_{e,k}
=\widehat r_{e,k}-\min_{e'\in\{\attm,\attdidm,\attdid\}}\widehat r_{e',k},
\qquad
\widehat{\overline{\Reg}}_{e}=\max_k\widehat{\Reg}_{e,k}.
\]
Table~\ref{tab:moment_calibrated_main_results} gives the resulting decision problem.

\begin{table}[t]
\centering
\scriptsize
\caption{Distinguishability and Worst-Case Regret in the Moment-Calibrated Design}
\label{tab:moment_calibrated_main_results}
\begin{threeparttable}
\begin{tabular}{lcccccc}
\toprule
Sample & World checks & 3-way acc. & $\overline{\Reg}_{\attm}$ & $\overline{\Reg}_{\attdidm}$ & $\overline{\Reg}_{\attdid}$ & Minimax \\
\midrule
NSW + CPS  & 3/3 & 0.367 & 0.073 & 0.043 & 2.103 & $\attdidm$ \\
NSW + PSID & 3/3 & 0.356 & 0.099 & 0.020 & 2.363 & $\attdidm$ \\
\bottomrule
\end{tabular}
\begin{tablenotes}[flushleft]
\footnotesize
\item ``World checks'' reports, for each comparison sample, the number of worlds (out of
three) in which the design that is supposed to be correctly specified satisfies its
identifying restriction in the simulated data, as measured by the standardized moment gaps
in Appendix~\ref{app:common_rf_details}. A value of 3/3 indicates that each world satisfies
exactly its intended restriction. The three-way accuracy is the held-out random-forest
accuracy for distinguishing the three worlds; chance accuracy is $1/3$. Worst-case regrets
are the row maxima of the regret matrices in Table~\ref{tab:moment_regret_matrix}.
\end{tablenotes}
\end{threeparttable}
\end{table}

Table~\ref{tab:moment_regret_matrix} reports the underlying $3\times 3$ regret matrices.
Each column subtracts the smallest MSE in that world, so every world has at least one zero.
The ``Worst'' column gives the row maximum, and the minimax value is the smallest entry in
that column (in bold).

\begin{table}[t]
\centering
\scriptsize
\caption{Quadratic Regret Matrices in the Moment-Calibrated Design}
\label{tab:moment_regret_matrix}
\begin{threeparttable}
\begin{tabular}{lcccc|cccc}
\toprule
& \multicolumn{4}{c|}{NSW + CPS} & \multicolumn{4}{c}{NSW + PSID} \\
Estimand & $M$ & DIDM & DID & Worst & $M$ & DIDM & DID & Worst \\
\midrule
$\attm$    & 0.000 & 0.073 & 0.034 & 0.073          & 0.000 & 0.099 & 0.068 & 0.099 \\
$\attdidm$ & 0.043 & 0.000 & 0.000 & \textbf{0.043} & 0.020 & 0.000 & 0.000 & \textbf{0.020} \\
$\attdid$  & 2.103 & 1.996 & 0.000 & 2.103          & 2.097 & 2.363 & 0.000 & 2.363 \\
\bottomrule
\end{tabular}
\begin{tablenotes}[flushleft]
\footnotesize
\item Entries are quadratic-loss regrets, rounded to three decimals. The ``Worst'' column
reports the row maximum for each estimand and comparison sample. The minimax-regret estimand
is the one achieving the smallest worst-case regret (bold), namely $\attdidm$ in both
samples.
\end{tablenotes}
\end{threeparttable}
\end{table}

In both comparison samples, $\attdidm$ is the minimax-regret estimand. The diagonal zeros in
Table~\ref{tab:moment_regret_matrix} show the intended pointwise pattern: $\attm$ is best in
the $M$-world, $\attdid$ is best in the DID-world, and $\attdidm$ is best in the DIDM-world.
The decision criterion is minimax regret:
\[
\attdidm
=
\argmin_{e\in\{\attm,\attdidm,\attdid\}}
\max_{k\in\{M,DID,DIDM\}}
\widehat{\Reg}_{e,k}.
\]
The random-forest three-way accuracies, $0.367$ for NSW+CPS and $0.356$ for NSW+PSID, are
close to the chance benchmark of $1/3$. The calibrated worlds therefore satisfy their
intended identifying restrictions while remaining difficult to distinguish before the regret
criterion is applied.

\section{Conclusion}

Researchers routinely choose among matching, DID, and hybrid DIDM designs in panel settings, yet applied work offers little formal guidance on which observational target should anchor the main result when experimental benchmarks are unavailable. This paper provides such guidance.

Our analysis delivers two related results. First, under two economically interpretable conditions, negative selection into treatment and stable untreated outcome dynamics, the three observational estimands satisfy the double-bracketing relation.

Second, once this ordering holds, DIDM is minimax-regret optimal among the three candidate headline estimands under a broad class of symmetric, distance-based loss functions. DIDM therefore emerges as a natural default when a researcher must report a single observational estimate while remaining uncertain about which identifying restriction is closest to the truth.

The main implication for applied work is that, when the double-bracketing logic is credible in a given setting, DIDM should be reported as the headline estimate, with matching and DID serving as lower and upper benchmarks. 


The paper develops a decision-theoretic framework for choosing among common panel-data designs under uncertainty about identifying assumptions. The framework does not replace substantive judgment, but it shows that, in a large class of empirically relevant environments, researchers can make this choice in a disciplined way. When double bracketing is plausible, DIDM provides a robust default.

\setlength{\baselineskip}{6.8mm}
\bibliographystyle{apalike}
\bibliography{reference}

\newpage 
\appendix
\section*{Appendix}

\section{Mathematical Proofs}

\begin{proof}[Proof of Proposition \ref{lem:double_bracketing}]
The identification errors can be written as
\begin{align*}
\Delta\left(\attm\right)= & E\left[E\left[Y_1(0) \mid W=1, Y_{-s}\right]-E\left[Y_1(0) \mid W=0, Y_{-s}\right] \mid W=1\right],
\\
\Delta\left(\attdid\right)= & E\left[Y_1(0)-Y_0(0) \mid W=1\right]-E\left[Y_1(0)-Y_0(0) \mid W=0\right],
\\
\Delta\left(\attdidm\right)= & E\left[E\left[Y_1(0)-Y_0(0) \mid W=1, Y_{-s}\right]-E\left[Y_1(0)-Y_0(0) \mid W=0, Y_{-s}\right] \mid W=1\right].
\end{align*}

First, observe that
\begin{align*}
\Delta(\attdidm) - \Delta(\attm)
=&
E\big[E[Y_0(0)\mid W=0,Y_{-s}] 
- E[Y_0(0)\mid W=1,Y_{-s}] \mid W=1\big]
\\
=&
E\big[E[Y_0\mid W=0,Y_{-s}] 
- E[Y_0\mid W=1,Y_{-s}] \mid W=1\big]
\geq 0,
\end{align*}
where the second equality is due to $Y_t(0)=Y_t$ for all $t \leq 0$, and
the last inequality follows from Assumption \ref{a:special} \eqref{a:special:selection}.

Second, observe that
\begin{align*}
&\Delta(\attdid) - \Delta(\attdidm)
\\
=&
E\big[E[Y_1(0)-Y_0(0)\mid W=0,Y_{-s}] \mid W=1\big]
-
E\big[Y_1(0)-Y_0(0)\mid W=0\big]
\\
=&
E\big[E[Y_1(0)-Y_0(0)\mid W=0,Y_{-s}] \mid W=1\big]
-
E\big[E[Y_1(0)-Y_0(0)\mid W=0,Y_{-s}] \mid W=0\big]
\\
=&
E\big[E[Y_1-Y_0\mid W=0,Y_{-s}] \mid W=1\big]
-
E\big[E[Y_1-Y_0\mid W=0,Y_{-s}] \mid W=0\big]
\\
=& 
E[\Phi(Y_{-s}) \mid W=1]
-
E[\Phi(Y_{-s}) \mid W=0]
\geq 0,
\end{align*}
where the second equality follows from the law of iterated expectations, the third equality is due to $Y_t(0)=Y_t$ given $W=0$, and the last inequality follows from Assumption \ref{a:special} \eqref{a:special:sd}--\eqref{a:special:dec}.
\end{proof}

\begin{remark}[Reversed ordering in the proof]\label{rem:reverse_bracketing_proof}
As noted in Remark~\ref{rem:reverse_bracketing} in the main text, if the selection ordering is reversed, then the same proof goes through with the direction of the inequalities reversed. More specifically, if part~\ref{a:special:selection} holds with the opposite inequality and part~\ref{a:special:sd} holds with the opposite stochastic-dominance direction, while part~\ref{a:special:dec} is unchanged, then the algebra above is identical and only the signs of the two concluding inequalities change.

In the first step, the sign of
\[
\Delta(\attdidm) - \Delta(\attm)
\]
reverses because the conditional mean ordering in part~\ref{a:special:selection} is reversed. In the second step, the sign of
\[
\Delta(\attdid) - \Delta(\attdidm)
\]
reverses because the stochastic-dominance ordering in part~\ref{a:special:sd} is reversed, while the monotonicity condition in part~\ref{a:special:dec} is unchanged. Consequently,
\[
\Delta(\attm)
\;\geq\;
\Delta(\attdidm)
\;\geq\;
\Delta(\attdid),
\]
and hence
\[
\attm
\;\geq\;
\attdidm
\;\geq\;
\attdid.
\]
Thus, only the direction of the bracketing changes; DIDM remains the middle estimand.
\end{remark}

\begin{proof}[Proof of Theorem \ref{theorem:main}]
First, Proposition \ref{lem:double_bracketing} under Assumption \ref{a:special} implies 
\[
\attm \leq \attdidm \leq \attdid.
\]
Therefore,
\begin{align}
\max\{|\attdidm-\attm|,|\attdid-\attdidm|\} 
\;\leq\; |\attdidm-\attm| + |\attdid-\attdidm| = |\attdid - \attm|.
\label{eq:max_ineqaulity}
\end{align}

Second, Assumption \ref{a:loss} implies
\begin{align*}
R(a,\theta) 
&= 
L(a,\theta) - \min_{a' \in \ATT} L(a',\theta)
\\
&=
\bigl[c(\theta) + \ell(|a-\theta|)\bigr] - \bigl[c(\theta) + \ell(0)\bigr] 
= \ell(|a-\theta|).
\end{align*}
Thus, the worst-case regrets are
\begin{align*}
\max_{\theta \in \ATT} R(\attm,\theta) 
&= \ell(|\attdid-\attm|),
\\
\max_{\theta \in \ATT} R(\attdidm,\theta) 
&= \max\{ \ell(|\attdidm-\attm|), \ell(|\attdid-\attdidm|) \},
\\
\max_{\theta \in \ATT} R(\attdid,\theta) 
&= \ell(|\attdid-\attm|).
\end{align*}
The inequality \eqref{eq:max_ineqaulity} and the monotonicity of $\ell$ in Assumption \ref{a:loss} imply
\begin{align*}
\max_{\theta \in \ATT} R(\attdidm,\theta)
&=
\max\{ \ell(|\attdidm-\attm|), \ell(|\attdid-\attdidm|) \}
\\
&\leq
\ell\Big(\max\{|\attdidm-\attm|,|\attdid-\attdidm|\}\Big)
\\
&\leq \ell(|\attdid - \attm|),
\end{align*}
where the last expression $\ell(|\attdid - \attm|)$ equals both 
$\max_{\theta \in \ATT} R(\attm,\theta)$
and
$\max_{\theta \in \ATT} R(\attdid,\theta)$.
Therefore, $\max_{\theta \in \ATT} R(\attdidm,\theta) = \min_{a' \in \ATT} \max_{\theta \in \ATT}R(a',\theta)$ follows.
\end{proof}

\subsection{Detailed Calculations for Section \ref{sec:parametric}}\label{sec:detailed_calculations}

This subsection presents detailed calculations to derive the mathematical expressions in Section \ref{sec:parametric}.
We omit the $i$ subscript throughout this appendix section.

Iterated applications of \eqref{eq:parametric} yield
\begin{align*}
E[Y_1|W=1,Y_{-1}] =& \alpha + \beta + \gamma + \delta_1 + \rho (\alpha + \gamma + \delta_0 + \rho Y_{-1}) 
\\
=& (1+\rho)\alpha + \beta + (1+\rho)\gamma + \delta_1 + \rho\delta_0 + \rho^2 Y_{-1}
\qquad{and}
\\
E[Y_1|W=0,Y_{-1}] =& \alpha + \delta_1 + \rho (\alpha + \delta_0 + \rho Y_{-1}) 
\\
=& (1+\rho)\alpha + \delta_1 + \rho\delta_0 + \rho^2 Y_{-1}.
\end{align*}
Substituting these expressions and using the law of iterated expectations yield
\begin{align*}
\attm &= E[ Y_1 | W=1] - E[ E[Y_1 |W=0, Y_{-1} ] | W=1] 
\\
&= E[ E[Y_1|W=1,Y_{-1}] | W=1] - E[ E[Y_1 |W=0, Y_{-1} ] | W=1] 
= \beta + (1+\rho)\gamma.
\end{align*}

Next, observe that \eqref{eq:parametric} yields
\begin{align*}
Y_1-Y_0 =& \beta W + (\delta_1 - \delta_0) + \rho (Y_0-Y_{-1}) + \epsilon_1 - \epsilon_0 
\\
=& \rho\alpha + (\beta + \rho\gamma) W + \delta_1 - (1-\rho)\delta_0 - \rho(1-\rho) Y_{-1} + \epsilon_1 - (1-\rho)\epsilon_0,
\end{align*}
where the second equality follows from
$Y_0-Y_{-1} = \alpha + \gamma W + \delta_0 - (1-\rho) Y_{-1} + \epsilon_0$ by \eqref{eq:parametric}.
Thus, we have
\begin{align*}
E[Y_1-Y_0|Y_{-1},W=1] =& \rho\alpha + \beta + \rho\gamma + \delta_1 - (1-\rho)\delta_0 - \rho(1-\rho)Y_{-1}
\qquad\text{and}
\\
E[Y_1-Y_0|Y_{-1},W=0] =& \rho\alpha + \delta_1 - (1-\rho)\delta_0 - \rho(1-\rho)Y_{-1}.
\end{align*}
Substituting these expressions yields
\begin{align*}
\attdidm &= E[E[ Y_1 -Y_0 | Y_{-1}, W=1] - E[ Y_1 -Y_0 | Y_{-1}, W=0]|W=1]
= \beta + \rho\gamma.
\end{align*}

Similarly, we have
\begin{align*}
E[Y_1-Y_0|W=1] =& \rho\alpha + \beta + \rho\gamma + \delta_1 - (1-\rho)\delta_0 - \rho(1-\rho)E[Y_{-1}|W=1]
\qquad\text{and}
\\
E[Y_1-Y_0|W=0] =& \rho\alpha + \delta_1 - (1-\rho)\delta_0 - \rho(1-\rho)E[Y_{-1}|W=0].
\end{align*}
Substituting these expressions yields
\begin{align*}
\attdid &= E[ Y_1 - Y_0 | W= 1] - E[ Y_1 - Y_0 | W=0]
\\
&= \beta + \rho\gamma + \rho(1-\rho) (E[Y_{-1}|W=0] - E[Y_{-1}|W=1]).
\end{align*}

Then, it is straightforward to see that 

\begin{align}
\attdidm-\attm
&=
-\gamma
\;\geq\;
0,
\\
\attdid-\attdidm
&=
\rho(1-\rho)\Big(E[Y_{i,-1}\mid W_i=0]-E[Y_{i,-1}\mid W_i=1]\Big)
\;\geq\;
0.
\end{align}
Hence, $\attm \le \attdidm \le \attdid$.

\section{Detail of Literature Review of Past AER Papers}\label{app:aer_census}

This appendix describes how we construct the census of empirical articles from the \emph{American Economic Review} (AER) and how we classify their causal designs into difference-in-differences (DID), matching / lagged dependent variables (M), and hybrid DIDM designs.


We first used Google Scholar to identify AER articles over the five most recent annual volumes at the time of the search.\footnote{Following the spirit of \citet{andrews2024communicating}, we rely on a Google Scholar query.} The search was restricted to articles published in the \emph{American Economic Review} and used combinations of the keywords
``panel data'' and ``repeated cross-section''.

This search is intentionally broad, designed to capture empirical studies that plausibly use a panel or repeated cross-section structure in a causal design.

The initial query returned $114$ items, of which four were either not AER research articles (for example, reference-list fragments or work published in other outlets) or duplicate entries, and were removed, leaving $110$ articles. For each article, we downloaded the main PDF (excluding comments, replies, and purely methodological pieces) and passed it to an LLM-based coding pipeline described below.

\subsection{Scope: Panel-Based Causal Designs}

We then screened each article to determine whether it was \emph{in scope} for our analysis.
An article is classified as \emph{in scope} if it satisfies both of the following:
\begin{enumerate}[(i)]
    \item it uses panel data, that is, repeated observations on the same units over time (unit $\times$ time); this restriction reflects that the M and DIDM estimands require within-unit lagged outcomes and differencing, so the three-way comparison is only jointly defined on panel data; and
    \item the panel structure is used in a regression or design intended to identify causal effects (what we term a ``panel-based causal design'').
\end{enumerate}
Articles that use panel data only for descriptive statistics or event plots, that are purely theoretical, or that study cross-sectional RCTs without a panel dimension are classified as \emph{out of scope}.

Out of the $110$ articles, $33$ were classified as out of scope and excluded from the main population, leaving $77$ articles that employ a panel-based causal design. Throughout this appendix we refer to these $77$ articles as the ``in-scope'' AER panel sample.

\subsection{LLM-Based Coding Protocol}\label{app:llm_protocol}

For each in-scope article we classify the panel-based causal design along three dimensions:
\begin{itemize}
    \item Code 1 (DID): difference-in-differences, two-way fixed effects (TWFE), or event-study designs.
    \item Code 2 (M): matching, lagged dependent variable (LDV) models, or specifications that condition on pre-period outcomes.
    \item Code 3 (DIDM): hybrid designs that combine DID/TWFE/event study with matching or explicit conditioning on pre-treatment outcomes in the same identification strategy.
\end{itemize}

To implement this classification at scale, we used Gemini Pro.
For each article, we supplied the PDF together with the following coding instructions as a prompt:

\begin{quote}\small
\textbf{Prompt:}You are given a PDF of an economics paper. Your task is to determine whether it uses a panel-based causal design, and if so, whether it uses any of three methods:

\begin{itemize}
    \item Code 1 = DID (Difference-in-Differences / TWFE / event study)
    \item Code 2 = M (Matching / LDV / conditioning on pre-period outcomes)
    \item Code 3 = DIDM (Hybrid DID $\times$ M)
\end{itemize}

Always follow the steps and output format below. Do NOT include the word "markdown" in your output. Do NOT explain your reasoning outside the specified fields.

--------------------------------

Step 1: Pre-screen (in scope vs out of scope)

--------------------------------

First, decide if the paper is IN SCOPE for this review.

A paper is IN SCOPE if BOTH of the following hold:
\begin{enumerate}
    \item It uses panel data: repeated observations on the same units over time (unit $\times$ time).
    \item The panel structure is used in a regression or design meant to estimate causal effects (i.e., a ``panel-based causal design'').
\end{enumerate}

The paper is OUT OF SCOPE if it does NOT meet these conditions. Examples of OUT OF SCOPE:
\begin{itemize}
    \item Purely cross-sectional (no time dimension).
    \item Pure theory (no empiwerical causal design).
    \item Descriptive only (no causal regression or design).
    \item Single-period RCT with no panel dimension.
    \item Panel data used only for descriptive averages or plots, with no panel regression or causal identification strategy.
\end{itemize}

If the paper is OUT OF SCOPE:
\begin{itemize}
    \item Set \texttt{IncludeInPopulation: false}
    \item Set \texttt{MainLabel: OUT\_OF\_SCOPE}
    \item Give a 1--2 line reason with at least one quote and page reference.
    \item STOP. Do NOT try to detect DID, M, or DIDM.
\end{itemize}

If the paper IS IN SCOPE (uses panel data in a causal design):
\begin{itemize}
    \item Set \texttt{IncludeInPopulation: true}
    \item Set \texttt{MainLabel: PANEL\_CAUSAL}
    \item Proceed to Step 2.
\end{itemize}

--------------------------------

Step 2: Detect DID, M, DIDM (only)

--------------------------------

Focus particularly on whether these three methods are used in the causal design:

\begin{enumerate}
    \item DID (Code 1): Difference-in-Differences / TWFE / event study
    \item M (Code 2): Matching on lags / LDV / conditioning on pre-period outcomes
    \item DIDM (Code 3): Hybrid DID $\times$ M
\end{enumerate}

For EACH of the three methods (DID, M, DIDM), you must assign a STATUS:

\begin{itemize}
    \item USED\_FOR\_MAIN\_ID = used as a main identification strategy
    \item ROBUSTNESS\_ONLY = used only as a robustness or secondary specification
    \item MENTIONED\_ONLY = mentioned but not actually estimated
    \item NOT\_PRESENT = not mentioned and not used
\end{itemize}

You must base all classifications on explicit information from the paper.

-------------------------

Definition: DID (Code 1)

-------------------------

Count as DID if BOTH are true:

\begin{enumerate}
    \item There are unit and time fixed effects (or equivalent panel structure), AND
    \item The authors describe the design as ``difference-in-differences'', ``event study'', ``two-way fixed effects'', or otherwise rely on a parallel trends idea.
\end{enumerate}

Indicators that DID is used:
\begin{itemize}
    \item Regression equations that include unit and time fixed effects.
    \item Language like ``difference-in-differences'', ``DID'', ``event study'', ``two-way fixed effects'', ``TWFE'', ``parallel trends''.
    \item Discussion of pre-trends tests in an event-study setting.
\end{itemize}

Examples that should be coded as DID (if they fit the above):
\begin{itemize}
    \item Panel regression with unit and time fixed effects estimating the impact of a policy.
    \item Event-study plots using leads and lags of treatment with unit/time FEs.
\end{itemize}

-----------------------

Definition: M (Code 2)

-----------------------

Count as M if ANY MAIN CAUSAL SPECIFICATION includes the outcome from a pre-treatment period (or pre-treatment outcome path) on the right-hand side as a regressor, matching variable, or weighting variable, regardless of the label the authors use.

This includes:
\begin{itemize}
    \item Lagged dependent variable models (e.g., $y_{it}$ regressed on $y_{i,t-1}$).
    \item Controls like ``baseline test score'', ``prior achievement'', ``previous year's earnings'', ``pre-program outcome'', ``outcome in $t-1$'', etc.
    \item Matching or weighting on pre-treatment outcome paths (e.g., matching treated and control units on pre-period outcomes, or constructing weights based on pre-treatment outcomes).
\end{itemize}

The paper does NOT have to use the words ``matching'', ``LDV'', or ``lagged dependent variable''. You must infer M from:
\begin{itemize}
    \item Regression equations that include lagged or baseline outcomes, and/or
    \item Descriptions like ``we control for last year's score'', ``we include the baseline value of Y as a control'', ``we condition on pre-program outcomes'', ``we match on pre-treatment outcome trajectories'', etc.
\end{itemize}

----------------------------

Definition: DIDM (Code 3)

----------------------------

Count as DIDM if the SAME identification strategy COMBINES:

\begin{enumerate}
    \item A DID / TWFE / event-study structure (unit \& time FEs, parallel trends idea), AND
    \item Matching / weighting / conditioning on pre-treatment outcomes (M) as part of the design for comparability.
\end{enumerate}

In other words, DIDM is present if pre-treatment outcomes (e.g., lagged Y, baseline Y, pre-trends) are explicitly used to construct or refine the DID comparison itself, not just as a completely separate robustness check.

Typical patterns that COUNT as DIDM:
\begin{itemize}
    \item ``We first match (or weight) treated and control units on pre-treatment outcomes/trends and then estimate a difference-in-differences / TWFE model on the matched/weighted sample.''
    \item ``We construct weights based on pre-program outcomes and then run an event-study with unit and time fixed effects on the reweighted data.''
    \item ``Our main specification is a DID that conditions flexibly on the lagged outcome / pre-period outcome path to address differential trends.''
\end{itemize}

Cases that do NOT count as DIDM:
\begin{itemize}
    \item The paper has a baseline DID spec and a separate LDV/matching robustness spec, but they are presented as distinct estimators (e.g., ``as a robustness check, we also estimate an LDV model'').
    \begin{itemize}
        \item In that case:
        \begin{itemize}
            \item Classify DID and M separately (e.g., DID: USED\_FOR\_MAIN\_ID, M: ROBUSTNESS\_ONLY).
            \item Set DIDM: NOT\_PRESENT.
        \end{itemize}
    \end{itemize}
\end{itemize}

\end{quote}
For each article, the LLM returned:
\begin{enumerate}[(i)]
    \item an in-scope vs.\ out-of-scope decision;
    \item a label for each of the three methods (DID, M, DIDM) indicating whether it was used for main identification, used only for robustness, mentioned only, or not present; and
    \item supporting quotes from the paper with page and section references.
\end{enumerate}

\subsection{Human Verification and Classification Error}\label{app:verification}

The raw LLM classifications were then subjected to human verification.
For each in-scope article, we read the provided quotes and assessed whether the assigned labels were reasonable given the context.
If a quote appeared clearly inconsistent with the assigned label (for example, a quote describing a purely cross-sectional regression coded as DID, or a baseline-control specification coded as matching when it was used only in a robustness check), we manually inspected the relevant sections of the paper and corrected the classification.

To assess the accuracy of this coding, a member of the research team independently reviewed every article in the census, comparing the assigned codes against the supporting quotes and the underlying papers. The classification was judged incorrect for $14$ of the $110$ articles (approximately $12.7\%$); these cases were corrected, and the corrected classification is the one summarized below.\footnote{Because some designs sit in a gray zone, this figure is best read as an approximate upper bound on the residual misclassification rate rather than an exact error rate.} Some classifications necessarily involve judgment, and the reported figures should be read with this in mind.

\subsection{Summary of Design Types in the AER Panel Sample}\label{app:aer_summary}

Table~\ref{tab:aer_design_types} summarizes the distribution of design types in the in-scope AER panel sample after applying the LLM-based coding and human verification described above.
Recall that the codes are not mutually exclusive: a paper may, for example, use both DID and M (Code 1 + Code 2) in its causal analysis.

Out of the $110$ AER articles in the census, $33$ ($30.00\%$) were out of scope.
The remaining $77$ articles ($100\%$ of the in-scope sample) employ a panel-based causal design.
Among these $77$ in-scope articles, the combinations of methods are as follows:
\begin{itemize}
    \item $33$ papers ($42.86\%$) use only DID-type methods (Code 1) and neither M nor DIDM.
    \item $9$ papers ($11.69\%$) use only M-type methods (Code 2) and neither DID nor DIDM.
    \item No paper uses only DIDM (Code 3) without DID or M.
    \item $12$ papers ($15.58\%$) use both DID and M (Codes 1 + 2).
    \item $4$ papers ($5.19\%$) use DID together with DIDM (Codes 1 + 3) but not standalone M.
    \item No paper uses M and DIDM without DID (Codes 2 + 3).
    \item $5$ papers ($6.49\%$) use all three (Codes 1 + 2 + 3).
    \item $14$ papers ($18.18\%$) use none of the three methods as defined above (e.g., they rely on alternative panel designs such as synthetic control, pure fixed-effects models without a parallel-trends interpretation, or other identification strategies).
\end{itemize}

Aggregating across combinations, a total of $54$ papers ($70.13\%$) use DID (Code 1) in some capacity (either as the main identification strategy or as a robustness specification), $26$ papers ($33.77\%$) use M-type methods (Code 2), and $9$ papers ($11.69\%$) use DIDM-type hybrid designs (Code 3).

\begin{table}[t]
    \centering
    \caption{Classification of Panel-Based Causal Designs in the AER Panel Sample}
    \label{tab:aer_design_types}
    \begin{tabular}{lrr}
        \toprule
        Category & Number & Percentage \\
        \midrule
        Out of scope (no panel-based causal design) & 33 & 30.00\% \\
        In-scope (panel-based causal design)        & 77 & 100.00\% \\
        \addlinespace
        Only Code 1 (DID)                          & 33 & 42.86\% \\
        Only Code 2 (M)                            &  9 & 11.69\% \\
        Only Code 3 (DIDM)                         &  0 &  0.00\% \\
        Codes 1 + 2                                & 12 & 15.58\% \\
        Codes 1 + 3                                &  4 &  5.19\% \\
        Codes 2 + 3                                &  0 &  0.00\% \\
        Codes 1 + 2 + 3                            &  5 &  6.49\% \\
        Neither 1, 2, nor 3                        & 14 & 18.18\% \\
        \addlinespace
        Total using Code 1 (any combination)       & 54 & 70.13\% \\
        Total using Code 2 (any combination)       & 26 & 33.77\% \\
        Total using Code 3 (any combination)       &  9 & 11.69\% \\
        \bottomrule
    \end{tabular}
    \vspace{0.5em}
    
    \footnotesize
    \emph{Notes:} The table reports the distribution of design types among the $77$ in-scope AER articles that employ a panel-based causal design.
    Codes 1, 2, and 3 correspond to DID, M, and DIDM, respectively.
    Categories ``Only 1'', ``Only 2'', and ``Only 3'' indicate that the paper uses exactly that code and neither of the others.
    ``Neither'' indicates that none of the three codes is present according to our classification.
    ``Total'' rows count all papers in which the corresponding code appears in any role (main identification, robustness, or otherwise).
\end{table}

\section{Details of Data}
\label{sec:appendix:details}

This appendix section provides details of the data used in the empirical applications and in the empirical assessment of the assumptions.

\subsection{Details of the NSW Data}\label{sec:appendix:nsw}
\subsubsection{Data Description}

Our primary dataset originates from the National Supported Work (NSW) Demonstration, a transitional subsidized work experience program that operated for four years across 15 locations in the United States. This initiative specifically targeted four distinct groups: female long-term AFDC recipients, former drug addicts, ex-offenders, and young school dropouts. Approximately 10,000 individuals took part in the program, each engaging in 12 to 18 months of employment.

The NSW program aimed to assist individuals who faced significant barriers to employment. It provided a structured training environment initially, followed by support in securing regular employment. To ensure the program reached those in genuine need, participants were required to be currently unemployed and to have limited recent employment experience, highlighting the program’s focus on individuals with considerable employment challenges.

A standout feature of the NSW program was its experimental design, which included a randomized control trial at 10 locations between April 1975 and August 1977. In this trial, 6,616 participants were randomly assigned to either a treatment group, which received the program services, or a control group, which did not. Data collection involved a retrospective baseline interview and four follow-up interviews, covering two years before random assignment and up to 36 months afterward. The dataset provides comprehensive information on demographics, employment history, job search behavior, mobility, household income, housing, and drug use.

\subsubsection{Key Variables}

In our analysis, we concentrate on the following variables, which are consistent across both the experimental and non-experimental datasets.

The primary outcome of interest is $Y_t$, representing the participants' self-reported earnings. Specifically, we analyze real earnings adjusted to 1982 dollars, in line with the methodology established by \citet{lalonde1986evaluating}.
Next, $W$ serves as a binary indicator of treatment, denoting whether an individual was assigned to the NSW program.
Additionally, we incorporate demographic variables such as race, and education level (indicating high school dropout status) as auxiliary covariates commonly used in M, DIDM, and DID.

The straightforward mean-difference estimate of the NSW program's impact on male participants within the experimental sample is \$886, a figure that is statistically significant at the 10 percent level. This experimental estimate serves as the benchmark against which we evaluate the non-experimental M, DID, and DIDM estimands. The NSW participants are economically disadvantaged, exhibiting low pre-program earnings and a decline in earnings from 1974 to 1975 that is widely recognized as ``Ashenfelter's dip.'' As discussed in Section~\ref{sec:lagged_outcome}, this pre-program pattern bears not on the magnitude of the experimental impact but on the selection dynamics that drive the non-experimental estimands.

We chose not to utilize the Dehejia-Wahba (DW) dataset from their 1999 and 2002 studies in our analysis, based on several critical considerations. First, \citet{smith2005does} and \citet{dehejia2005practical} have debated the validity of the DW dataset, particularly questioning its representativeness and the potential biases introduced by the sample restrictions employed. DW exclude approximately 40 percent of the original LaLonde (1986) sample in order to include two years of pre-program earnings data in their model of program participation. This exclusion results in lower mean earnings in 1974 and 1975 for the DW sample compared to the larger LaLonde sample, leading to a significantly different and larger experimental impact estimate of \$1,794, which is more than double that of the LaLonde sample.

Additionally, the data we obtained from the authors of \citet{heckman1998_2matching} includes all pretreatment earnings outcomes, even for the subsample omitted in the DW dataset. This comprehensive dataset removes the need to impose arbitrary sample restrictions, which could otherwise increase sampling uncertainty and potentially bias the results. By employing the full sample, we ensure a broader and more representative analysis, thus upholding the principles of internal and external validity.

\subsection{Details of the JTPA Data}\label{sec:appendix:jtpa}

\subsubsection{Data Description}
We closely follow the description laid out in \cite{heckman1998characterizing}.
Our primary dataset originates from a randomized evaluation of the Job Training Partnership Act (JTPA) program, conducted across four training centers in the United States. The JTPA program aimed to provide job training and employment services to economically disadvantaged individuals, dislocated workers, and others who faced significant barriers to employment.

The dataset includes information on both experimental treatment and control groups, as well as a non-experimental comparison group of eligible nonparticipants (ENPs) who were located in the same labor markets but chose not to participate in the program at the time of random assignment. Random assignment occurred when individuals applied and were accepted into the JTPA program, ensuring that participants were comparable at the baseline. Control group members were excluded from receiving JTPA services for 18 months after random assignment.

The data collection involved comprehensive surveys administered to all groups, including the ENPs. These surveys captured detailed retrospective information on labor force participation, job spells, earnings, marital status, and other demographic characteristics. In this analysis, we focus on a sample of adult males aged 22 to 54, following \cite{heckman1998characterizing}.

\subsubsection{Key Variables}

In our analysis, we concentrate on the following key variables, which are consistent across both the experimental and non-experimental datasets.

The primary outcome of interest is $Y_t$, representing the participants' earnings. Specifically, we analyze real earnings over a specific period, adjusting for inflation where necessary. The variable $W$ serves as a binary indicator of treatment, denoting whether an individual was assigned to the JTPA program.

Additionally, we include demographic covariates such as age, which are commonly utilized in various econometric models like Matching (M), Difference-in-Differences Matching (DIDM), and Difference-in-Differences (DID). Because we could not reproduce the exact sample-selection criteria of \citet{heckman1998characterizing}, the JTPA evidence reported in this paper is taken from the published benchmark estimates summarized in \citet{chabe2017should} and \citet{smith2005does} rather than from an independent re-estimation; the variable definitions above describe the constructs underlying those estimates.

\subsection{Details of the Education Data}\label{sec:appendix:educ}
\subsubsection{Data Description}\label{sec:data}

Our primary observational data come from the administrative records of a large urban school district. The dataset includes information on approximately two million children in grades 3 through 8, covering those born between 1966 and 2001.

This dataset encompasses around 15 million test scores in English language arts and math. Due to changes in the testing regime over the past 20 years,such as the transition from district-specific to statewide tests and variations in test timing, we have normalized the test scores by year and grade to have a mean of zero and a standard deviation of one, following established research practices \citep[e.g.,][]{staiger2010searching}. This normalization ensures comparability with other samples across the nation. We also imputed missing test scores using cohort-specific means based on year of birth to account for cohort-level heterogeneity.


\subsubsection{Key Variables}

We focus on the following variables in our analysis. The primary outcome of interest is \textbf{$Y_t$}, representing students' test scores, specifically standardized scores that average results from both mathematics and English language arts.

Secondly, \textbf{$W$} is a binary indicator denoting treatment, which in this context refers to the assignment to a small class size.

Lastly, we use gender, race, and eligibility for free lunch to define subpopulations for further analysis.

\section{Empirical Evidence for the Assumptions}\label{sec:empirical_assumption}

Sections~\ref{sec:double_bracketing}--\ref{sec:general} provide theoretical conditions under which the double-bracketing relationship
\(
\attm \;\leq\; \attdidm \;\leq\; \attdid
\)
holds. We now examine whether the underlying assumptions, namelyAssumption~\ref{a:special} in the simple setup and, implicitly, Assumption~\ref{a:general} in richer designs, are supported by the data sets used in our empirical double-bracketing analysis in Section~\ref{sec:empirical_double_bracketing}.

We focus on the nonexperimental comparison samples for the NSW program (CPS and PSID; Section~\ref{sec:nsw}) and on the educational program of \citet{athey2025combining} (Section~\ref{sec:educ}).\footnote{Although we obtained the JTPA microdata, we were unable to reproduce the exact sample-selection criteria of the original study despite repeated communication with the authors; we therefore rely on the published benchmark estimates for the JTPA application and do not report assumption diagnostics for it.}

\subsection{Job-Training Programs}\label{sec:nsw:assumption}

Section~\ref{sec:nsw} documented robust empirical evidence for the double-bracketing relationship $\attm \leq \attdidm \le \attdid$ for the NSW program under both the CPS and PSID comparison samples. Here we examine each of the three components (i)--(iii) of Assumption~\ref{a:special} in those data.

In the CPS sample, we trim the extreme upper tail of baseline earnings by restricting to $Y_0<26{,}000$, which removes observations above a heavily top-coded / heaped region. In the PSID sample, we restrict to $Y_0<50{,}000$, removing upper-tail controls outside the treated support.

\paragraph{Negative selection on $Y_0$.}
Assumption~\ref{a:special}\eqref{a:special:selection} requires
\[
E[Y_0\mid W=0,Y_{-s}=y]
\;\geq\;
E[Y_0\mid W=1,Y_{-s}=y]
\quad
\text{for all } y.
\]
We estimate the two conditional expectations nonparametrically using partitioning-based least squares regression,\footnote{\label{foot:nonparametric_estimation}We use the \texttt{R} package \texttt{lspartition} of \citet{cattaneo2019lspartition} with default settings.} and plot
\[
y\mapsto E[Y_0\mid W=0,Y_{-s}=y],\qquad
y\mapsto E[Y_0\mid W=1,Y_{-s}=y],
\]
with pointwise 95\% confidence bands in Figure~\ref{fig:nsw:ass1}. The inequality holds uniformly in $y$ in both CPS and PSID data, providing support for Assumption~\ref{a:special}\eqref{a:special:selection}.

\begin{figure}[t]
\centering
CPS\\
\includegraphics[width=0.55\textwidth]{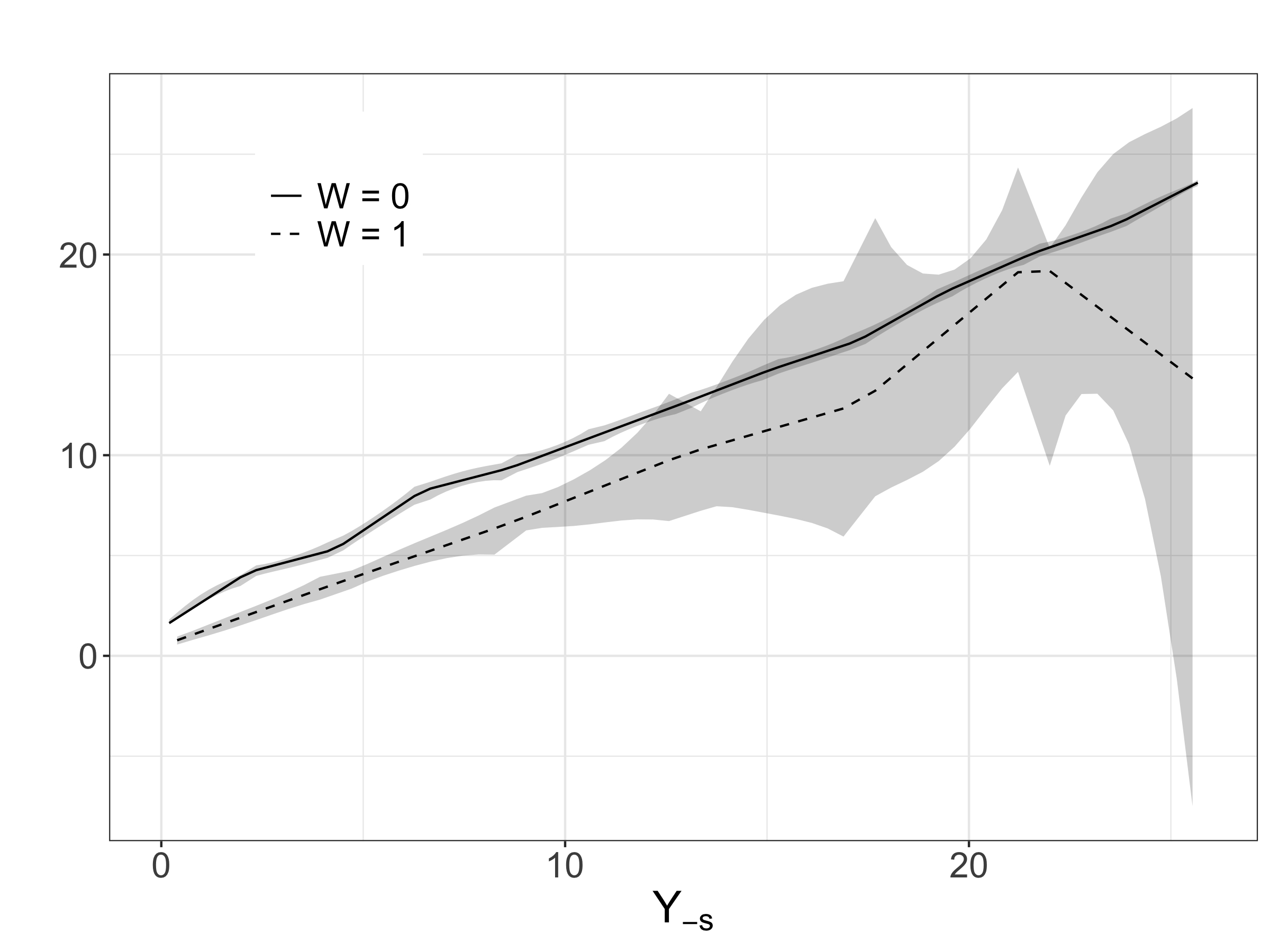}
\\[0.5em]
PSID\\
\includegraphics[width=0.55\textwidth]{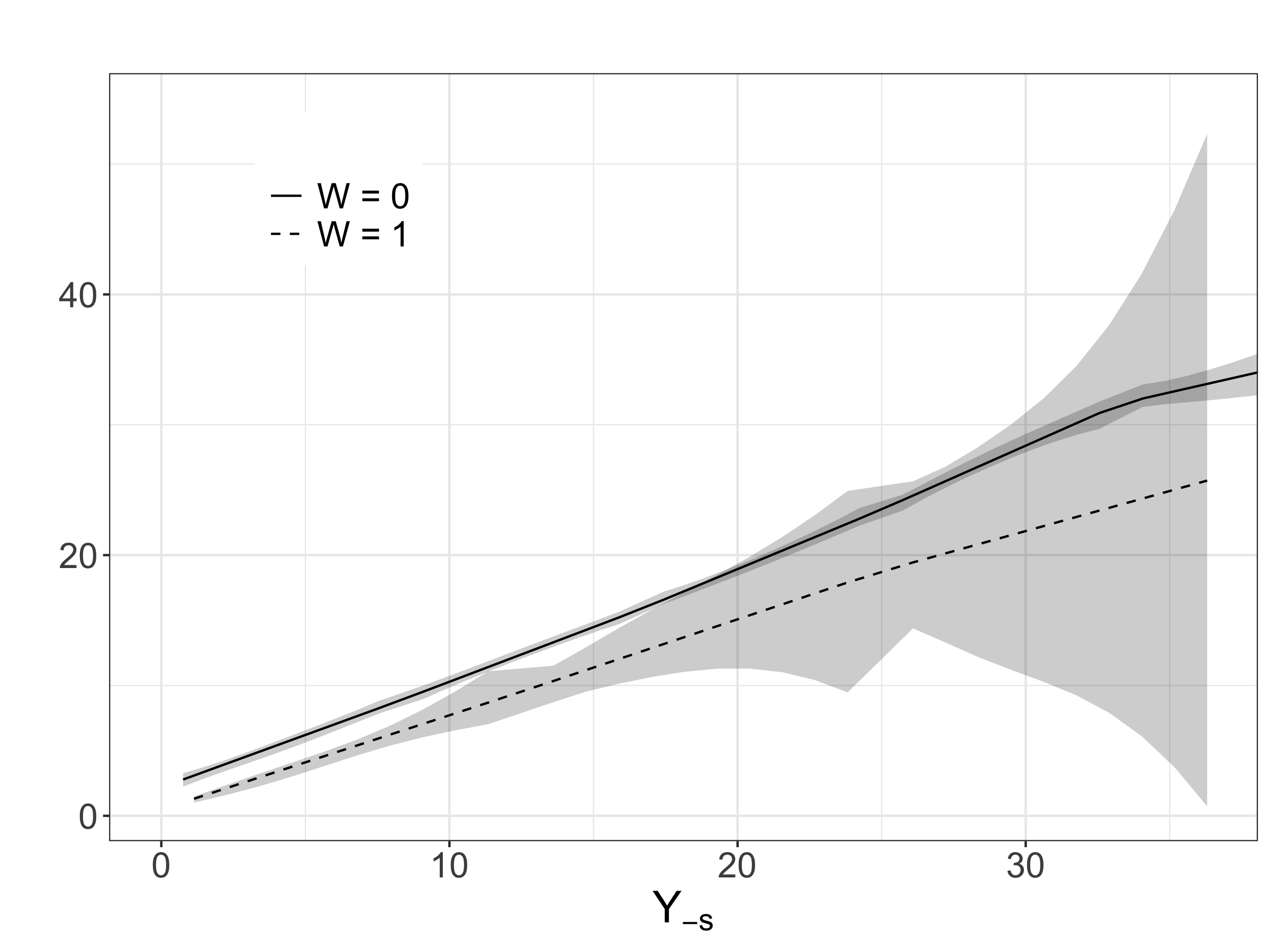}
\caption{Evidence on Assumption~\ref{a:special}\eqref{a:special:selection} for the NSW program, CPS (top) and PSID (bottom). Solid and dashed lines show estimates of $E[Y_0\mid W=0,Y_{-s}=y]$ and $E[Y_0\mid W=1,Y_{-s}=y]$, respectively. Shaded regions are 95\% confidence bands. Both axes are in thousands of U.S.\ dollars. See Footnote~\ref{foot:nonparametric_estimation} for estimation details.}
\label{fig:nsw:ass1}
\end{figure}

\paragraph{Distributional dominance in $Y_{-s}$.}
Assumption~\ref{a:special}\eqref{a:special:sd} requires $F_{Y_{-s}\mid W=0}$ to first-order stochastically dominate $F_{Y_{-s}\mid W=1}$. Figure~\ref{fig:nsw:ass2} plots empirical CDFs of $Y_{-s}$ separately by $W$, with 95\% confidence bands. In both CPS and PSID, the control CDF lies uniformly below the treated CDF, consistent with Assumption~\ref{a:special}\eqref{a:special:sd}.

\begin{figure}[t]
\centering
\begin{tabular}{cc}
CPS & PSID\\
\includegraphics[width=0.5\textwidth]{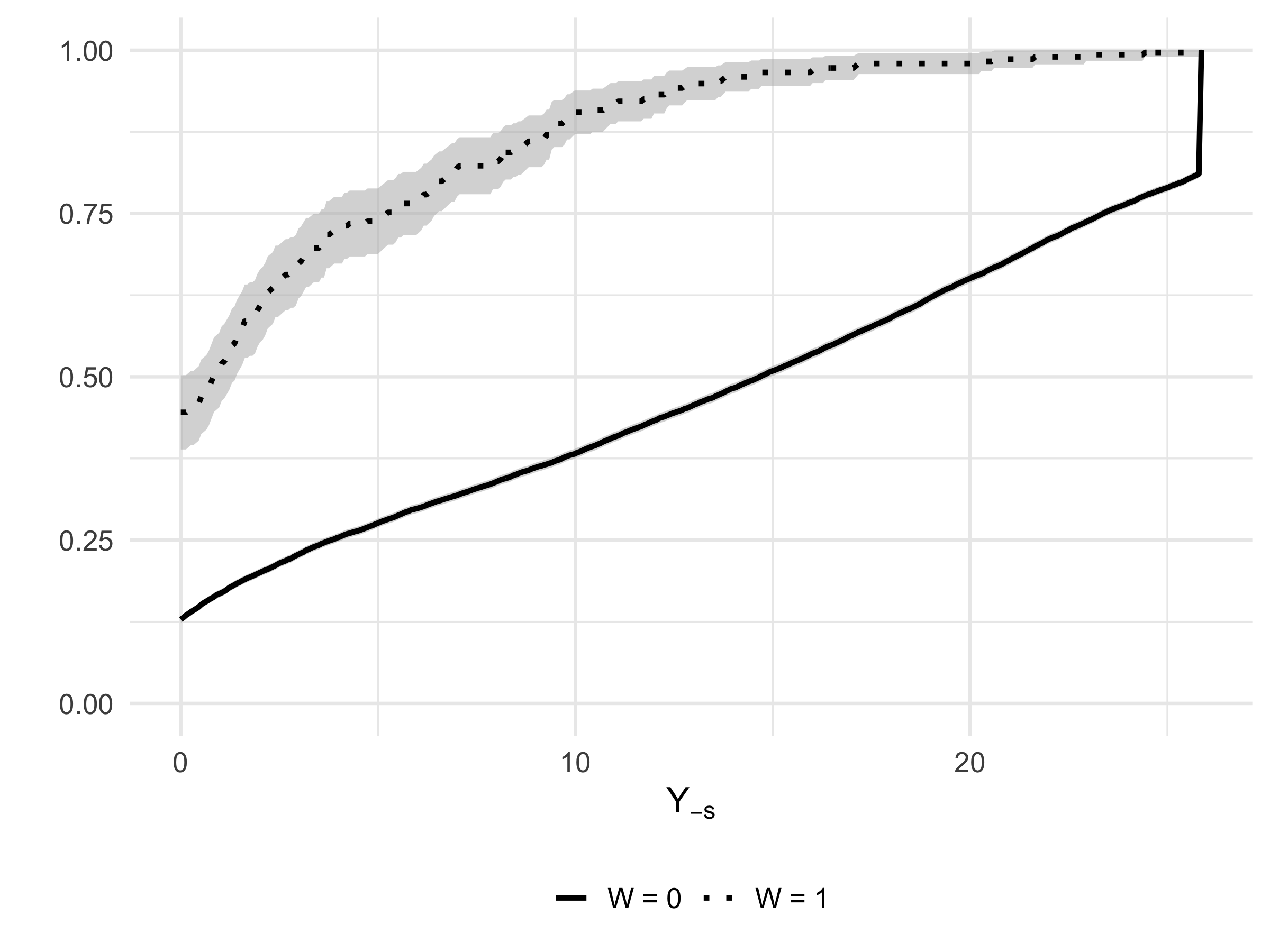}
&
\includegraphics[width=0.5\textwidth]{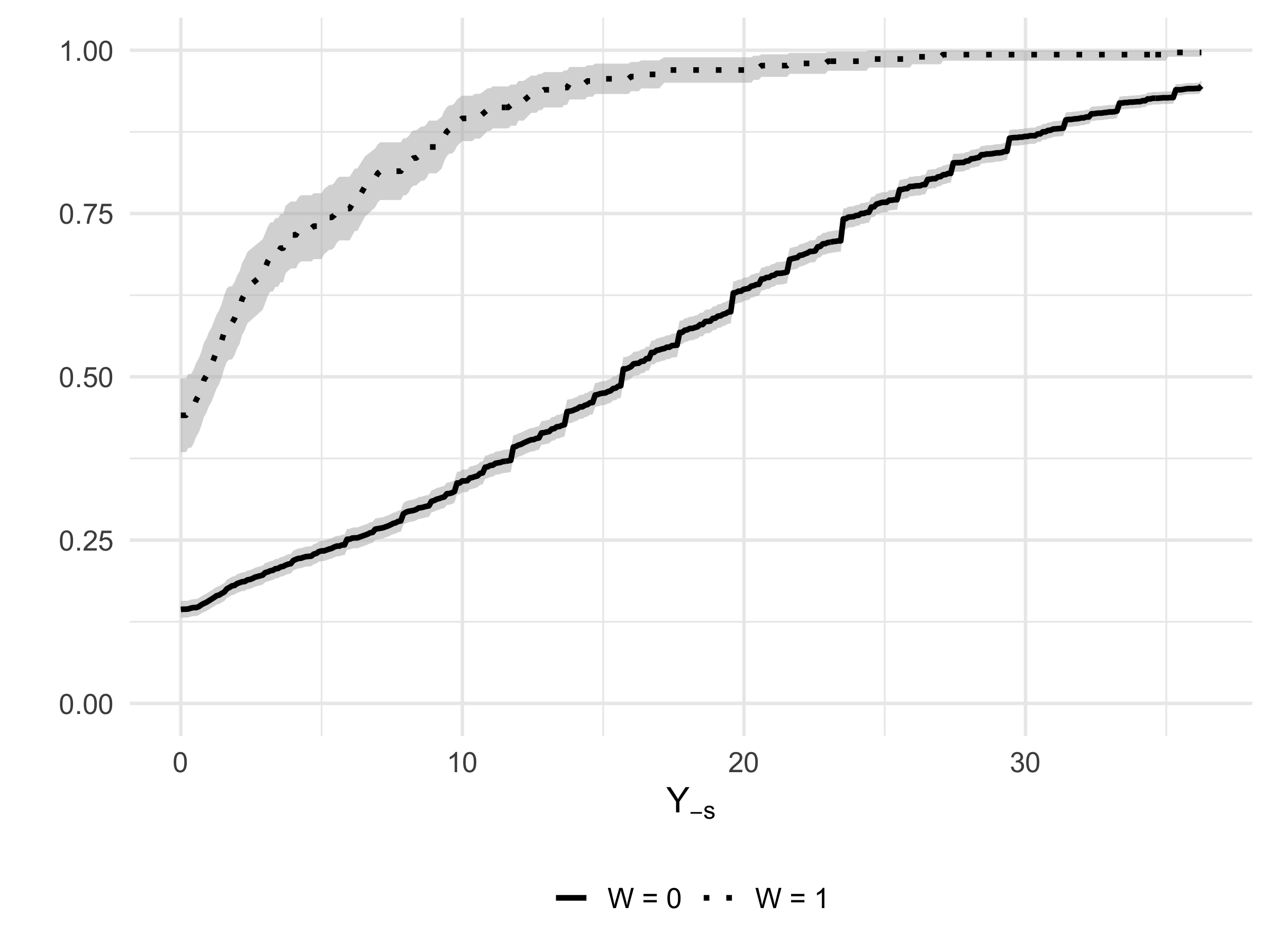}
\end{tabular}
\caption{Evidence on Assumption~\ref{a:special}\eqref{a:special:sd} for the NSW program, CPS (left) and PSID (right). Solid and dotted lines show empirical CDFs of $Y_{-s}$ for $W=0$ and $W=1$, respectively, with 95\% confidence bands. Horizontal axes are in thousands of U.S.\ dollars.}
\label{fig:nsw:ass2}
\end{figure}

\paragraph{Decreasing untreated growth.}
Assumption~\ref{a:special}\eqref{a:special:dec} requires the function
\[
\Phi(y) := E[Y_1-Y_0\mid W=0,Y_{-s}=y]
\]
to be weakly decreasing. We estimate $\Phi$ nonparametrically using partitioning-based least squares regression and plot the resulting curve with 95\% confidence bands in Figure~\ref{fig:nsw:ass3}. In both CPS and PSID, the estimated function is consistent with a weakly decreasing pattern, and the hypothesis of monotonicity is not refuted by the bands.

\begin{figure}[t]
\centering
CPS\\
\includegraphics[width=0.55\textwidth]{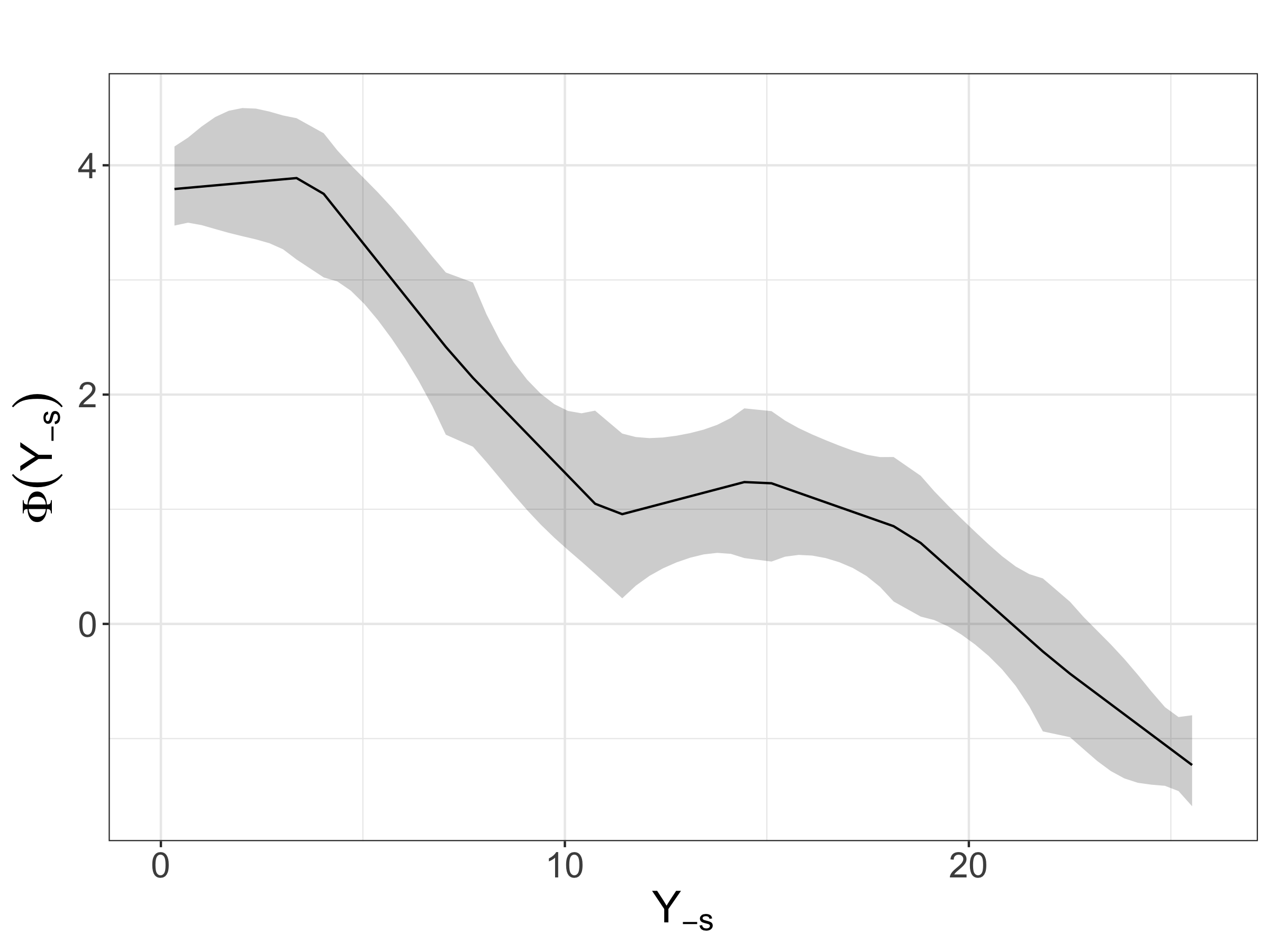}
\\[0.5em]
PSID\\
\includegraphics[width=0.55\textwidth]{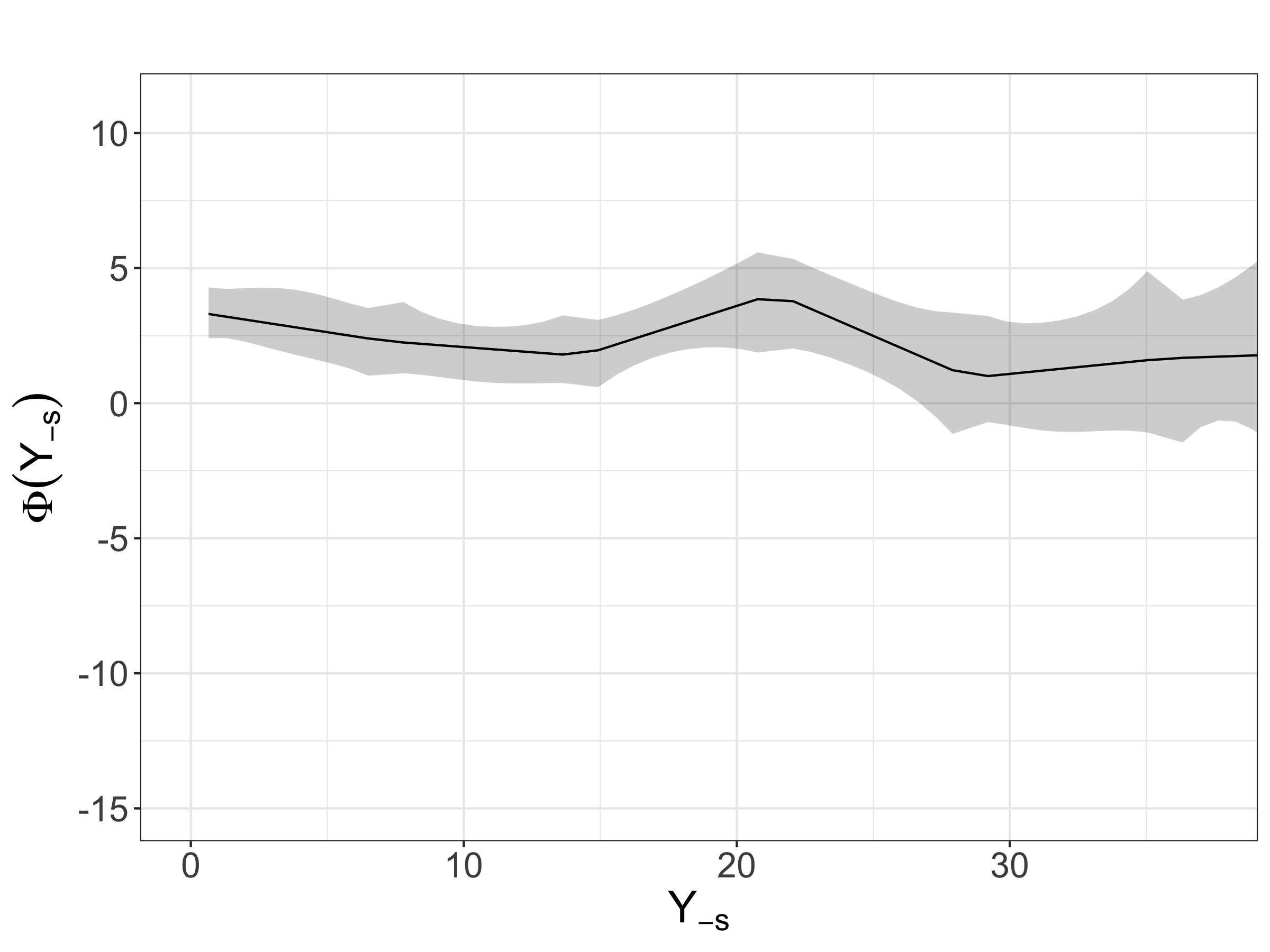}
\caption{Evidence on Assumption~\ref{a:special}\eqref{a:special:dec} for the NSW program, CPS (top) and PSID (bottom). The function $\Phi(y)=E[Y_1-Y_0\mid W=0,Y_{-s}=y]$ is estimated nonparametrically with 95\% confidence bands. Both axes are in thousands of U.S.\ dollars. See Footnote~\ref{foot:nonparametric_estimation} for estimation details.}
\label{fig:nsw:ass3}
\end{figure}

Appendix~\ref{sec:additional:nsw} reports analogous diagnostics after residualizing outcomes with respect to auxiliary covariates; the conclusions are unchanged. Taken together, these diagnostics provide empirical support for Assumption~\ref{a:special} in the NSW applications, consistent with the observed double bracketing in Section~\ref{sec:nsw} and the theoretical predictions of Proposition~\ref{lem:double_bracketing} and Theorem~\ref{theorem:main}.

\subsection{Educational Program}\label{sec:educ:assumption}

Section~\ref{sec:educ} showed that the double-bracketing relationship $\attm \leq \attdidm \leq \attdid$ also holds robustly in the educational intervention studied by \citet{athey2025combining}. We now examine Assumption~\ref{a:special} in that data set.

As before, we focus on the three components of Assumption~\ref{a:special}.

\paragraph{Negative selection on $Y_0$.}
We estimate $E[Y_0\mid W=w,Y_{-s}=y]$ for $w\in\{0,1\}$ using partitioning-based least squares regression, and plot the two conditional expectations with 95\% confidence bands in Figure~\ref{fig:educ:ass1}. The inequality
\[
E[Y_0\mid W=0,Y_{-s}=y]
\;\geq\;
E[Y_0\mid W=1,Y_{-s}=y]
\]
holds across the support of $Y_{-s}$, providing strong support for Assumption~\ref{a:special}\eqref{a:special:selection}.

\begin{figure}[t]
\centering
\includegraphics[width=0.55\textwidth]{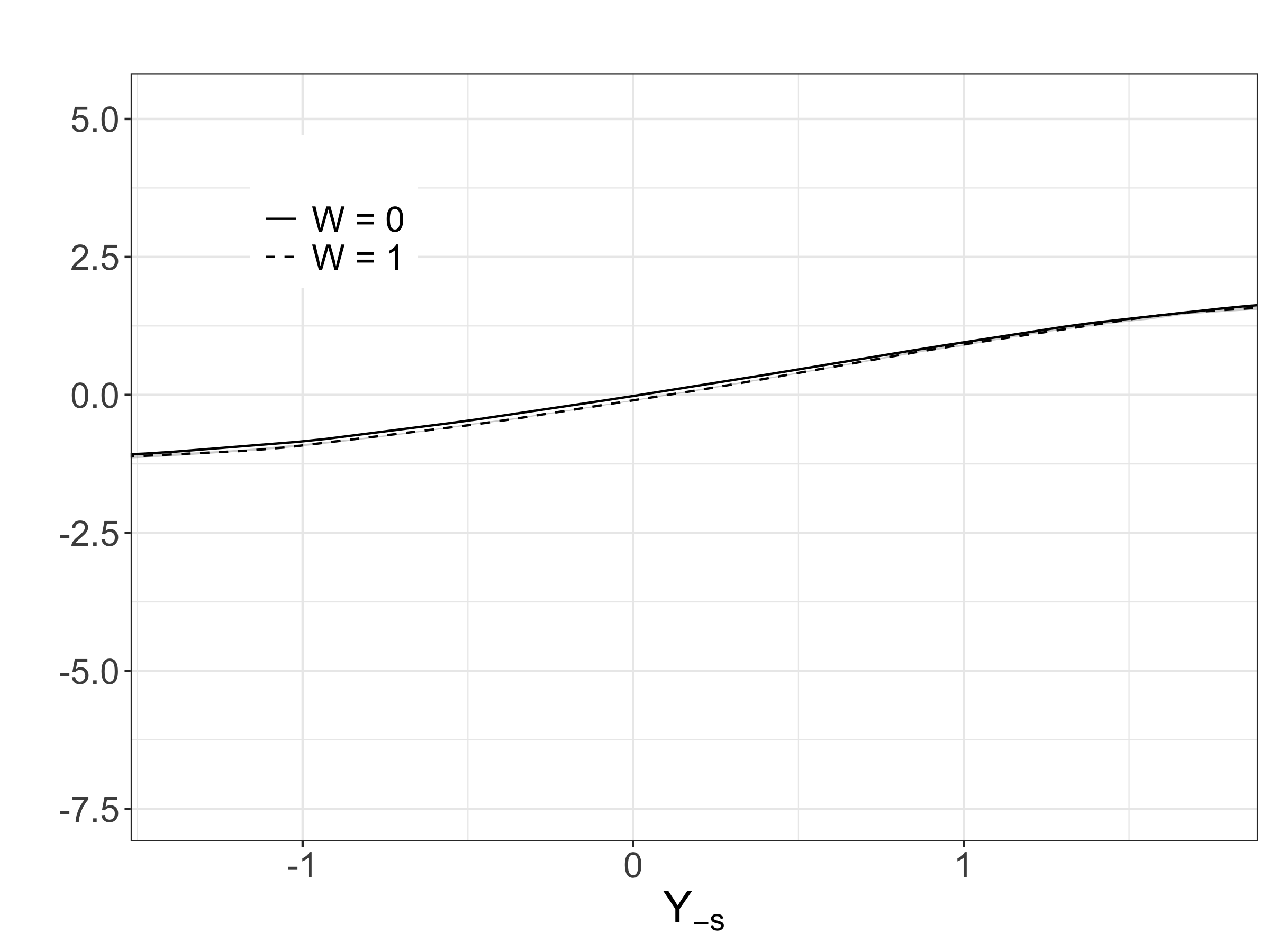}
\caption{Evidence on Assumption~\ref{a:special}\eqref{a:special:selection} for the educational program. Solid and dashed lines show estimates of $E[Y_0\mid W=0,Y_{-s}=y]$ and $E[Y_0\mid W=1,Y_{-s}=y]$, respectively, with 95\% confidence bands (barely visible due to large sample size).}
\label{fig:educ:ass1}
\end{figure}

\paragraph{Distributional dominance in $Y_{-s}$.}
Figure~\ref{fig:educ:ass2} plots empirical CDFs of $Y_{-s}$ for treated and control units, with 95\% confidence bands. The control CDF lies uniformly below the treated CDF, indicating first-order stochastic dominance and supporting Assumption~\ref{a:special}\eqref{a:special:sd}.

\begin{figure}[t]
\centering
\includegraphics[width=0.5\textwidth]{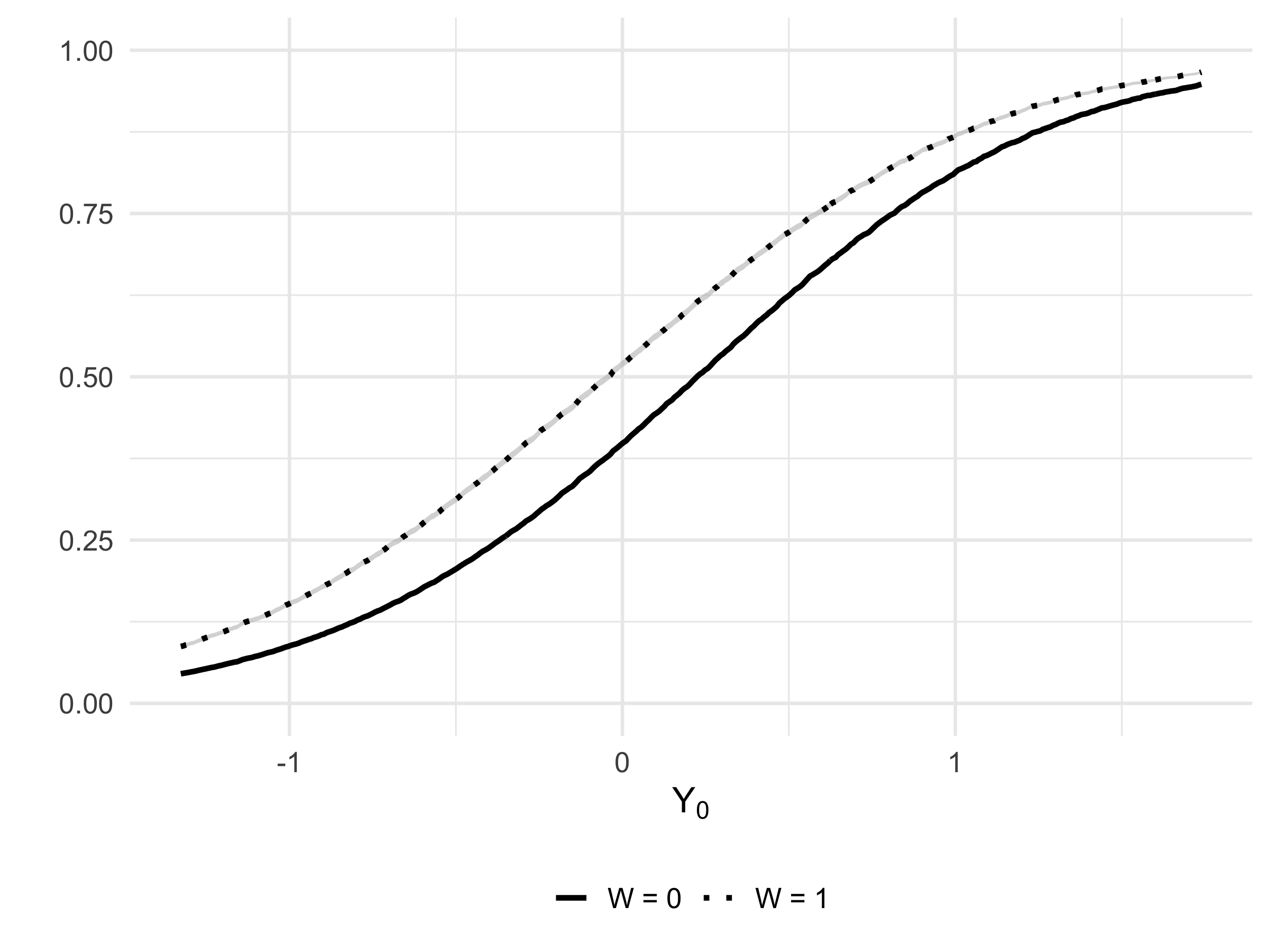}
\caption{Evidence on Assumption~\ref{a:special}\eqref{a:special:sd} for the educational program. Solid and dotted lines show empirical CDFs of $Y_{-s}$ for $W=0$ and $W=1$, respectively, with 95\% confidence bands.}
\label{fig:educ:ass2}
\end{figure}

\paragraph{Decreasing untreated growth.}
We estimate
\[
\Phi(y)=E[Y_1-Y_0\mid W=0,Y_{-s}=y]
\]
nonparametrically and plot the estimated function with 95\% confidence bands in Figure~\ref{fig:educ:ass3}. The estimated $\Phi$ is weakly decreasing in $y$, and the bands are consistent with monotonicity, supporting Assumption~\ref{a:special}\eqref{a:special:dec}.

\begin{figure}[t]
\centering
\includegraphics[width=0.55\textwidth]{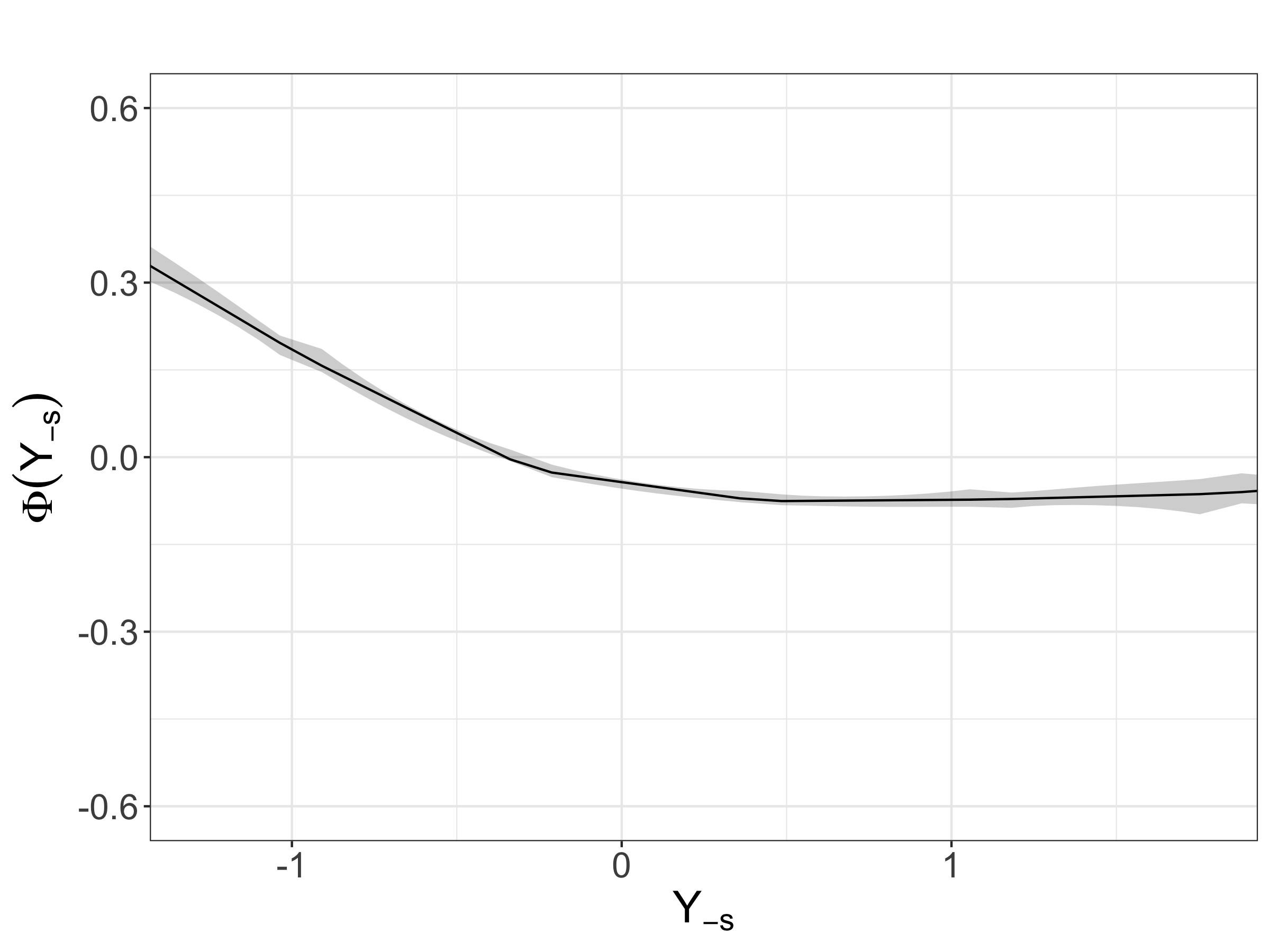}
\caption{Evidence on Assumption~\ref{a:special}\eqref{a:special:dec} for the educational program. The function $\Phi(y)=E[Y_1-Y_0\mid W=0,Y_{-s}=y]$ is estimated nonparametrically with 95\% confidence bands.}
\label{fig:educ:ass3}
\end{figure}

Appendix~\ref{sec:additional:educ} presents diagnostics after residualizing outcomes with respect to observed covariates; the conclusions remain unchanged. Thus, Assumption~\ref{a:special} appears to hold robustly in this setting as well, in line with the observed double bracketing in Section~\ref{sec:educ}.

\subsection{Additional Empirical Evidence of the Assumptions}

In the above, we examine Assumption \ref{a:special} for each empirical application, focusing on the primary variables without considering auxiliary covariates. Here, we provide additional empirical evidence for Assumption \ref{a:special}, now accounting for the auxiliary covariates that were omitted in the main text.

To incorporate these covariates while maintaining the same diagnostic structure as in the main text, we first residualize the relevant variables with respect to the observed auxiliary covariates and then apply the same plotting procedures to the residualized variables. In this appendix, we report residualized counterparts of the figures for Assumption \ref{a:special} \eqref{a:special:selection} and Assumption \ref{a:special} \eqref{a:special:dec}. 

\subsubsection{Additional Empirical Analyses for the NSW Data}\label{sec:additional:nsw}

Figure \ref{fig:nsw:ass1_resid} presents the counterparts of Figure \ref{fig:nsw:ass1}, after residualizing with respect to the auxiliary covariates.
Observe that the required inequality $E[Y_0|W=0,Y_{-s}=y] \ge E[Y_0|W=1,Y_{-s}=y]$ is still satisfied both for the CPS and PSID data sets, providing robust evidence in support of our Assumption \ref{a:special} \eqref{a:special:selection} even after accounting for the auxiliary covariates.
\begin{figure}[t]
\centering
CPS\\
\includegraphics[width=0.5\textwidth]{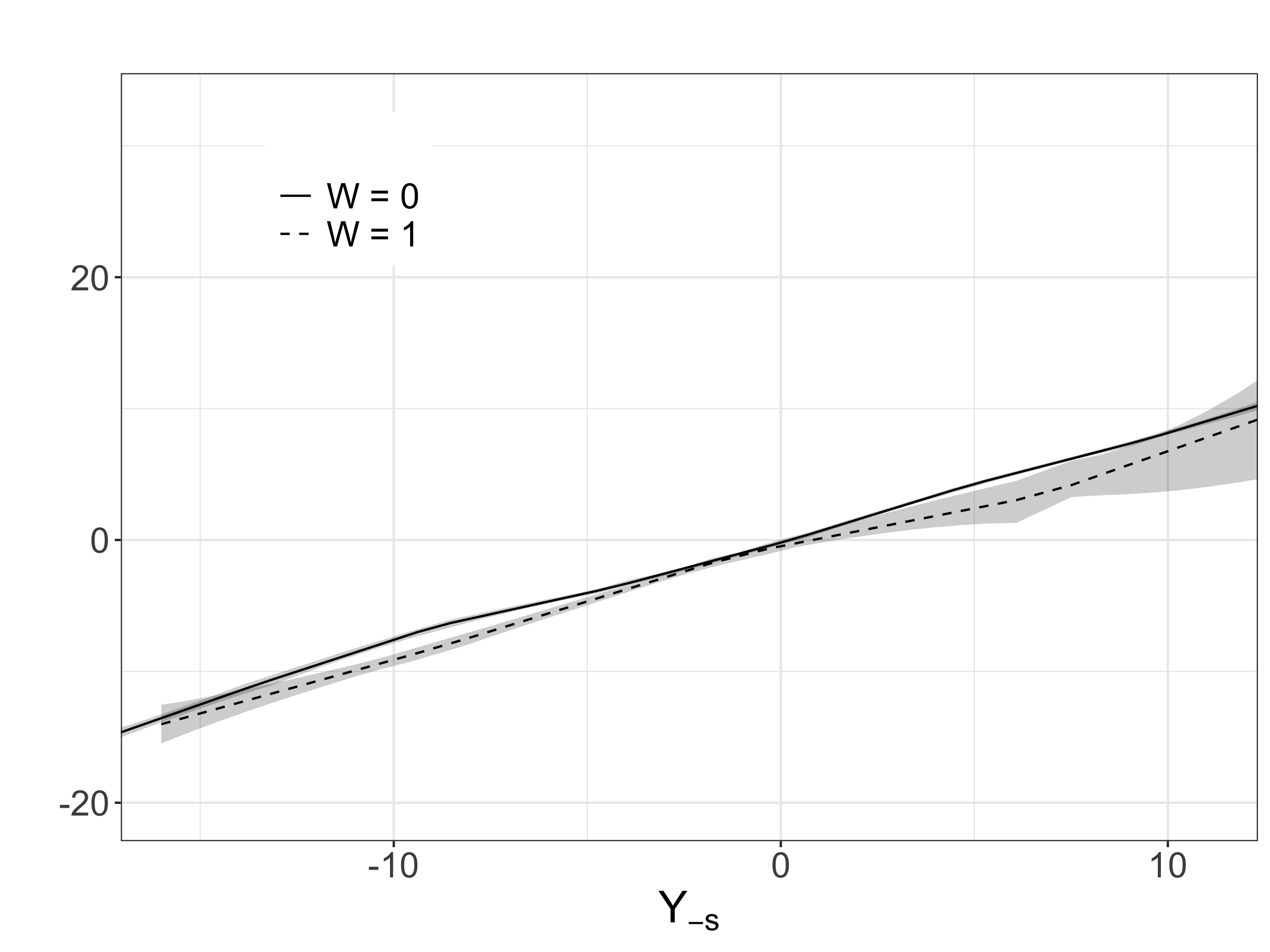}
\\
PSID\\
\includegraphics[width=0.5\textwidth]{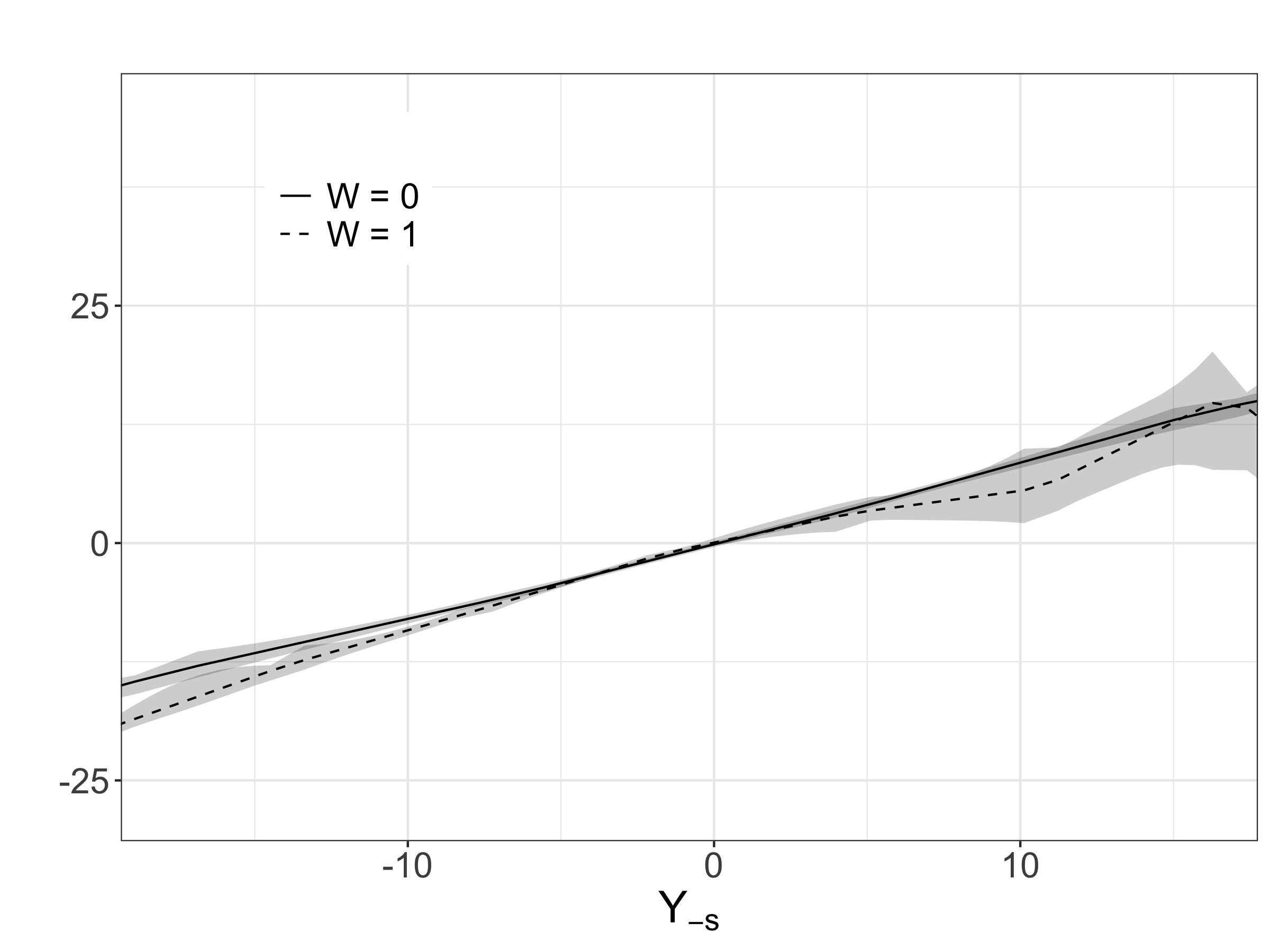}
\caption{Evidence of Assumption \ref{a:special} \eqref{a:special:selection} for the NSW program using the CPS (top) and PSID (bottom) data sets after residualizing with respect to auxiliary covariates. The solid and dashed lines represent estimates of the conditional expectation functions $y \mapsto E[Y_0|W=0,Y_{-s}=y]$ and $y \mapsto E[Y_0|W=1,Y_{-s}=y]$, respectively, computed on the residualized variables. Shaded areas denote 95\% confidence bands. Both axes are measured in thousands of U.S. dollars. For details on the estimation method, refer to Footnote \ref{foot:nonparametric_estimation}.}${}$
\label{fig:nsw:ass1_resid}
\end{figure}

Figure \ref{fig:nsw:ass3_resid} presents the counterparts of Figure \ref{fig:nsw:ass3}, after residualizing with respect to the auxiliary covariates.
Observe that the regression curves remain non-increasing up to sampling uncertainty for both the CPS and PSID data sets, with the hypothesis of monotonicity not refuted by the 95\% confidence bands, providing evidence in support of our Assumption \ref{a:special} \eqref{a:special:dec} even after accounting for the auxiliary covariates.
\begin{figure}[t]
\centering
CPS\\
\includegraphics[width=0.5\textwidth]{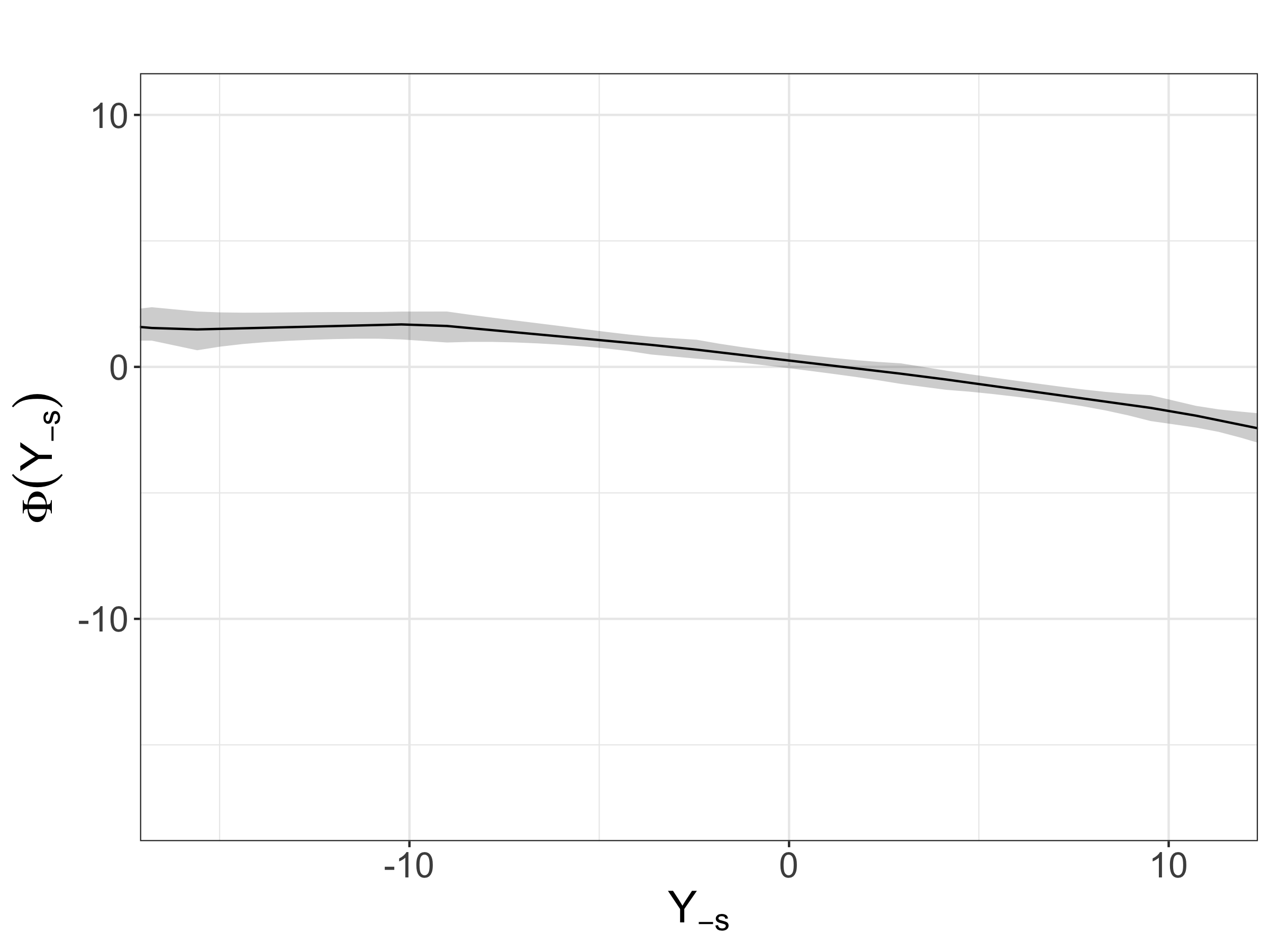}
\\
PSID\\
\includegraphics[width=0.5\textwidth]{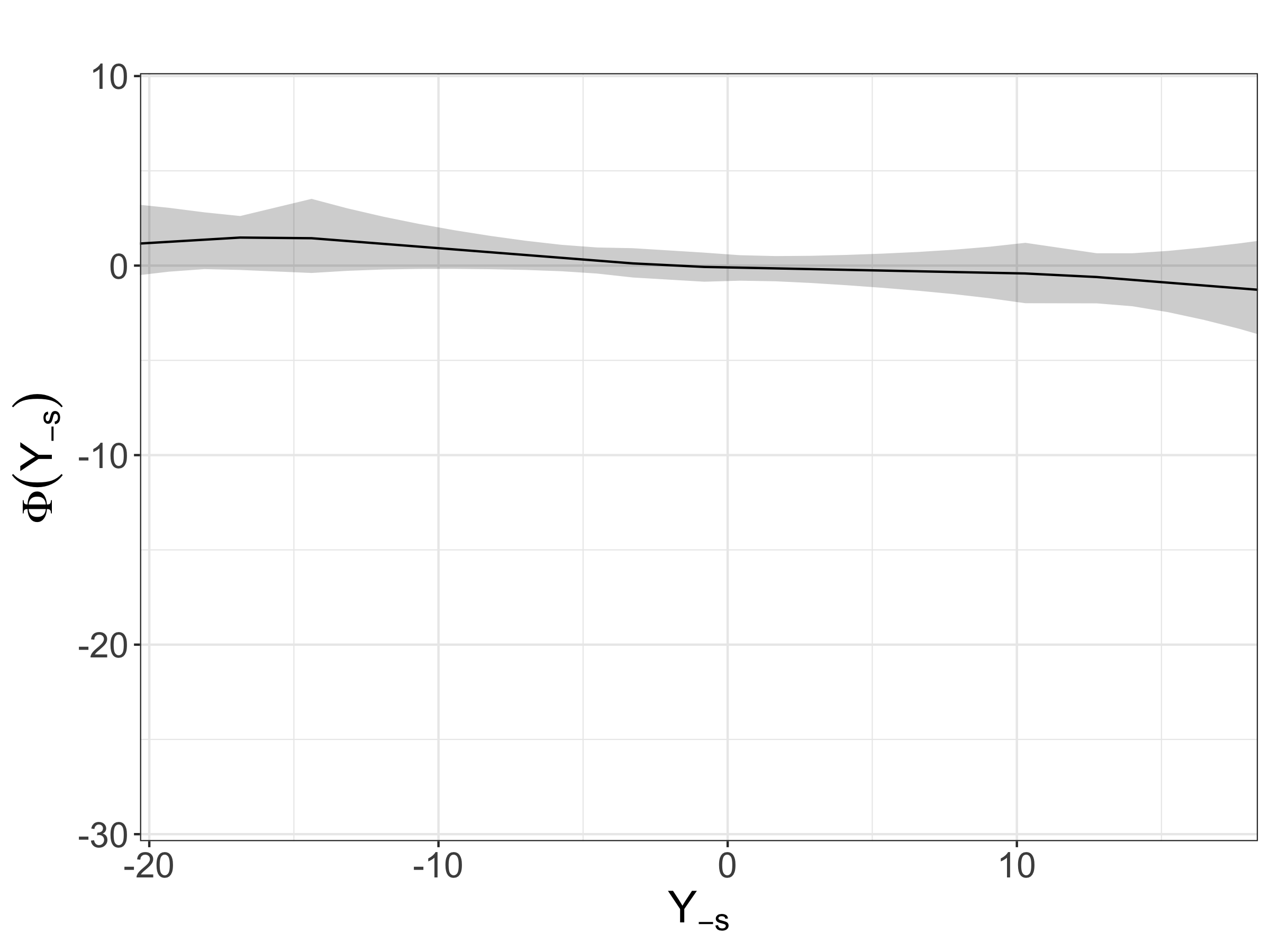}
\caption{Evidence of Assumption \ref{a:special} \eqref{a:special:dec} for the NSW program using the CPS (top) and PSID (bottom) data sets after residualizing with respect to auxiliary covariates. The conditional expectation function $\Phi$ is estimated non-parametrically by the partitioning-based least squares regression on the residualized variables. The estimates, along with their 95\% confidence bands, are plotted. Both the vertical and horizontal axes are measured in thousands of U.S. dollars.}${}$
\label{fig:nsw:ass3_resid}
\end{figure}

\newpage${}$\newpage

\subsubsection{Additional Empirical Analyses for the Education Data}\label{sec:additional:educ}

Figure \ref{fig:educ:ass1_resid} presents the counterparts of Figure \ref{fig:educ:ass1}, after residualizing with respect to the auxiliary covariates.
Observe that the required inequality $E[Y_0|W=0,Y_{-s}=y] \ge E[Y_0|W=1,Y_{-s}=y]$ is still satisfied, providing robust evidence in support of our Assumption \ref{a:special} \eqref{a:special:selection} even after accounting for the auxiliary covariates.
\begin{figure}[t]
\centering
\includegraphics[width=0.5\textwidth]{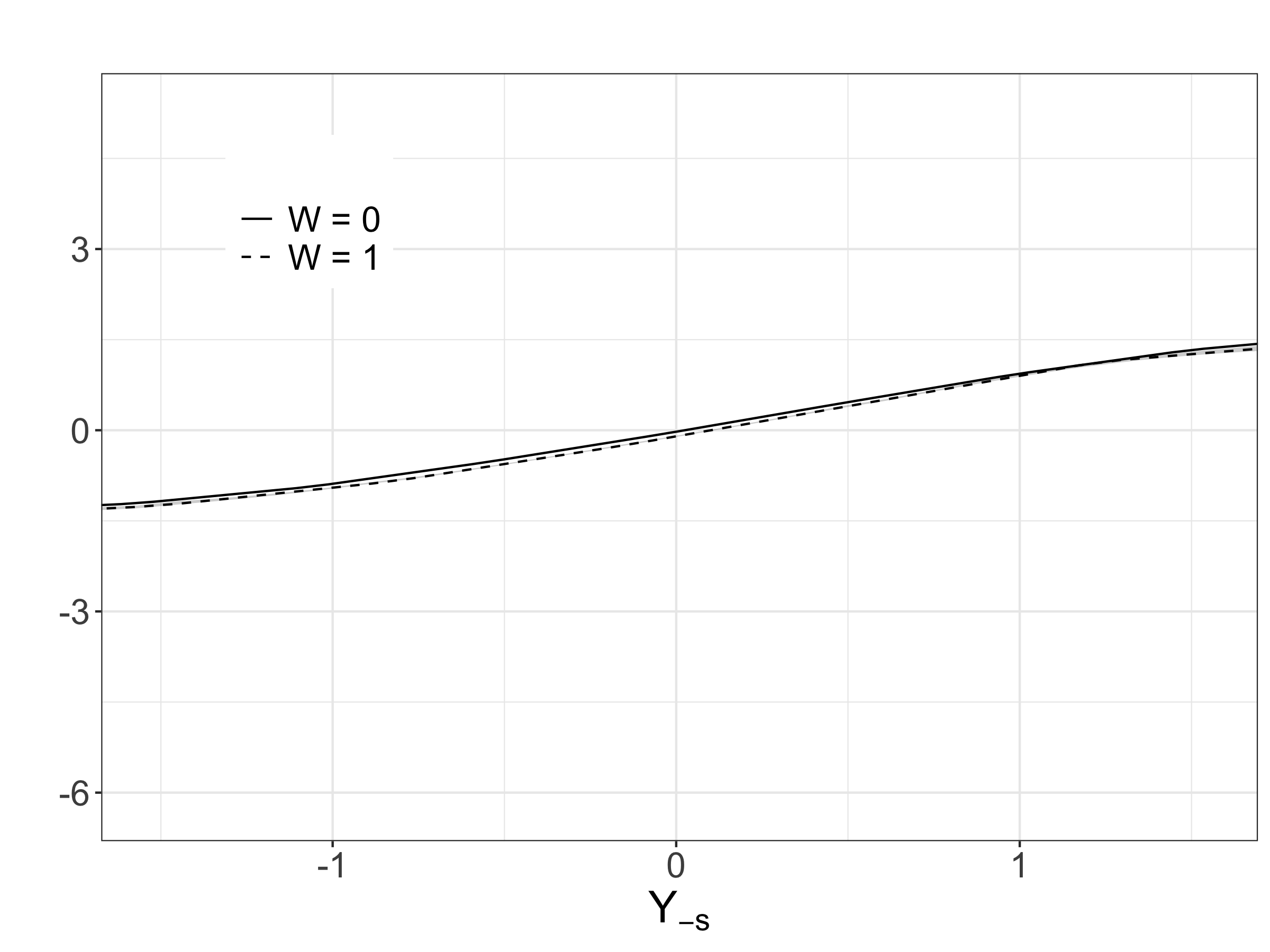}
\caption{Evidence of Assumption \ref{a:special} \eqref{a:special:selection} for the educational program after residualizing with respect to auxiliary covariates. The solid and dashed lines represent estimates of the conditional expectation functions $y \mapsto E[Y_0|W=0,Y_{-s}=y]$ and $y \mapsto E[Y_0|W=1,Y_{-s}=y]$, respectively, computed on the residualized variables. Shaded areas denote 95\% confidence bands. For details on the estimation method, refer to Footnote \ref{foot:nonparametric_estimation}.}${}$
\label{fig:educ:ass1_resid}
\end{figure}

Figure \ref{fig:educ:ass3_resid} presents the counterparts of Figure \ref{fig:educ:ass3}, after residualizing with respect to the auxiliary covariates.
Observe that the regression curves remain non-increasing up to sampling uncertainty, with the hypothesis of monotonicity not refuted by the 95\% confidence bands, providing evidence in support of our Assumption \ref{a:special} \eqref{a:special:dec} even after accounting for the auxiliary covariates.
\begin{figure}[t]
\centering
\includegraphics[width=0.5\textwidth]{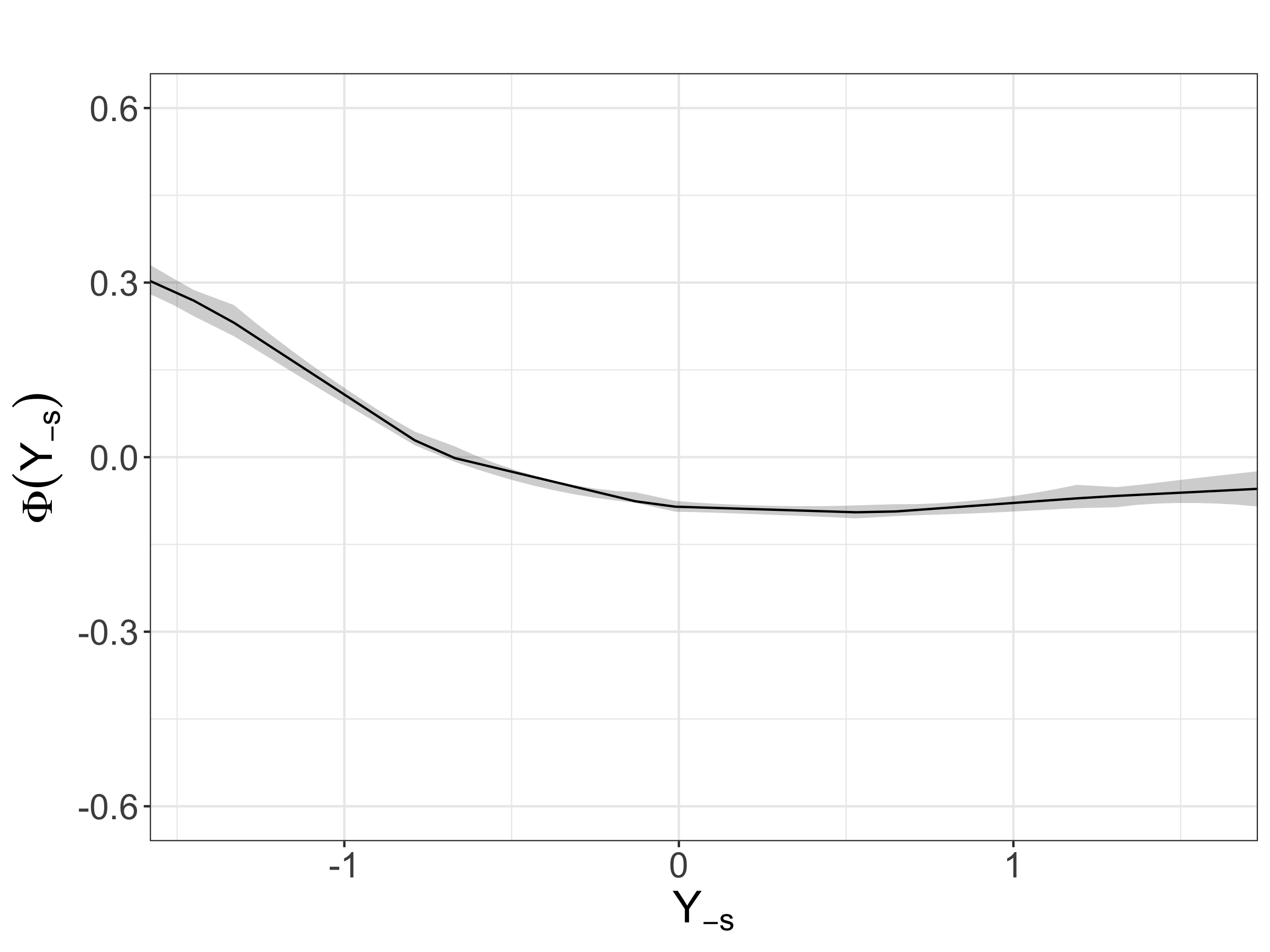}
\caption{Evidence of Assumption \ref{a:special} \eqref{a:special:dec} for the educational program after residualizing with respect to auxiliary covariates. The conditional expectation function $\Phi$ is estimated non-parametrically by the partitioning-based least squares regression on the residualized variables. The estimates, along with their 95\% confidence bands, are plotted.}${}$
\label{fig:educ:ass3_resid}
\end{figure}

\newpage${}$\newpage

\section{Mutual Non-Nestedness of The M, DID, and DIDM Conditions}\label{sec:non_nested}

The three identifying assumptions underlying the respective estimands
$\attm$, $\attdid$, and $\attdidm$ each restrict the DGP in a distinct way, and none implies another.\footnote{Throughout, we consider the general case $s>0$, in which the matching variable $Y_{-s}$ is distinct from the differencing baseline $Y_0$. In the boundary case $s=0$ we have $Y_{-s}=Y_0$, so the $Y_0$ terms cancel in the conditional DIDM contrast and DIDM collapses to M; Conditions M and DIDM then coincide, and the M-versus-DIDM separation below should be read as concerning $s>0$.}
To make this non-nestedness concrete, consider log earnings $Y_t$ as the outcome.  
Denote pre-program earnings by $Y_{-s}$ and let $W\in\{0,1\}$ indicate the treatment status.  
Four stylised DGPs show that every pair of adjacent conditions can be separated. 
A more formal illustration and mathematical derivations are provided in Appendix \ref{sec:detailed_calc_nonnest}.

\begin{itemize}
\item \textbf{Condition M holds but Condition DIDM fails (levels versus trends).}\\
  Within each $Y_{-s}$ cell, untreated earnings \emph{levels} are identical across $W$ \big(i.e., $E[Y_{1}(0)\mid Y_{-s},W]=E[Y_{1}(0)\mid Y_{-s}]$\big), yet the \emph{growth} $Y_{1}(0)-Y_{0}(0)$ is higher for the high-ability trainees who selected into the program.  
  Matching recovers levels, but conditional trends diverge. Hence DIDM fails. See Appendix \ref{sec:nonnested1} for a formal mathematical exposition illustrating this case.

\item \textbf{Condition DIDM holds but Condition M fails (levels versus trends).}\\
  The union that runs the training center negotiates a permanent wage premium $\mu>0$ for enrollees:  
  $E[Y_{t}(0)\mid Y_{-s},W=1]=E[Y_{t}(0)\mid Y_{-s},W=0]+\mu$.  
  Year-on-year earnings growth, however, is identical once we condition on $Y_{-s}$, so conditional parallel trends (DIDM) hold while level-matching (M) fails. See Appendix \ref{sec:nonnested2} for a formal mathematical exposition illustrating this case. 

\item \textbf{Condition DIDM holds but Condition DID fails (aggregation/composition).}\\
  Trainees are drawn disproportionately from the bottom of the earnings distribution.  
  Within each $Y_{-s}$ stratum, untreated growth is the same across $W$ (conditional parallel trend), but low earners naturally catch up faster, so the \emph{average} untreated growth $E[Y_{1}(0)-Y_{0}(0)\mid W]$ is larger in the treated group.  
  Aggregation, therefore, breaks unconditional DID while leaving DIDM intact. See Appendix \ref{sec:nonnested3} for a formal mathematical exposition illustrating this case.

\item \textbf{Condition DID holds but Condition DIDM fails (aggregation/cancellation).}\\
  The trainee pool mixes equal numbers of “fast-growers’’ and “slow-growers’’ whose opposite conditional trends around $Y_{-s}$ exactly cancel in the aggregate.  
  Unconditional trends match \big(DID holds\big), but once we condition on $Y_{-s}$ the opposing slopes re-emerge and DIDM breaks down. See Appendix \ref{sec:nonnested4} for a formal mathematical exposition illustrating this case.
\end{itemize}

\noindent\textbf{Implication for Practice:}
These examples underscore that each assumption addresses a different facet of the DGP. Practitioners must \emph{ex ante} commit to one of these mutually non-nested assumptions based on the context and the nature of available data. Whether one relies on matching (Condition M), unconditional parallel trends (DID), or conditional parallel trends (DIDM) is not a matter of nested robustness but of fundamentally different identifying restrictions, each carrying its own trade-offs for causal inference.
Simply reporting the M, DID, and DIDM estimates side by side and informally declaring “robustness’’ is therefore not a coherent decision rule. In Section~\ref{sec:minimax}, we instead formulate a formal design problem and show that, under additional structure, DIDM is minimax-regret optimal among these three estimands.

\subsection{A Formal Illustration}\label{sec:detailed_calc_nonnest}

\subsubsection{Condition M Holds but Condition DIDM Fails}\label{sec:nonnested1}

Consider the following data-generating process (DGP) for the potential outcomes:
$$
Y_1(0)=\eta_1(0),
\qquad
Y_1(1)=\tau+\eta_1(1),
\qquad
Y_0(0)=\mu W+\eta_0,
$$
where \(\mu\neq 0\), and
$
(\eta_1(0),\eta_1(1),\eta_0) \indep W \mid Y_{-s}.
$
Under this DGP,
$
(Y_1(1),Y_1(0))\indep W\mid Y_{-s},
$
holds.
On the other hand,
$
\E[Y_1(0)-Y_0(0)\mid Y_{-s},W=1]
=
\E[\eta_1(0)-\eta_0\mid Y_{-s},W=1]-\mu
=
\E[Y_1(0)-Y_0(0)\mid Y_{-s},W=0]-\mu
\neq
\E[Y_1(0)-Y_0(0)\mid Y_{-s},W=0].
$
In other words, Condition M holds, but Condition DIDM fails.

\subsubsection{Condition DIDM Holds but Condition M Fails}\label{sec:nonnested2}

Conversely, consider the following DGP for the potential outcomes:
\[
Y_0(0)=g(Y_{-s})+\mu W+\eta_0,
\qquad
Y_1(0)=g(Y_{-s})+\mu W+\delta(Y_{-s})+\eta_1,
\]
where \(\mu\neq 0\) and
$
E[\eta_0\mid Y_{-s},W]=E[\eta_1\mid Y_{-s},W]=0.
$
Under this DGP, Condition DIDM holds since
$
E[Y_1(0)-Y_0(0)\mid Y_{-s},W=1]
=
\delta(Y_{-s})
=
E[Y_1(0)-Y_0(0)\mid Y_{-s},W=0].
$
On the other hand,
$
E[Y_1(0)\mid Y_{-s},W=1]
=
g(Y_{-s})+\mu+\delta(Y_{-s})
=
E[Y_1(0)\mid Y_{-s},W=0] + \mu
\neq
E[Y_1(0)\mid Y_{-s},W=0],
$
implying
$
(Y_1(1),Y_1(0))\not\indep W\mid Y_{-s}.
$
Hence, Condition M fails.

\subsubsection{Condition DIDM Holds but Condition DID Fails}\label{sec:nonnested3}

Consider the following DGP for untreated potential outcomes:
\begin{align*}
Y_0(0) &= \mu(Y_{-s}) + \eta_0(0), & Y_1(0) &= \mu(Y_{-s}) + \eta_1(0),
\end{align*}
where $E[\eta_1(0)-\eta_0(0)|Y_{-s},W]=Y_{-s}$ and $E[Y_{-s}|W=0] \neq E[Y_{-s}|W=1]$.
Under this DGP, we have 
$E[Y_1(0)-Y_0(0)|Y_{-s},W=0]=Y_{-s}=E[Y_1(0)-Y_0(0)|Y_{-s},W=1]$, but
$
E[Y_1(0)-Y_0(0)|W=0] = E[E[\eta_1(0)-\eta_0(0)|Y_{-s},W=0]|W=0] = E[Y_{-s}|W=0]
\neq
E[Y_{-s}|W=1] = E[E[\eta_1(0)-\eta_0(0)|Y_{-s},W=1]|W=1] = E[Y_1(0)-Y_0(0)|W=1].
$
In other words, Condition DIDM holds, but Condition DID fails.   Intuitively, even if the parallel trends hold at every level of \(Y_{-s}\), differences in the distribution of \(Y_{-s}\) between treatment groups can lead to unequal unconditional trends.

\subsubsection{Condition DID Holds but Condition DIDM Fails}\label{sec:nonnested4}

Conversely, consider the following DGP for untreated potential outcomes:
\begin{align*}
Y_0(0) &= \mu(Y_{-s}) + \eta_0(0), & Y_1(0) &= \mu(Y_{-s}) + \eta_1(0),
\end{align*}
where $E[\eta_1(0)-\eta_0(0)|Y_{-s},W] = (2W-1)Y_{-s}$ and $E[Y_{-s}|W=0]+E[Y_{-s}|W=1]=0$.
Under this DGP,
$
E[Y_1(0)-Y_0(0)|W=0] = E[E[Y_1(0)-Y_0(0)|Y_{-s},W=0]|W=0] = E[-Y_{-s}|W=0]
=
E[Y_{-s}|W=1] = E[E[Y_1(0)-Y_0(0)|Y_{-s},W=1]|W=1] = E[Y_1(0)-Y_0(0)|W=1],
$
but
$E[Y_1(0)-Y_0(0)|Y_{-s},W=0] = -Y_{-s} \neq Y_{-s} = E[Y_1(0)-Y_0(0)|Y_{-s},W=1].$
In other words, Condition DID holds but Condition DIDM fails. Intuitively, even if the overall (unconditional) trends are equal (due to cancellation when aggregating over \(Y_{-s}\)), the trends might differ for each subpopulation defined by \(Y_{-s}\).

\bigskip

\section{Extension to General Cases}\label{sec:general}

The double-bracketing result developed in Section~\ref{sec:double_bracketing} extends to a more general class of data-generating processes and estimands. This section presents such an extension, with recent event-study designs as leading examples.

\subsection{Setup}

The previous notations do \textit{not} carry over to the current section. 
Suppose that a researcher is interested in identifying the average treatment effect on the treated (ATT) defined by 
\begin{align}
\att = E\left[\left. \widetilde Y_1(1) - \widetilde Y_1(0) \right|W=1\right].
\label{eq:general:att}
\end{align}
At this moment, we have not introduced the specific meanings of the notations.
They will be discussed in the contexts of specific examples in Section \ref{sec:examples}.
With this said, we want to remark that they parallel with those notations introduced in Section \ref{sec:simple}.
Unlike the previous section, however, the subscripts no longer indicate the time in general.

Similarly to the previous section, we define the alternative estimands
\begin{align}
\attm &= E\left[\left. \widetilde Y_1 \right| W=1\right] - E\left[\left. E\left[\left. \widetilde Y_1 \right|W=0, X \right] \right| W=1\right],
\label{eq:general:m}
\\
\attdid &= E\left[\left. \widetilde Y_1 - \widetilde Y_0 \right| W= 1\right] - E\left[\left. \widetilde Y_1 - \widetilde Y_0 \right| W=0\right],
\qquad\text{and}
\label{eq:general:did}
\\
\attdidm &= E\left[\left. E\left[\left. \widetilde Y_1 - \widetilde Y_0 \right| X, W=1\right] - E\left[\left. \widetilde Y_1 - \widetilde Y_0 \right| X, W=0\right] \right|W=1\right],
\label{eq:general:didm}
\end{align}
called the matching (M), the difference-in-differences (DID), and the difference-in-differences matching (DIDM), respectively.
The $p$-dimensional random vector $X$ is now used as a matching criterion.

The following conditions are imposed:
\begin{align}\label{eq:pretreatment}
\text{Pre-Treatment:}&&    \widetilde Y_0(0)=&\widetilde Y_0 \text{ a.s. given $W \in \{0,1\}$.}
\\
\label{eq:comparison}
\text{Comparison:}&&    \widetilde Y_1(0)=&\widetilde Y_1 \text{ a.s. given } W=0.
\end{align}
Condition \eqref{eq:pretreatment} requires that observed outcomes be the potential outcome without treatment for every unit prior to treatment.
Condition \eqref{eq:comparison} requires that the observed outcome be the potential outcome without treatment for the control group.

\subsection{Examples of the M, DID, and DIDM Estimands in Event Studies}\label{sec:examples}

In this section, we demonstrate that our general framework \eqref{eq:general:att}--\eqref{eq:comparison} encompasses alternative estimands studied in the literature of event studies as examples.
\subsubsection{Example 1: M in Event Studies}\label{sec:event_m}

\citet{acemoglu2019democracy} consider the ATT
\begin{equation}\label{eq:acemoglu:att}
E\left[ Y_t^s(1) - Y_t^s(0) | D_t=1, D_{t-1}=0\right],
\end{equation}
where $D_t$ denotes the indicator of democracy, $Y_t^s(1)$ denotes the potential GDP in period $t+s$ when a country is treated between periods $t-1$ and $t$ (i.e., $D_t=1$ and $D_{t-1}=0$), and $Y_t^s(0)$ denotes the potential GDP in period $t+s$ when such a treatment does not occur (i.e., $D_t=D_{t-1}=0$).\footnote{The original paper by \citet{acemoglu2019democracy} considers
$
E\left[ (Y_t^s(1)-Y_{t-1}) - (Y_t^s(0)-Y_{t-1}) | D_t=1, D_{t-1}=0\right]
$
as the parameter of interest, 
where $Y_{t-1}$ denotes the realized GDP in period $t-1$,
but this is equivalent to \eqref{eq:acemoglu:att}.}
\citet{acemoglu2019democracy} identify this ATT by
\begin{equation}\label{eq:acemoglu:estimand_pre}
E\left[\left. Y_t^s-Y_{t-1} \right| D_t=1, D_{t-1}=0 \right]
-
E\left[\left. E\left[\left. Y_t^s-Y_{t-1} \right| D_t=0, D_{t-1}=0, X \right] \right| D_t=1, D_{t-1} = 0 \right],
\end{equation}
where
$Y_t^s$ denotes the observed GDP in period $t+s$,
$Y_{t-1}$ denotes the observed GDP in period $t-1$, and
$X := (Y_{t-1},\ldots,Y_{t-4})'$
in their baseline model with additional covariates in extended robustness analyses.

Since $X$ contains $Y_{t-1}$ in particular, the conditioning theorem\footnote{Specifically, the conditioning theorem yields $E[Y_{t-1}|D_t=0,D_{t-1}=0,X]=Y_{t-1}$ when $X$ contains $Y_{t-1}$.} cancels $Y_{t-1}$ between the two terms in \eqref{eq:acemoglu:estimand_pre}, so the identifying formula \eqref{eq:acemoglu:estimand_pre} of \citet{acemoglu2019democracy} boils down to
\begin{equation}\label{eq:acemoglu:estimand}
E\left[\left. Y_t^s \right| D_t=1, D_{t-1}=0 \right]
-
E\left[\left. E\left[\left. Y_t^s \right| D_t=0, D_{t-1}=0, X \right] \right| D_t=1, D_{t-1} = 0 \right].
\end{equation}

Our general framework encompasses this example.
Specifically, our ATT \eqref{eq:general:att} reduces to \eqref{eq:acemoglu:att} and our M estimand \eqref{eq:general:m} reduces to \eqref{eq:acemoglu:estimand} by setting
\begin{align*}
\widetilde Y_0(0) :=& Y_{t-1},
&
\widetilde Y_0(1) :=& Y_{t-1},
&
\widetilde Y_0 :=& Y_{t-1},
\\
\widetilde Y_1(0) :=& Y_t^s(0),
&
\widetilde Y_1(1) :=& Y_t^s(1),
&
\widetilde Y_1 :=& Y_t^s,
\end{align*}
\begin{align*}
\text{and} \qquad
W := 
\begin{cases}
1 & \text{if } D_t=1 \text{ and } D_{t-1}=0
\\
0 & \text{if } D_t=0 \text{ and } D_{t-1}=0
\\
-1 & \text{otherwise}
\end{cases}
\end{align*}
The pre-treatment condition \eqref{eq:pretreatment} is satisfied by construction, as $\widetilde Y_0(0) = Y_{t-1} = \widetilde Y_0$.
The comparison condition \eqref{eq:comparison} is also satisfied by construction via the definition of $Y_t^s(0)$ as the potential outcome under $D_t=D_{t-1}=0$.
Namely, $\widetilde Y_1(0) = Y_t^s(0) = Y_t^s = \widetilde Y_1$ holds given $D_t=D_{t-1}=0$.

\subsubsection{Example 2: DID in Event Studies}\label{sec:event_did}
\citet{callaway2018difference} consider the ATT
\begin{equation}\label{eq:callaway:att}
E\left[\left. Y_t(g) - Y_t(\infty) \right| G=g \right],
\end{equation}
where $G$ denotes the treatment period, $Y_t(g)$ denotes the potential outcome at period $t \geq g$ when an individual is treated at period $g$, and $Y_t(\infty)$ denotes the potential outcome at period $t$ when an individual does not receive a treatment.
\citet{callaway2018difference} identify this ATT by
\begin{equation}\label{eq:callaway:estimand}
E\left[\left.Y_t - Y_{g-1} \right| G=g \right]
-
E\left[\left.Y_t - Y_{g-1} \right| G=g' \right]
\end{equation}
for $g' \geq t+1$,
where $Y_t$ denotes the observed outcome at period $t$.
(The second term of \eqref{eq:callaway:estimand} may be aggregated over $G' \in \{t+1,t+2,\ldots\}$.)

Our general framework encompasses this example.
Specifically, our ATT \eqref{eq:general:att} reduces to \eqref{eq:callaway:att} and our DID estimand \eqref{eq:general:did} reduces to \eqref{eq:callaway:estimand} by setting
\begin{align*}
\widetilde Y_0(0) :=& Y_{g-1}(\infty),
&
\widetilde Y_0(1) :=& Y_{g-1}(g),
&
\widetilde Y_0 :=& Y_{g-1},
\\
\widetilde Y_1(0) :=& Y_t(\infty),
&
\widetilde Y_1(1) :=& Y_t(g),
&
\widetilde Y_1 :=& Y_t,
\end{align*}
\begin{align*}
\text{and} \qquad
W := 
\begin{cases}
1 & \text{if } G=g
\\
0 & \text{if } G=g'
\\
-1 & \text{otherwise}
\end{cases}
\end{align*}
The pre-treatment condition \eqref{eq:pretreatment} is satisfied by construction, as $\widetilde Y_0(0) = Y_{g-1}(\infty) = Y_{g-1} = \widetilde Y_0$ given $G=g$ or $G=g' \geq t+1 > g$.
The comparison condition \eqref{eq:comparison} is also satisfied by construction, as $\widetilde Y_1(0) = Y_t(\infty) = Y_t = \widetilde Y_1$ given $G=g' \geq t+1$.

\subsubsection{Example 3: DIDM in Event Studies}\label{sec:event_didm}
\citet[][Section 4.1]{dube2023local} consider the ATT
\begin{equation}\label{eq:dube:att}
E\left[ Y_{t+h}(1) - Y_{t+h}(0) | \Delta D_t=1\right],
\end{equation}
where $\Delta D_t$ denotes the indicator of policy change, $Y_{t+h}(1)$ denotes the potential outcome in period $t+h$ when a policy changes between periods $t-1$ and $t$ (i.e., $\Delta _t=1$), and $\Delta Y_{t+h}(0)$ denotes the potential outcome in period $t+h$ when such a change does not occur (i.e., $\Delta D_t=0$).
\citet[][Section 4.1]{dube2023local} identify this ATT by
\begin{equation}\label{eq:dube:estimand}
E\left[\left. Y_{t+h}-Y_{t-1} \right| \Delta D_t=1 \right]
-
E\left[\left. E\left[\left. Y_{t+h}-Y_{t-1} \right| \Delta D_t=0, X \right] \right| \Delta D_t=1 \right],
\end{equation}
where
$Y_{t+h}$ denotes the observed outcome in period $t+h$,
$Y_{t-1}$ denotes the observed outcome in period $t-1$, and
$X$ is a vector of general covariates.

Our general framework encompasses this example.
Specifically, our ATT \eqref{eq:general:att} reduces to \eqref{eq:dube:att} and our DIDM estimand \eqref{eq:general:didm} reduces to \eqref{eq:dube:estimand} by setting
\begin{align*}
\widetilde Y_0(0) :=& Y_{t-1},
&
\widetilde Y_0(1) :=& Y_{t-1},
&
\widetilde Y_0 :=& Y_{t-1},
\\
\widetilde Y_1(0) :=& Y_{t+h}(0),
&
\widetilde Y_1(1) :=& Y_{t+h}(1),
&
\widetilde Y_1 :=& Y_{t+h},
\end{align*}
\begin{align*}
\text{and} \qquad
W := 
\begin{cases}
1 & \text{if } \Delta D_t=1
\\
0 & \text{if } \Delta D_t=0
\\
-1 & \text{otherwise}
\end{cases}
\end{align*}
The pre-treatment condition \eqref{eq:pretreatment} is satisfied by construction, as $\widetilde Y_0(0) = Y_{t-1} = \widetilde Y_0$.
The comparison condition \eqref{eq:comparison} is also satisfied by construction via the definition of $Y_{t+h}(0)$ as the potential outcome under $\Delta D_t=0$.
Namely, $\widetilde Y_1(0) = Y_{t+h}(0) = Y_{t+h} = \widetilde Y_1$ holds given $\Delta D_t=0$.

Also see the DID${}_\text{M}$ estimator of \citet{deChaisemartin2020two}, and the (panel) matching estimator of \citet{imai2023matching}, as well as the extended DID method of \citet[][Section 4.1]{dube2023local} -- they all propose and analyze the properties of what we refer to as the DIDM.

As pointed out by \citet{dube2023local}, their framework encompasses \citet{acemoglu2019democracy} as a special case. Indeed, when $X$ contains $\widetilde Y_0$, as is the case with \citet{acemoglu2019democracy} presented in Section \ref{sec:event_m}, our DIDM framework reduces to our M framework.
In general, however, the DIDM differs from the M.

\subsubsection{Summary and Discussions of the Three Examples}
Albeit there are slight differences in their notations, the three examples presented above focus on similar setups.
They fundamentally differ only in terms of the estimands: the three examples focus on the M, DID, and DIDM estimands in our language.
In general, a researcher does not know which of them achieves the identification.
The M estimand identifies the true ATT (i.e, $\attm = \att$ holds) if the matching condition
\begin{align*}
\text{Condition M:} \ \ \ \left.\left(\widetilde Y_1(1), \widetilde Y_1(0) \right) \indep W \right| X
\end{align*}
is satisfied.
The DID estimand identifies the true ATT (i.e, $\attdid = \att$ holds) if the parallel trend condition
\begin{align*}
\text{Condition DID:} \ \ \
\E\left[\left.\widetilde Y_1(0)-\widetilde Y_0(0)\right|W=0\right]
=
\E\left[\left.\widetilde Y_1(0)-\widetilde Y_0(0)\right|W=1\right]
\end{align*}
is satisfied.
Finally, the DIDM estimand identifies the true ATT (i.e, $\attdidm = \att$ holds) if the conditional parallel trend condition
\begin{align*}
\text{Condition DIDM:} \ \ \
\E\left[\left.\widetilde Y_1(0)-\widetilde Y_0(0)\right|X,W=0\right]
=
\E\left[\left.\widetilde Y_1(0)-\widetilde Y_0(0)\right|X,W=1\right]
\end{align*}
is satisfied.
In the absence of knowledge of the underlying data-generating process, however, committing to a wrong assumption can lead to biased estimates by M, DID, or DIDM.
It is therefore of interest to characterize the relation among the three estimands.
The following subsection investigates this point.

\subsection{The General Double Bracketing Result}
Now, focus on the generic framework \eqref{eq:general:att}--\eqref{eq:comparison} again.
Let
\begin{align*}
\Delta(\attm) =& \attm-\att,
\\
\Delta(\attdid) =& \attdid-\att, \qquad\text{and}
\\
\Delta(\attdidm) =& \attdidm-\att
\end{align*}
be the identification errors of the estimands, $\attm$, $\attdid$, and $\attdidm$, respectively.
We establish the double bracketing relation 
$\Delta(\attm) \leq \Delta(\attdidm) \leq \Delta(\attdid)$ under the following assumption.

\begin{assumption}\label{a:general}The following conditions hold.
\begin{enumerate}[(i)]
\item\label{a:general:selection}
$E\left[\left.\widetilde Y_0 \right|W=0,X=x\right] \geq E\left[\left. \widetilde Y_0 \right|W=1,X=x\right]$ for all $x$.
\item \label{a:general:sd}
$F_{X|W=0}$
multivariate first-order stochastically dominates
$F_{X|W=1}$.\footnote{We say that $X$ multivariate first-order stochastically dominates $X^\ast$ if $E[f(X)] \geq E[f(X^\ast)]$ for every bounded, coordinatewise nondecreasing function $f$; see \citet[\S6.B]{shaked2007stochastic}.}
\item\label{a:general:dec}
$x \mapsto \Phi(x) := E\left[\left.\widetilde Y_1-\widetilde Y_0\right|W=0,X=x\right]$ is weakly decreasing.\footnote{We say that $\Phi$ is weakly decreasing if $\Phi(x_1,\ldots,x_p) \geq \Phi(x^\ast_1,\ldots,x^\ast_p)$ holds whenever $x_1 \leq x^\ast_1$, $\ldots$ and, $x_p \leq x^\ast_p$.}
\end{enumerate}
\end{assumption}
The three parts \eqref{a:general:selection}--\eqref{a:general:dec} of this assumption parallel those in Assumption \ref{a:special}, albeit that $X$ is now possibly multi-dimensional.
Hence, similar interpretations can be made especially when $X$ consists of lagged outcomes as in the first example \textit{a la} \citet{acemoglu2019democracy} presented in Section \ref{sec:event_m}.
Such a convenient interpretation may not be feasible if $X$ contains other covariates, but we want to stress that each of the three conditions \eqref{a:general:selection}--\eqref{a:general:dec} of this assumption is still empirically testable.

The following theorem states the extended double bracketing result for the general cases.

\begin{theorem}\label{theorem:general}
Suppose that Assumption \ref{a:general} holds for \eqref{eq:general:att}--\eqref{eq:comparison}.
Then, we have the bracketing relationship
\begin{align*}
\Delta(\attm) \leq \Delta(\attdidm) \leq \Delta(\attdid).
\end{align*}
\end{theorem}

\begin{proof}[Proof of Theorem \ref{theorem:general}  ]
Note that the identification errors can be written as
\begin{align*}
\Delta\left(\attm\right)= & E\left[E\left[\widetilde Y_1(0) \mid W=1, X\right]-E\left[\widetilde Y_1(0) \mid W=0, X\right] \mid W=1\right],
\\
\Delta\left(\attdid\right)= & E\left[\widetilde Y_1(0)-\widetilde Y_0(0) \mid W=1\right]-E\left[\widetilde Y_1(0)-\widetilde Y_0(0) \mid W=0\right],
\qquad\text{and}
\\
\Delta\left(\attdidm\right)= & E\left[E\left[\widetilde Y_1(0)-\widetilde Y_0(0) \mid W=1, X\right]-E\left[\widetilde Y_1(0)-\widetilde Y_0(0) \mid W=0, X\right] \mid W=1\right].
\end{align*}

First, observe that
\begin{align*}
\Delta(\attdidm) - \Delta(\attm)
=&
E\left[E\left[\widetilde Y_0(0)|W=0,X\right] 
- E\left[\widetilde Y_0(0)|W=1,X\right]|W=1\right]
\\
=&
E\left[E\left[\widetilde Y_0|W=0,X\right] 
- E\left[\widetilde Y_0|W=1,X\right]|W=1\right]
\geq 0
\end{align*}
where 
the second equality is due to \eqref{eq:pretreatment}, and
the last inequality follows from Assumption \ref{a:general} \eqref{a:general:selection}.

Second, observe that
\begin{align*}
&\Delta(\attdid) - \Delta(\attdidm)
\\
=&
E\left[E\left[\widetilde Y_1(0)-\widetilde Y_0(0)|W=0,X\right] | W=1\right]
-
E\left[\widetilde Y_1(0)-\widetilde Y_0(0)|W=0\right]
\\
=&
E\left[E\left[\widetilde Y_1(0)-\widetilde Y_0(0)|W=0,X\right] | W=1\right]
-
E\left[E\left[\widetilde Y_1(0)-\widetilde Y_0(0)|W=0,X\right] | W=0\right]
\\
=&
E\left[E\left[\widetilde Y_1-\widetilde Y_0|W=0,X\right] | W=1\right]
-
E\left[E\left[\widetilde Y_1-\widetilde Y_0|W=0,X\right] | W=0\right]
\\
=& 
E\left[\Phi(X) | W=1\right]
-
E\left[\Phi(X) | W=0\right]
\geq 0,
\end{align*}
where 
the second equality follows from the law of iterated expectations,
the third equality is due to \eqref{eq:pretreatment}--\eqref{eq:comparison}, and
the last inequality follows from Assumption \ref{a:general} \eqref{a:general:sd}--\eqref{a:general:dec}.
\end{proof}

\section{Details of the Calibrated Simulation}
\label{app:common_rf_details}

This appendix gives the details behind the calibration in
Section~\ref{sec:simulation_common_rf}. The five-parameter DGP is held fixed across the
three worlds, which differ only by the zero restriction defining each world. The data enter
through the common nonzero coefficients and the subgroup treatment-effect schedules. The
design should be read as a stylized calibration to the NSW+CPS and NSW+PSID moments.

\subsection{Calibration Moments and Objective}

The calibration matches a vector of reduced-form moments. As in the main text, the
lagged-outcome index is
\[
L=2\cdot\mathbf{1}\{Y_{-s}>\operatorname{med}(Y_{-s})\}-1\in\{-1,1\},
\]
a balanced split at the sample median, and the $M$ and DIDM estimators condition on this
same index in the simulation. For each world $k$, we recode the real data using the
subgroup treatment-effect schedule $\widehat\tau_k(g)$, removing the treatment effect from
the observed post-period outcome to obtain an estimated untreated counterfactual:
\[
\widehat Y_{1,k}(0)=Y_1-W\widehat\tau_k(g),
\qquad
\widehat\Delta_k(0)=\widehat Y_{1,k}(0)-Y_0,
\]
where $\widehat\Delta_k(0)$ is the corresponding estimated untreated trend.

The calibration targets the moments in Table~\ref{tab:moment_calibration_moments}, which map
directly into the five structural coefficients. The scale anchors pin down the magnitudes of
$p$, $q$, and $m$; the treatment-share and treated--control gap moments pin down the
selection parameters $a$ and $c$, together with the interaction between selection and the
hidden outcome channels.

\begin{table}[H]
\centering
\scriptsize
\caption{Calibration Moments}
\label{tab:moment_calibration_moments}
\begin{threeparttable}
\begin{tabular}{p{0.18\textwidth}p{0.34\textwidth}p{0.38\textwidth}}
\toprule
Group & Moments & Description \\
\midrule
Scale anchors & $p_k^*,q_k^*,m_k^*$ repeated five times & Direct anchors for the trend-in-$L$, hidden-trend, and hidden-level coefficients \\
Treatment shares & $\Pr(W=1\mid L=1)$, $\Pr(W=1\mid L=-1)$ & Marginal treatment selection by lagged-outcome index \\
Level gaps & $E[\widehat Y_{1,k}(0)\mid W=1,L=\ell]-E[\widehat Y_{1,k}(0)\mid W=0,L=\ell]$ & Hidden post-period level imbalance, for $\ell\in\{-1,1\}$ \\
Trend gaps & $E[\widehat\Delta_k(0)\mid W=1,L=\ell]-E[\widehat\Delta_k(0)\mid W=0,L=\ell]$ & Hidden untreated trend imbalance, for $\ell\in\{-1,1\}$ \\
\bottomrule
\end{tabular}
\begin{tablenotes}[flushleft]
\footnotesize
\item The scale anchors are
$p_k^*=\operatorname{sign}\{E[\widehat\Delta_k(0)\mid L=1]-E[\widehat\Delta_k(0)\mid L=-1]\}\cdot
\operatorname{sd}\{\widehat\Delta_k(0)\}$,
$q_k^*=\operatorname{sd}\{\widehat\Delta_k(0)\}$, and
$m_k^*=\operatorname{sd}\{\widehat Y_{1,k}(0)\}$, where $\operatorname{sd}\{\cdot\}$ denotes
the sample standard deviation. Repeating the three scale anchors five times each prevents the
calibration from shrinking the economically important hidden-level and hidden-trend channels
too aggressively.
\end{tablenotes}
\end{threeparttable}
\end{table}

For a candidate $\theta$, the implied treatment-share moments are
\[
\pi_\ell(a,c)
=
\tfrac{1}{2}\Lambda(a\ell+c\ell)
+\tfrac{1}{2}\Lambda(a\ell-c\ell),
\qquad \ell\in\{-1,1\},
\]
and, writing
\[
s_\ell(a,c)
=
E[U\mid W=1,L=\ell]
-
E[U\mid W=0,L=\ell]
\]
for the confounder imbalance within each $L$ cell, the implied level and trend gaps are
$m\,s_\ell(a,c)$ and $q\,s_\ell(a,c)$, respectively. The implied calibration vector contains
$(p,q,m)$ repeated five times, followed by $\pi_1,\pi_{-1}$, then $m s_1, m s_{-1}$, and
$q s_1, q s_{-1}$.

For each comparison sample, the calibration chooses a common $\theta=(a,c,p,q,m)$ and
evaluates the three zero-restricted vectors $\theta_M$, $\theta_{DID}$, and $\theta_{DIDM}$.
Let $\widehat\gamma_k$ be the target vector for world $k$ and $\gamma(\theta_k)$ the implied
vector. We minimize the scale-normalized objective
\[
\widehat Q(\theta)
=
\frac{1}{3}
\sum_{k\in\{M,DID,DIDM\}}
\sum_j
\left(
\frac{\gamma_j(\theta_k)-\widehat\gamma_{k,j}}
{\max\{|\widehat\gamma_{k,j}|,1\}}
\right)^2
+\mathcal P(\theta),
\]
where the denominator $\max\{|\widehat\gamma_{k,j}|,1\}$ normalizes each moment by its own
scale, and $\mathcal P(\theta)$ is a penalty that rules out degenerate worlds by keeping each
non-target identifying restriction away from zero while maintaining overlap in treatment
assignment. The calibration then imposes the world-specific zero restrictions shown in
Table~\ref{tab:moment_calibrated_coefficients}.

\subsection{Direct Verification of the Identifying Restrictions}

Because the simulation stores the untreated potential outcomes, we can check each
identifying restriction directly, rather than only through the moments used in calibration.
For each world we compute the standardized violation of each restriction:
\begin{align*}
\delta_M
=&
\frac{
\max_{\ell\in\{-1,1\}}
\bigl|
\E[Y_1(0)\mid W=1,L=\ell]
-
\E[Y_1(0)\mid W=0,L=\ell]
\bigr|
}{\operatorname{sd}\{Y_1(0)\}},
\\
\delta_{DID}
=&
\frac{
\bigl|
\E[\Delta(0)\mid W=1]
-
\E[\Delta(0)\mid W=0]
\bigr|
}{\operatorname{sd}\{\Delta(0)\}},
\qquad\text{and}
\\
\delta_{DIDM}
=&
\frac{
\max_{\ell\in\{-1,1\}}
\bigl|
\E[\Delta(0)\mid W=1,L=\ell]
-
\E[\Delta(0)\mid W=0,L=\ell]
\bigr|
}{\operatorname{sd}\{\Delta(0)\}}.
\end{align*}
Each $\delta$ measures, on a standard-deviation scale, how far the corresponding restriction
is from holding: $\delta_M$ is the conditional level imbalance that matching rules out,
$\delta_{DID}$ the unconditional trend imbalance that DID rules out, and $\delta_{DIDM}$ the
conditional trend imbalance that DIDM rules out. A world ``verifies'' its intended
restriction when the corresponding $\delta$ falls below the assumption-specific tolerance
reported in the table note. Table~\ref{tab:moment_assumption_verification} shows that, in
both comparison samples, each world satisfies exactly the restriction it is designed to
satisfy, and violates the other two.

\begin{table}[H]
\centering
\scriptsize
\caption{Direct Verification of the Identifying Restrictions}
\label{tab:moment_assumption_verification}
\begin{threeparttable}
\begin{tabular}{llcccc}
\toprule
Sample & World & $\delta_M$ & $\delta_{DID}$ & $\delta_{DIDM}$ & Verified restriction \\
\midrule
NSW + CPS  & $M$-world    & 0.000 & 0.225 & 0.253 & $M$ \\
NSW + CPS  & DID-world    & 0.466 & 0.004 & 0.281 & DID \\
NSW + CPS  & DIDM-world   & 0.424 & 0.222 & 0.011 & DIDM \\
\midrule
NSW + PSID & $M$-world    & 0.007 & 0.156 & 0.195 & $M$ \\
NSW + PSID & DID-world    & 0.409 & 0.006 & 0.201 & DID \\
NSW + PSID & DIDM-world   & 0.393 & 0.163 & 0.013 & DIDM \\
\bottomrule
\end{tabular}
\begin{tablenotes}[flushleft]
\footnotesize
\item A world verifies a restriction when its standardized gap falls below the
assumption-specific tolerance: $0.025$ for $M$, $0.020$ for DID, and $0.028$ for DIDM. In
each world, the gap for the intended restriction (in the diagonal cells) lies below
tolerance, while the gaps for the other two restrictions are an order of magnitude larger.
\end{tablenotes}
\end{threeparttable}
\end{table}

\subsection{Distinguishability Diagnostics}

The calibration creates three worlds that differ in their identifying restrictions while
remaining similar on reduced-form features a researcher might inspect before choosing a
design. Table~\ref{tab:moment_observational_similarity} reports several such features. The
lagged-outcome distribution, treatment share, baseline outcome level, and covariance between
the lagged outcome and $Y_0$ are close across worlds within each comparison sample.
Post-period levels differ more, as expected, since the worlds deliberately encode different
hidden level and trend channels.

\begin{table}[H]
\centering
\scriptsize
\caption{Low-Dimensional Reduced-Form Similarity across Simulated Worlds}
\label{tab:moment_observational_similarity}
\begin{threeparttable}
\begin{tabular}{llrrrrrr}
\toprule
Sample & World & $E[W]$ & $E[Y_{-s}]$ & $E[Y_0]$ & $E[Y_1]$ & $\Cov(Y_{-s},Y_0)$ & $\Cov(Y_{-s},Y_1)$ \\
\midrule
NSW + CPS  & $M$    & 0.510 & 13.522 & 13.147 & 12.368 & 79.671 & 87.559 \\
NSW + CPS  & DID    & 0.496 & 13.513 & 13.146 & 14.477 & 78.912 & 96.202 \\
NSW + CPS  & DIDM   & 0.505 & 13.680 & 13.263 & 14.293 & 79.709 & 97.887 \\
\midrule
NSW + PSID & $M$    & 0.498 & 14.825 & 14.528 & 12.713 & 128.171 & 134.789 \\
NSW + PSID & DID    & 0.498 & 14.989 & 14.682 & 15.429 & 130.352 & 150.720 \\
NSW + PSID & DIDM   & 0.506 & 15.169 & 14.810 & 15.289 & 129.373 & 148.277 \\
\bottomrule
\end{tabular}
\begin{tablenotes}[flushleft]
\footnotesize
\item Earnings are in thousands of dollars. The table reports reduced-form summaries an
applied researcher could inspect ex ante; it is only a low-dimensional check of similarity.
\end{tablenotes}
\end{threeparttable}
\end{table}

As a more demanding check, we train a held-out random-forest classifier to distinguish the
three simulated worlds from one another, since the classifier can combine many features
nonlinearly. Table~\ref{tab:moment_distinguishability} reports the resulting three-way
classification accuracy. If the worlds were indistinguishable from these features, accuracy
would be near the chance level of $1/3$; the observed accuracies, $0.367$ and $0.356$, are
close to that benchmark. The worlds therefore satisfy distinct identifying restrictions
while remaining hard to tell apart from the data a researcher would observe, which is what
makes the regret comparison in Section~\ref{sec:simulation_common_rf} the relevant one.

\begin{table}[H]
\centering
\scriptsize
\caption{Three-Way Distinguishability across Simulated Worlds}
\label{tab:moment_distinguishability}
\begin{threeparttable}
\begin{tabular}{lccc}
\toprule
Sample & Chance benchmark & 3-way accuracy & Classification error \\
\midrule
NSW + CPS  & 0.333 & 0.367 & 0.633 \\
NSW + PSID & 0.333 & 0.356 & 0.644 \\
\bottomrule
\end{tabular}
\begin{tablenotes}[flushleft]
\footnotesize
\item Chance accuracy for a three-way classification is $1/3$. Accuracy near this benchmark
indicates the three worlds are difficult to distinguish from the listed features before the
regret criterion is applied.
\end{tablenotes}
\end{threeparttable}
\end{table}

\end{document}